%% file: main_CVPR.tex
\documentclass[10pt,twocolumn,letterpaper]{article}

\usepackage[pagenumbers]{cvpr} 


\usepackage{microtype}
\usepackage{graphicx}
\usepackage{booktabs}

\usepackage{amsmath}
\usepackage{amssymb}
\usepackage{amsthm}
\usepackage{multirow}
\usepackage{makecell}
\usepackage{adjustbox}

\usepackage{float}
\usepackage[skip=0pt]{caption}
\usepackage{bm}
\usepackage{color}
\usepackage{soul}
\usepackage{mathtools}
\usepackage{bigstrut,rotating}
\usepackage{enumitem}

\newcommand{\alg}{{DivergenceRadius}}

\newcommand{\myparatight}[1]{\smallskip\noindent{\bf {#1}:}~}
\newcommand{\mysubparatight}[1]

\newtheorem{thm}{Theorem}

\DeclareMathOperator*{\argmin}{argmin}

\definecolor{cvprblue}{rgb}{0.21,0.49,0.74}
\usepackage[pagebackref,breaklinks,colorlinks,allcolors=cvprblue]{hyperref}

\def\paperID{}
\def\confName{CVPR}
\def\confYear{2026}

\title{Robustness of Vision Foundation Models to Common Perturbations}

\author{Hongbin Liu$^1$, Zhengyuan Jiang$^1$, Cheng Hong$^2$, Neil Zhenqiang Gong$^1$\\
$^1$Duke University, $^2$Ant Group\\
\{hongbin.liu, zhengyuan.jiang, neil.gong\}@duke.edu\\
vince.hc@antgroup.com
}

\begin{document}

\maketitle

\input{0_abstract}
\input{1_introduction}
\input{2_problem}
\input{3_method}
\input{4_exp}
\input{5_related}
\input{6_conclusion}

{
    \small
    \bibliographystyle{ieeenat_fullname}
    \bibliography{main}
}

\clearpage
\input{7_appendix}

\end{document}

%% file: 0_abstract.tex
\begin{abstract}
A vision foundation model outputs an embedding vector for an image, which can be affected by common editing operations (e.g., JPEG compression, brightness, contrast adjustments). These \emph{common perturbations} alter embedding vectors and may impact the performance of downstream tasks using these embeddings. In this work, we present the \emph{first} systematic study on foundation models' robustness to such perturbations. We propose three robustness metrics and formulate five desired mathematical properties for these metrics, analyzing which properties they satisfy or violate. Using these metrics, we evaluate six industry-scale foundation models (OpenAI, Meta) across nine common perturbation categories, finding them generally non-robust. We also show that common perturbations degrade downstream application performance (e.g., classification accuracy) and that robustness values can predict performance impacts. Finally, we propose a fine-tuning approach to improve robustness without sacrificing utility.
\end{abstract}

%% file: 1_introduction.tex
\section{\label{sec:intro}Introduction}

A vision foundation model is a general-purpose feature extractor that outputs an embedding vector for an image. Typically, these models are pre-trained on vast collections of unlabeled images or image-text pairs in a self-supervised manner~\citep{radford2021learning,oquab2023dinov2} by major providers like OpenAI, Meta, and Google. For instance, OpenAI's CLIP~\citep{radford2021learning} jointly trains vision and language foundation models on 400 million image-text pairs, while Meta's DINO v2~\citep{oquab2023dinov2} trains a vision foundation model on a large set of unlabeled images. Foundation models empower various downstream applications like image classification and depth estimation.

Images in real-world settings often undergo common editing operations for various purposes. For instance, JPEG compression is widely used to reduce communication costs online, while brightness and contrast adjustments are also common (Table~\ref{tab_robustness_radius_main} in the Appendix lists 9 editing operations used in our experiments). These operations introduce \emph{common perturbations} to an image, unlike \emph{adversarial perturbations}~\citep{szegedy2013intriguing,carlini2017towards}, which are worst-case modifications designed to mislead models. Common perturbations, by contrast, occur frequently in non-adversarial, real-world scenarios.

The robustness of foundation models and their downstream applications to adversarial perturbations has been widely studied~\citep{jia2022badencoder,liu2022poisonedencoder,li2023embarrassingly,saha2022backdoor,fan2021does,jiang2020robust,qu2023reaas,jiang2023evading}. However, robustness to common perturbations remains largely unexplored. Specifically, three key questions arise regarding robustness to common perturbations: 1) How robust are foundation models, i.e., how much does an embedding vector change when an image undergoes common perturbations? 2) How robust are downstream applications, i.e., to what extent does classifier accuracy degrade with perturbed images? 3) How can we improve robustness in foundation models and their downstream applications against common perturbations?

A key challenge in answering these questions is designing a metric to quantify a foundation model’s robustness to common perturbations. Such a metric would enable systematic robustness assessments, facilitating comparisons across self-supervised learning algorithms, model architectures, and sizes. Additionally, a robustness metric could help predict the performance of downstream applications (e.g., accuracy) for perturbed images and provide guidance to enhance foundation model robustness.

\myparatight{Our work} In this work, we answer the three questions above via performing the \emph{first} systematic study on the robustness of foundation models to common perturbations. 

We begin by tackling the challenge of defining metrics to quantify a foundation model’s robustness to common perturbations. Given an image, common perturbations produce various perturbed versions, each with its own embedding vector. We quantify robustness by measuring the variations among these embedding vectors, exploring three metrics: one based on \emph{cosine similarity}, another on \emph{Euclidean distance}, and a third, \alg{}, which uses the radius of the smallest enclosing ball in the embedding space to capture robustness.

A suitable robustness metric should meet several intuitions, such as not increasing robustness as more perturbations are applied. We formalize these intuitions with five \emph{mathematical properties} and analyze which properties the metrics satisfy. We find that \alg{} satisfies all five properties, whereas the other metrics fail to meet one; additionally, the Euclidean and cosine similarity metrics are equivalent.

Using our robustness metrics, we address the first question through a systematic study of six industry-scale foundation models from the CLIP (OpenAI) and DINO v2 (Meta) families, covering different self-supervised learning algorithms, architectures, and sizes, across nine categories of common perturbations. Our findings are consistent across the three robustness metrics: foundation models generally lack robustness to common perturbations, often producing divergent embeddings for perturbed images. Additionally, we observe that foundation models based on Vision Transformer architectures are more robust than those based on ResNet architectures.

To address the second question, we evaluate the robustness of downstream classifiers and depth estimation models built on industry-scale foundation models against common perturbations. We observe that these perturbations degrade both classifier accuracy and depth estimation performance; for example, glass-blurring reduces the accuracy of a zero-shot ImageNet classifier by 9.4\%. This occurs due to variations in embedding vectors caused by perturbations. Additionally, we find that average classification accuracy and mean squared error of depth maps for perturbed images are roughly linear functions of the image's robustness value (e.g., cosine similarity or \alg{}), enabling accurate performance predictions for downstream tasks using a simple linear regression model.

Finally, we propose a fine-tuning method to enhance a foundation model’s robustness while preserving utility for downstream tasks. Our approach aims to balance two objectives: a \emph{robustness goal} and a \emph{utility goal}, each quantified by a corresponding loss term. We fine-tune the model by minimizing a weighted sum of these loss terms, with empirical results showing that our method successfully improves robustness without compromising utility.

%% file: 2_problem.tex
\section{\label{sec:problem}Problem Formulation}

\myparatight{Perturbation function} We represent common perturbations with a perturbation function $P(x, k)$, where $x$ is an image and $k$ a perturbation parameter, yielding a perturbed image $P(x, k)$. For instance, if $P$ represents JPEG compression, $k$ is the quality factor controlling compression level. For some functions, $k$ is multi-dimensional, such as \emph{fog blurring}, where $k$ includes \emph{density} and \emph{frequency}.

We denote the domain of $k$ as $\mathbb{K}$, the set from which $k$ is selected when applying $P$ to $x$. The domain $\mathbb{K}$ may include discrete values (e.g., JPEG quality factors) or continuous values (e.g., Gaussian noise standard deviation). We assume a special parameter $\bot \in \mathbb{K}$, where $P(x, \bot) = x$, returning the original image, to simplify descriptions.

\myparatight{Embedding vector} A foundation model $f$ outputs an embedding vector $f(x)$ for an image $x$. To prevent embedding magnitude from affecting downstream applications, foundation models often normalize embeddings to an $\ell_2$-norm of 1, ensuring $||f(x)||_2=1$ for any image $x$. Thus, all embedding vectors lie on a unit-radius hyper-sphere in the embedding space.

\myparatight{Desired mathematical properties of a robustness metric} 
Given a foundation model $f$, an image $x$, and a perturbation function $P$ with parameter domain $\mathbb{K}$, our goal is to define a robustness metric $\mathcal{R}(f, x, P, \mathbb{K})$ to quantify the robustness of $f$ for $x$ under $P$. This metric essentially measures variations among embedding vectors $\{f(P(x, k))\}_{k \in \mathbb{K}}$ for perturbed versions of $x$ generated by $P$.

The robustness metric should allow quantitative comparisons of different foundation models’ robustness to common perturbations. Since foundation models support downstream applications, the metric should also predict downstream performance for $x$ under $P$. For instance, if the downstream task is classification, the robustness value $\mathcal{R}(f, x, P, \mathbb{K})$ should help predict the accuracy for perturbed versions of $x$ within the domain $\mathbb{K}$.

\begin{table*}[!t]
\small
\centering
\caption{Three robustness metrics explored in this work and the desired mathematical properties they violate.} 
\begin{tabular}{|c|c|c|}
\hline
{Robustness Metric}    & {Formulation}                                    & {Violating Properties}      \\ \hline
\multirow{2}{*}{Cosine similarity}           &  \multirow{2}{*}{$\mathcal{R}_{cs}(f,x,P,\mathbb{K})  = \frac{1 - \min_{k_1,k_2 \in \mathbb{K}} cos(f(P(x,k_1)), f(P(x,k_2)))}{2}$} & \multirow{2}{*}{Worst-robustness}    \\ 
&   & \\ \hline
\multirow{2}{*}{Euclidean distance}           & \multirow{2}{*}{$\mathcal{R}_{ed}(f,x,P,\mathbb{K}) = \frac{\max_{k_1,k_2 \in \mathbb{K}} ||f(P(x,k_1)) - f(P(x,k_2))||_2}{2}$} & \multirow{2}{*}{Worst-robustness}      \\ 
&   &   \\ \hline
\multirow{2}{*}{\alg{}}           & \multirow{2}{*}{$\mathcal{R}_{dr}(f,x,P,\mathbb{K}) = \argmin r \text{  s.t.  } \exists c, ||f(P(x,k))  - c||_2 \leq r,  \forall k \in \mathbb{K}$} & \multirow{2}{*}{None}      \\ 
&   &   \\ \hline
\end{tabular}
\label{tab:negative_example_metrics}
\end{table*}

We have the following mathematical properties:

\begin{enumerate}[left=0.15cm]
    \item \myparatight{Bounded domain} To ensure interpretability and comparability, we design a scalar robustness metric with a bounded interval output, normalized to $[0,1]$ for simplicity. Thus, for any $f$, $x$, $P$, and $\mathbb{K}$, the robustness value $\mathcal{R}(f, x, P, \mathbb{K})$ should fall within $[0,1]$, where a larger value indicates \emph{less robustness}:
    \begin{align}
        \mathcal{R}(f,x,P,\mathbb{K})\in [0,1], \forall f, x, P, \mathbb{K}. 
    \end{align}

    \item \myparatight{Monotonicity} When the parameter domain $\mathbb{K}$ expands, the model should not become more robust under this larger domain. For instance, greater variation in JPEG quality factors should not decrease $\mathcal{R}(f,x,P,\mathbb{K})$. Formally:
    \begin{align}
        \mathcal{R}(f,x,P,\mathbb{K}_1) \leq \mathcal{R}(f,x,P,\mathbb{K}_2), \forall \mathbb{K}_1 \subseteq \mathbb{K}_2. 
    \end{align}

    \item \myparatight{Best robustness} The model is maximally robust for $x$ under $P$ if all perturbed versions of $x$ have the same embedding as $x$, resulting in $\mathcal{R}(f,x,P,\mathbb{K}) = 0$ if:
    \begin{align}
    f(P(x, k))=f(x), \forall k \in \mathbb{K}.
    \end{align}

    \begin{figure}[!t]
    \centering
    {\includegraphics[width=0.1\textwidth]{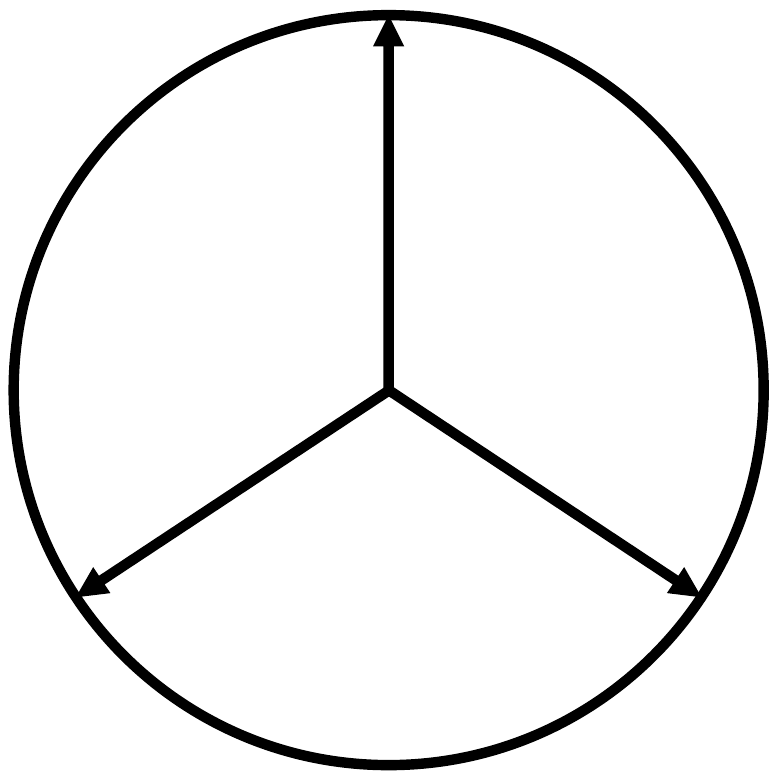}}
    \caption{An example for the worst-robustness property.}
    \label{fig:example-worst-robustness}
    \end{figure}

    \item \myparatight{Worst robustness} If the embedding vectors for perturbed images are uniformly distributed in the embedding space (sum to zero), then robustness is at its lowest with $\mathcal{R}(f,x,P,\mathbb{K}) = 1$, if $\exists \mathbb{K}' \subseteq \mathbb{K}$:
    \begin{align}
    \sum_{k \in \mathbb{K}'} f(P(x,k)) = \mathbf{0}.
    \end{align}
    where $\mathbf{0}$ is the zero vector. Figure~\ref{fig:example-worst-robustness} illustrates this in a two-dimensional space.

    \item \myparatight{Rotational invariance} Since embeddings are rotation-free, the robustness metric should be invariant to rotations in the embedding space. Let $M$ be a rotation matrix, then:
    \begin{align}
         \mathcal{R}(M\cdot f,x,P,\mathbb{K}) = \mathcal{R}(f,x,P,\mathbb{K}),
    \end{align}
    where $M\cdot f$ denotes rotating $f(x)$ by $M$. 
\end{enumerate}

%% file: 3_method.tex
\section{\label{sec:method}Robustness Metrics} 

We explore three robustness metrics and theoretically analyze what desired mathematical properties they satisfy/violate, which we summarize in Table~\ref{tab:negative_example_metrics}.

\subsection{Cosine Similarity} 
Since embedding vectors lie on a hypersphere, we can use the angles between them to quantify robustness. Specifically, the largest angle (smallest cosine similarity) between any two embedding vectors from perturbed images can define a robustness metric. Formally, for a model $f$, image $x$, perturbation function $P$, and parameter domain $\mathbb{K}$, the cosine similarity-based robustness metric $\mathcal{R}_{cs}$ is:
\begin{align}
\label{equ:cosine_similarity}
    &\mathcal{R}_{cs}(f,x,P,\mathbb{K}) \nonumber\\=& \frac{1 - \min_{k_1,k_2 \in \mathbb{K}} cos(f(P(x,k_1)), f(P(x,k_2)))}{2},
\end{align}
where $cos(\cdot, \cdot)$ is the cosine similarity. The constants normalize the value to [0,1].

We can verify that the robustness metric $\mathcal{R}_{cs}$ satisfies the five mathematical properties except the worst-robustness one. In particular, since $cos(\cdot, \cdot)$ is in [-1,1], $\mathcal{R}_{cs}(f,x,P,\mathbb{K})$ lies in [0, 1], meeting the bounded-domain property. As $\mathbb{K}$ expands, $\min_{k_1,k_2 \in \mathbb{K}} cos(f(P(x,k_1)), f(P(x,k_2)))$ does not increase, so $\mathcal{R}_{cs}$ does not decrease, satisfying monotonicity. The best-robustness property holds because $\mathcal{R}_{cs}(f,x,P,\mathbb{K})=0$ when all perturbed versions of $x$ have the same embedding. Finally, $\mathcal{R}_{cs}$ is rotation-invariant since cosine similarity is.

However, $\mathcal{R}_{cs}$ does not satisfy the worst-robustness property. For instance, in Figure~\ref{fig:example-worst-robustness}, $\mathcal{R}_{cs}$ is 0.75, not 1, as $\mathcal{R}_{cs}=1$ only when embedding vectors cover exactly half the hypersphere, with cosine similarity of -1. If embedding vectors are more evenly distributed, the smallest cosine similarity is greater than -1, violating worst-robustness.

\myparatight{Monotonically increasing function of cosine similarity} A robustness metric based on any monotonically increasing function of $\mathcal{R}_{cs}$ cannot satisfy all five desired properties. Formally, we define such a metric:
\begin{align}
\label{eq:pcs}
    \mathcal{R}_g(f,x,P,\mathbb{K}) =  g(\mathcal{R}_{cs}(f,x,P,\mathbb{K})),
\end{align}
where $g$ is a monotonically increasing function satisfying $g(0)=0$ and $g(1)=1$. We can show $\mathcal{R}_g$ satisfies bounded-domain, best-robustness, and rotation-invariance properties but fails the worst-robustness property. Formally, we have:

\begin{thm}[Monotonically Increasing Function of $\mathcal{R}_{cs}$]\label{theorem_imposibility_cosine_function}
Given any $g$ that is a monotonically increasing function of $\mathcal{R}_{cs}$ and satisfies $g(0)=0$ and $g(1)=1$, $\mathcal{R}_g = g(\mathcal{R}_{cs})$ does not satisfy the worst-robustness property. 
 \end{thm}
\begin{proof}
To prove this, we construct a counter-example. Since $\mathcal{R}_g = g(\mathcal{R}_{cs})$ increases with $\mathcal{R}_{cs}$, and since $g(0)=0$ and $g(1)=1$, we have $\mathcal{R}_g < 1$ if $\mathcal{R}_{cs} < 1$. In Figure~\ref{fig:example-worst-robustness}, $\mathcal{R}_{cs} = 0.75 < 1$, so $\mathcal{R}_g < 1$, countering the worst-robustness property.
\end{proof}

\subsection{Euclidean Distance} \label{sec:euclideandistance}
Another intuitive robustness metric is to use the Euclidean distances between embedding vectors of perturbed images to quantify their variations. Formally, given a model $f$, image $x$, perturbation function $P$, and parameter domain $\mathbb{K}$, we define a Euclidean distance-based robustness metric $\mathcal{R}_{ed}$ as:
\begin{align}
\label{equ:euclidean_distance}
&\mathcal{R}_{ed}(f,x,P,\mathbb{K}) \nonumber \\
=& \frac{\max_{k_1,k_2 \in \mathbb{K}} ||f(P(x,k_1)) - f(P(x,k_2))||_2}{2},
\end{align}
where $||\cdot||_2$ denotes the Euclidean distance, with the constant 2 normalizing the metric to [0,1].

We can show that $\mathcal{R}_{ed}$ is a monotonically increasing function of $\mathcal{R}_{cs}$, specifically as follows:

\begin{thm}\label{theorem_equivalence} For any model $f$, image $x$, perturbation function $P$, and parameter domain $\mathbb{K}$, $\mathcal{R}_{ed}$ is the square root of $\mathcal{R}_{cs}$:
\begin{align}
    \mathcal{R}_{ed}(f,x,P,\mathbb{K}) = \sqrt{\mathcal{R}_{cs}(f,x,P,\mathbb{K})}.
\end{align}
\end{thm}
\begin{proof}
    Please refer to the Appendix. 
\end{proof}

Theorem~\ref{theorem_equivalence} further shows that $\mathcal{R}_{cs}$ and $\mathcal{R}_{ed}$ are equivalent, as $\mathcal{R}_{cs}$ can be converted to $\mathcal{R}_{ed}$ by taking the square root. Therefore, we omit results for $\mathcal{R}_{ed}$ in our experiments for simplicity.

\subsection{\alg{}} 

\subsubsection{Formulating an Optimization Problem} We first outline the intuition behind \alg{} and then present its formulation.

\myparatight{Intuition} 
The cosine similarity and Euclidean distance metrics do not satisfy the worst-robustness property, which requires a maximum robustness value of 1 when the embedding vectors for perturbed images are equally distributed in the embedding space. Our intuition is that if the embedding vectors are uniformly spread on the unit hyper-sphere, the radius of the smallest ball enclosing these vectors will equal 1. Thus, constructing a minimum enclosing ball for the perturbed embedding vectors can satisfy the worst-robustness property and indicate robustness: a smaller radius implies greater robustness.

\myparatight{Optimization problem} 
We aim to find the smallest-radius high-dimensional ball with center $c$ that encloses all embedding vectors of perturbed versions of an image $x$ within the perturbation domain $\mathbb{K}$. This radius $r$, our \alg{}, is denoted as $\mathcal{R}_{dr}(f,x,P,\mathbb{K})$. Formally:
\begin{align}
\label{divergence_radius}
&\mathcal{R}_{dr}(f,x,P,\mathbb{K}) = \argmin r, \nonumber\\ &\text{  s.t.  } \exists c, ||f(P(x,k))  - c||_2 \leq r,  \forall k \in \mathbb{K}.
\end{align}

Figure~\ref{fig:minimum-enclosing-ball} illustrates an example of a minimum ball enclosing embedding vectors for three perturbed versions of an image $x$. Note that $\mathbb{K}$ includes a special parameter $\bot$ where $P(x, \bot) = x$.

\begin{figure}[!t]
    \centering
    {\includegraphics[width=0.2\textwidth]{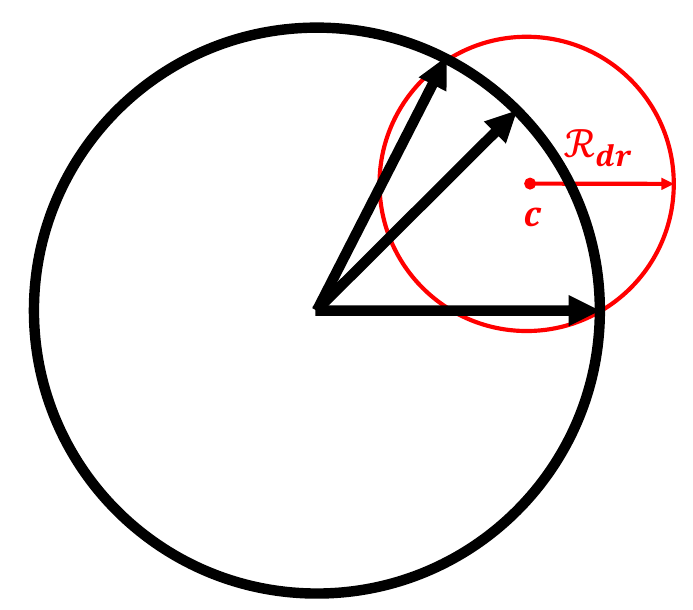}}
    \caption{An example to illustrate the minimum enclosing ball. The red circle is the minimum enclosing ball of the three embedding vectors.}
    \label{fig:minimum-enclosing-ball}
\end{figure}

\subsubsection{\alg{} Satisfies the Mathematical Properties} 

We show that the robustness metric $\mathcal{R}_{dr}$ satisfies all the five mathematical properties. 

\begin{proof}
Please refer to the Appendix.
\end{proof}

\subsubsection{Solving \alg{}} We consider both discrete and continuous $\mathbb{K}$.

\myparatight{Discrete $\mathbb{K}$} For a discrete domain $\mathbb{K}$, the embedding vectors $\{f(P(x,k))\}_{k\in\mathbb{K}}$ are discrete points on the unit hyper-sphere. In this case, we can use Welzl's algorithm~\citep{welzl2005smallest} to efficiently find the minimum enclosing ball’s center and radius, solving Equation~\ref{divergence_radius} in $O(dn)$ time, where $d$ is the embedding dimension and $n$ the number of discrete values in $\mathbb{K}$.

\myparatight{Continuous $\mathbb{K}$} When $\mathbb{K}$ is continuous, finding an exact solution is infeasible due to infinite embedding vectors. To approximate, we sample a discrete subset $\bar{\mathbb{K}}$ from $\mathbb{K}$, transforming the problem into a discrete one. We use two sampling methods: \emph{random sampling} (uniform random selection) and \emph{equally-spaced sampling}, where equally spaced values better represent the domain when using fewer samples. For a domain $\mathbb{K} = [a,b]$ and $m$ samples, equally-spaced sampling yields values $a, a + (b-a)/(m-1), \ldots, b$. Our experiments show that equally-spaced sampling outperforms random sampling for \alg{} estimation accuracy.

After obtaining $\bar{\mathbb{K}}$, we apply Welzl's algorithm to find an approximate \alg{}. For large discrete domains, a small sample set can similarly reduce computational cost. Random and equally-spaced sampling methods also apply to estimating the cosine similarity-based robustness $\mathcal{R}_{cs}$ and Euclidean distance-based robustness $\mathcal{R}_{ed}$.

%% file: 4_exp.tex
\input{4.1_exp_setup}
\input{4.2_exp_measurement}
\input{4.3_exp_downstream}
\input{4.4_exp_enhancement}

%% file: 4.1_exp_setup.tex
\section{\label{sec:measurement} Measuring Robustness of Real-world Foundation Models}

In this section, we evaluate the robustness of industry-scale foundation models using cosine similarity and \alg{}. Although cosine similarity does not theoretically satisfy the worst-robustness property, we include it in our experiments since the worst-robustness scenario does not occur in the evaluated perturbations. We omit Euclidean distance results due to its equivalence with cosine similarity.

\subsection{Measurement Setup}
\myparatight{Foundation models} We evaluate foundation models pre-trained using various algorithms, architectures, and sizes, allowing comparison of pre-training methods and model robustness. Specifically, we evaluate two popular families of vision foundation models: CLIP~\citep{radford2021learning} and DINO v2~\citep{oquab2023dinov2}. CLIP (by OpenAI) was trained with \emph{multi-modal self-supervised learning} on 400 million image-text pairs, while DINO v2 (by Meta) used \emph{self-supervised learning} on 142 million unlabeled images. In the CLIP family, we test ViT-B/16, ViT-L/14, RN50, and RN50$\times$64 (Vision Transformer and ResNet architectures). In the DINO v2 family, we evaluate ViT-L/14 and ViT-g/14. Details on these models are in Table~\ref{tab:models} in the Appendix.

\myparatight{Datasets} For robustness evaluation, we use two image classification datasets, ImageNet~\citep{deng2009imagenet} and Food101~\citep{bossard14}, and one depth estimation dataset, NYU-Depth V2~\citep{Silberman:ECCV12}. Details of these datasets are provided in Table~\ref{tab:dataset} (Appendix). Although labels are not required for robustness testing, they are used later to evaluate downstream applications.

\myparatight{Common perturbations}
Following~\citep{hendrycks2019benchmarking}, we use nine common perturbation functions representing typical image editing operations: \textit{JPEG compression}, \textit{Brightness adjustment}, \textit{Contrast adjustment}, \textit{Defocus blurring}, \textit{Elastic blurring}, \textit{Fog blurring}, \textit{Frost blurring}, \textit{Gaussian noise}, and \textit{Glass blurring}. Each has one variable parameter, with fixed values for additional parameters if present.

Table~\ref{tab_robustness_radius_main} in the Appendix details each perturbation's parameter domain $\mathbb{K}$, additional fixed parameters, and a visualization of a maximally distorted example for each function. Following prior work~\citep{hendrycks2019benchmarking}, the selected $\mathbb{K}$ for each perturbation reflects realistic image editing in real-world scenarios.

%% file: 4.2_exp_measurement.tex
\begin{figure}[!t]
    \centering
    \includegraphics[width=0.33\textwidth]{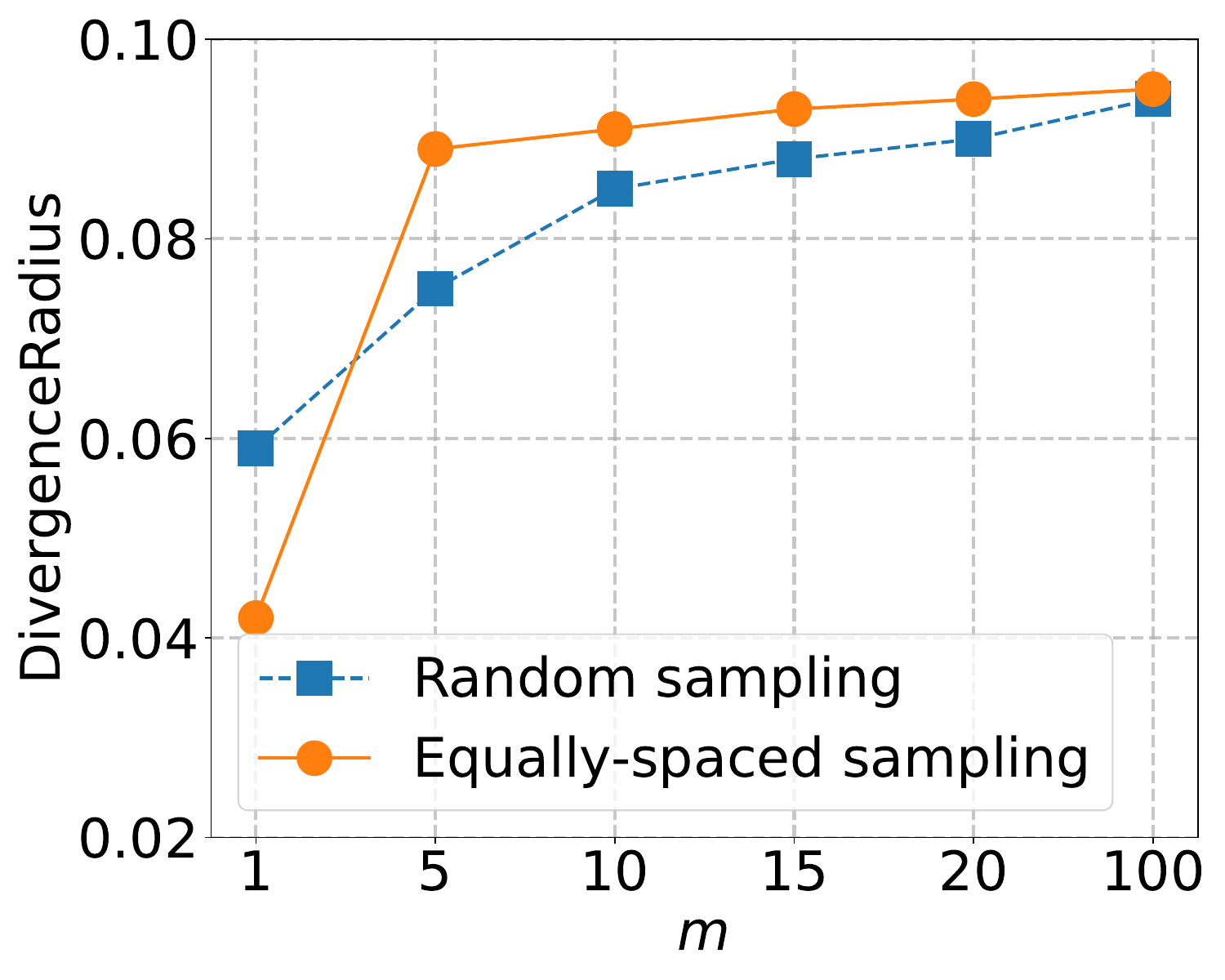}
    \caption{Comparing random sampling and equally-spaced sampling at computing \alg{}, where  CLIP ViT-L/14 foundation model, JPEG compression, and ImageNet images are used. $m$ is the number of discrete values sampled from the domain $\mathbb{K}$. }
    \label{fig:impact_sampling}
\end{figure}

\begin{figure*}[t!]
    \centering
    \subfloat[JPEG compression]{\includegraphics[width=0.33\textwidth]{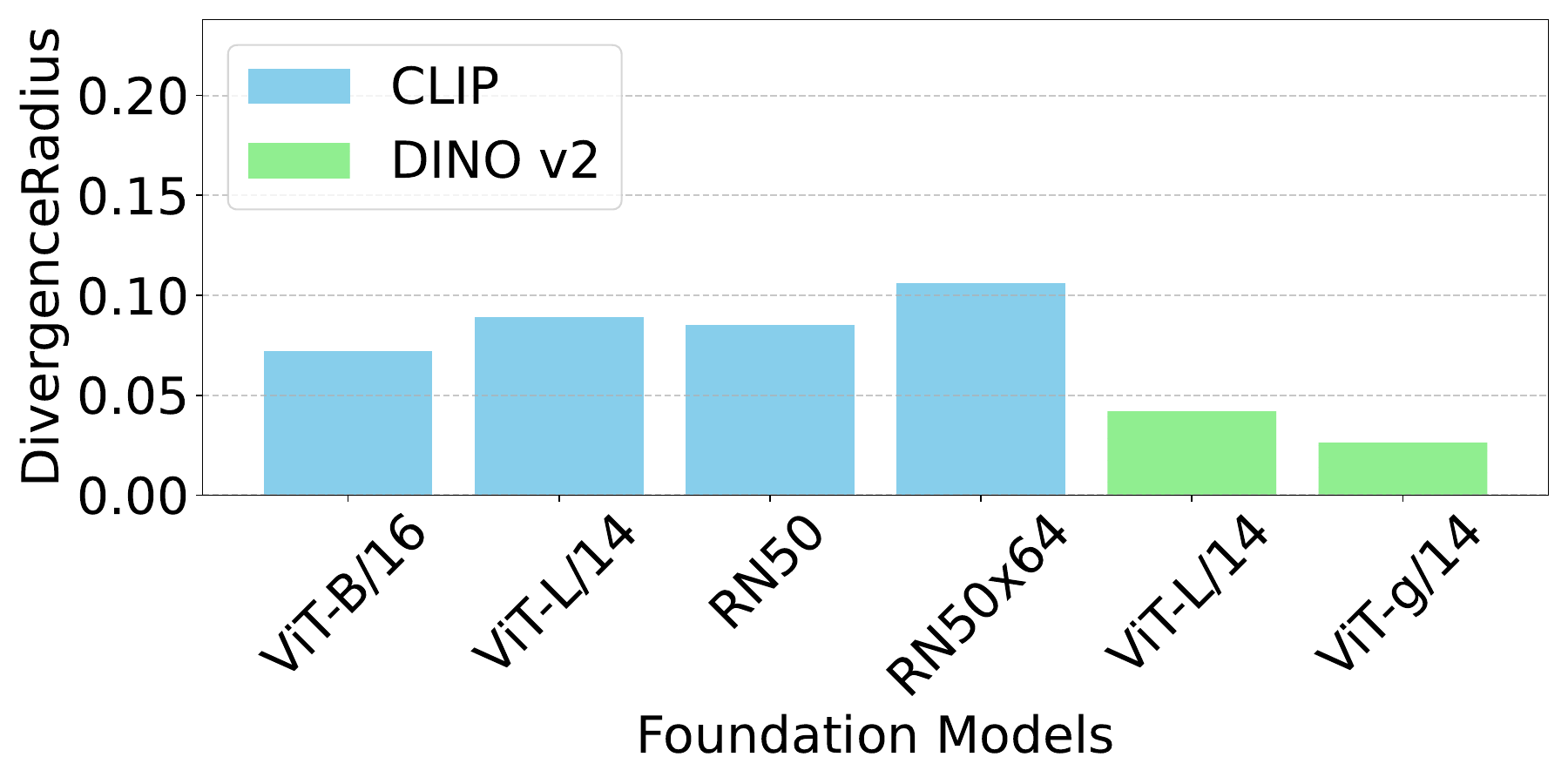}\label{fig:1}}
    \hfill
    \subfloat[Brightness adjustment]{\includegraphics[width=0.33\textwidth]{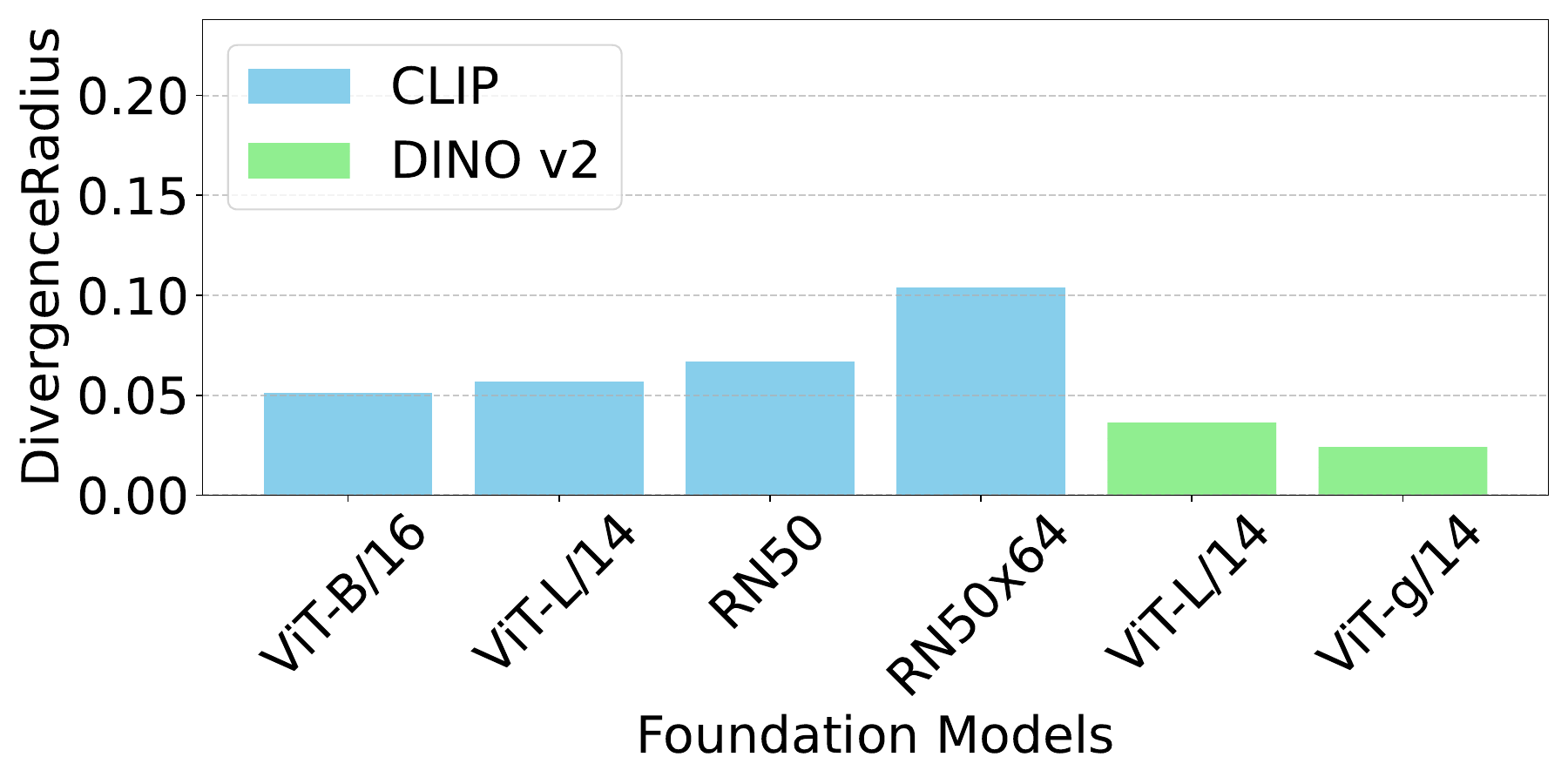}\label{fig:2}}
    \hfill
    \subfloat[Contrast adjustment]{\includegraphics[width=0.33\textwidth]{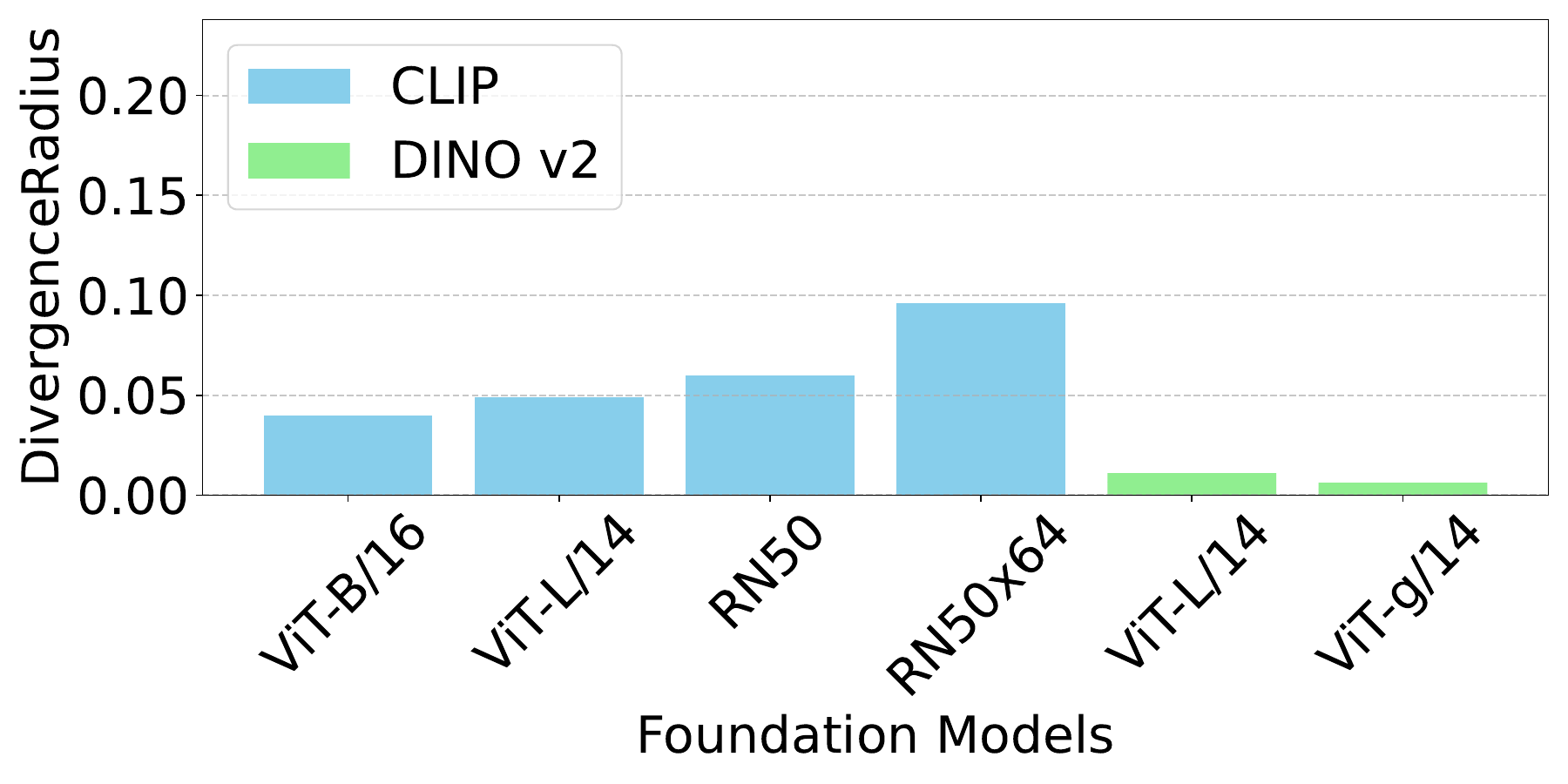}\label{fig:3}}
    \\
    \subfloat[Defocus blurring]{\includegraphics[width=0.33\textwidth]{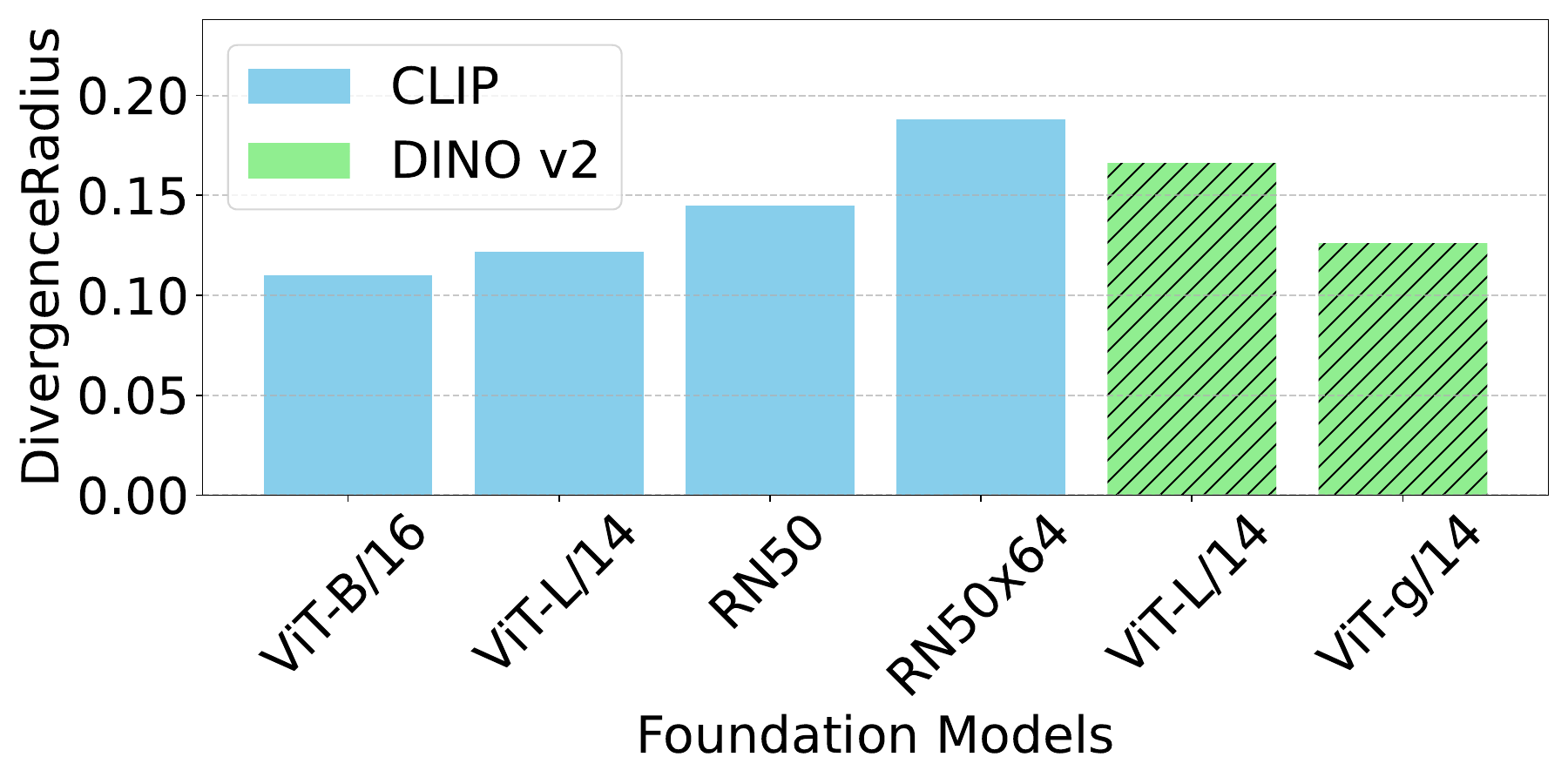}\label{fig:4}}
    \hfill
    \subfloat[Elastic blurring]{\includegraphics[width=0.33\textwidth]{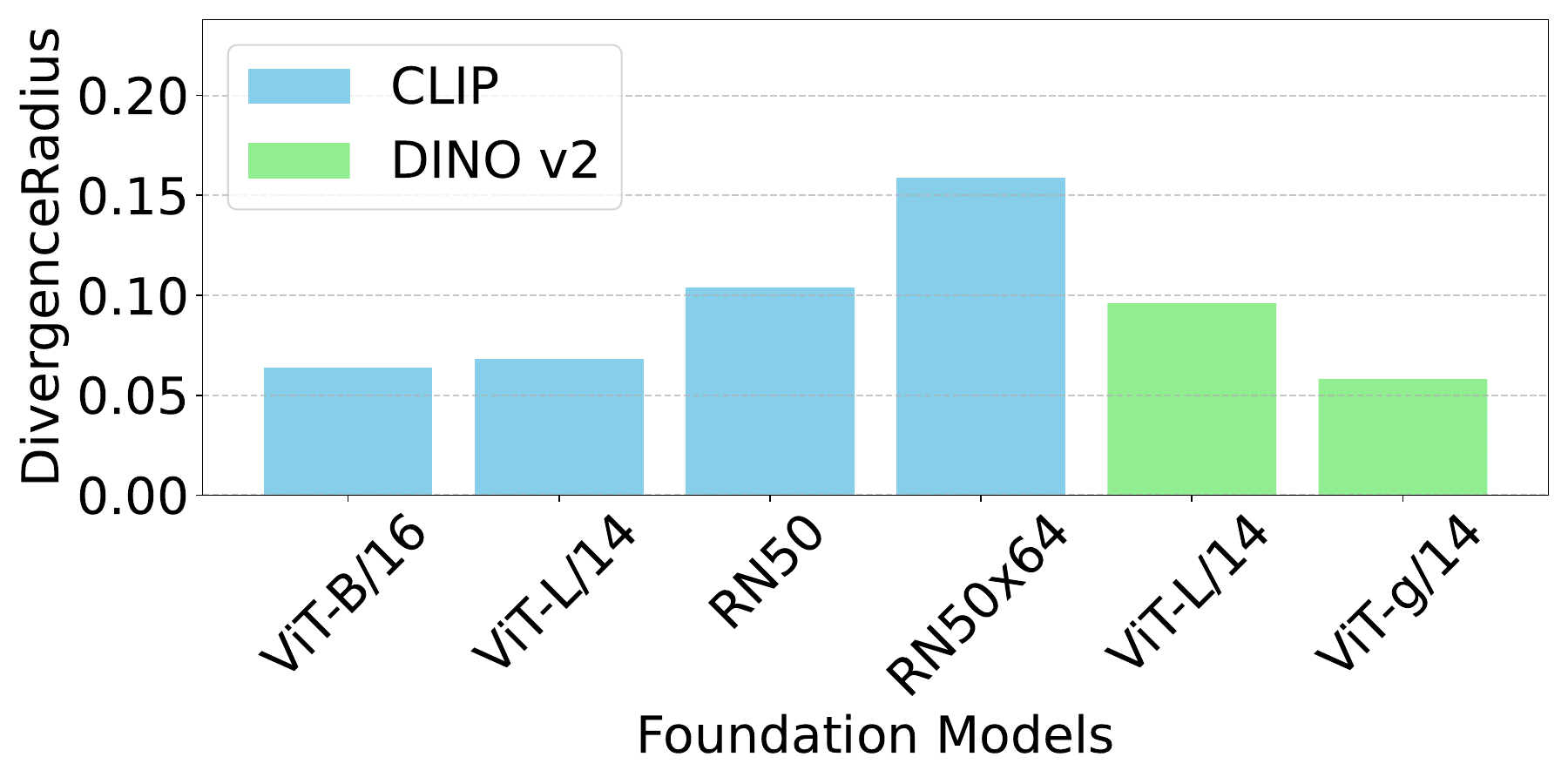}\label{fig:5}}
    \hfill
    \subfloat[Fog blurring]{\includegraphics[width=0.33\textwidth]{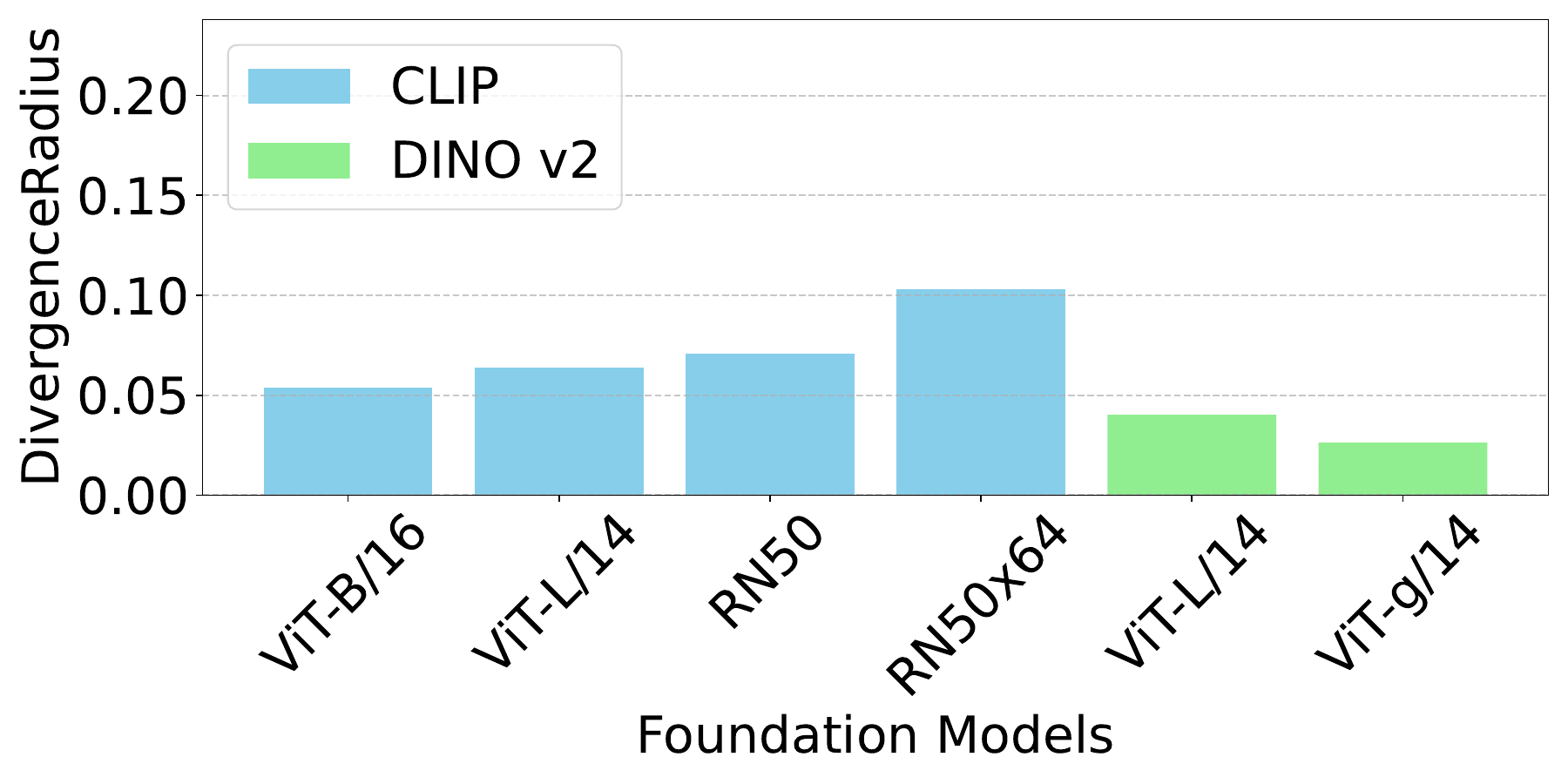}\label{fig:6}}
    \\
    \subfloat[Frost blurring]{\includegraphics[width=0.33\textwidth]{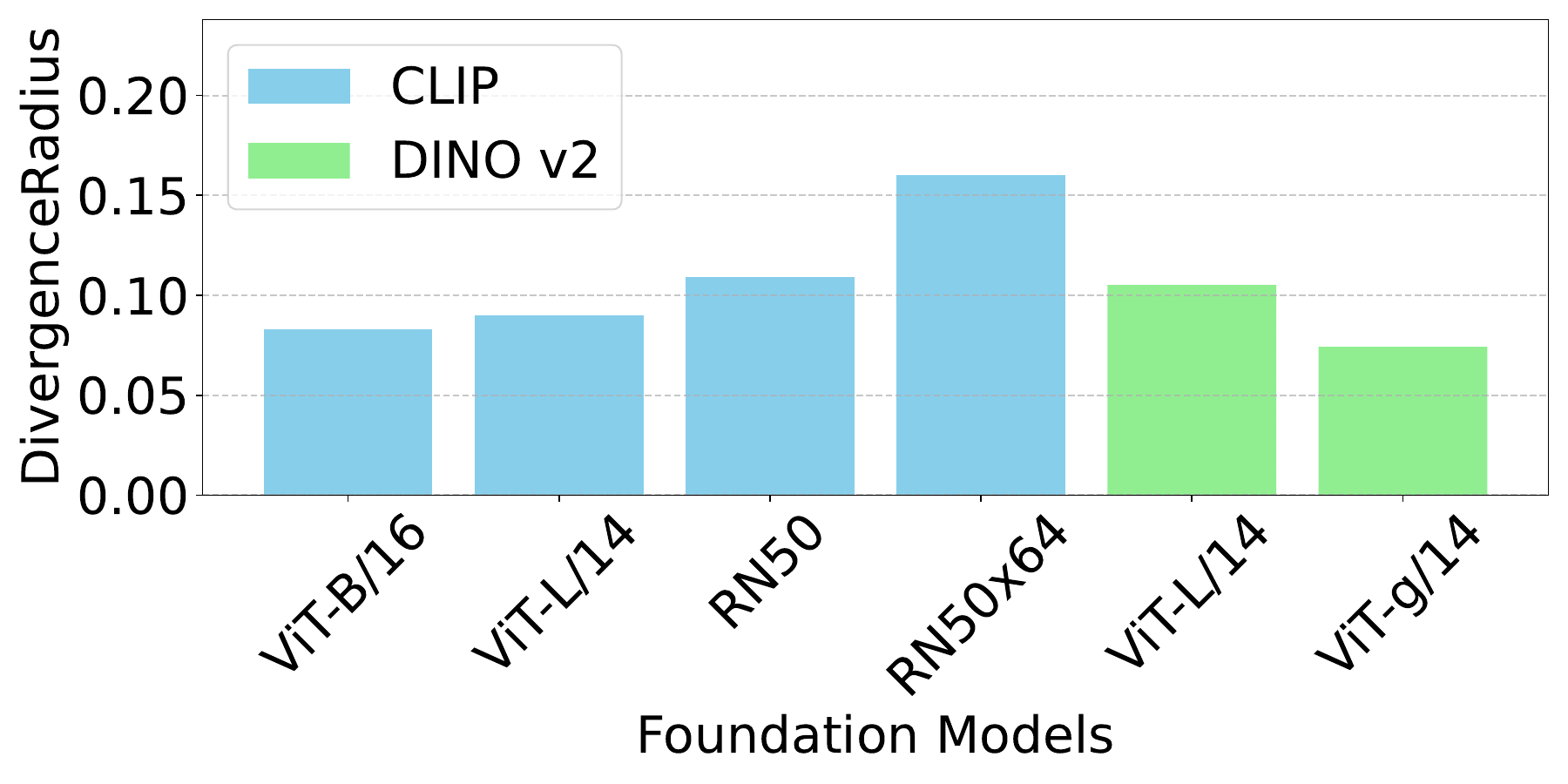}\label{fig:7}}
    \hfill
    \subfloat[Gaussian noise]{\includegraphics[width=0.33\textwidth]{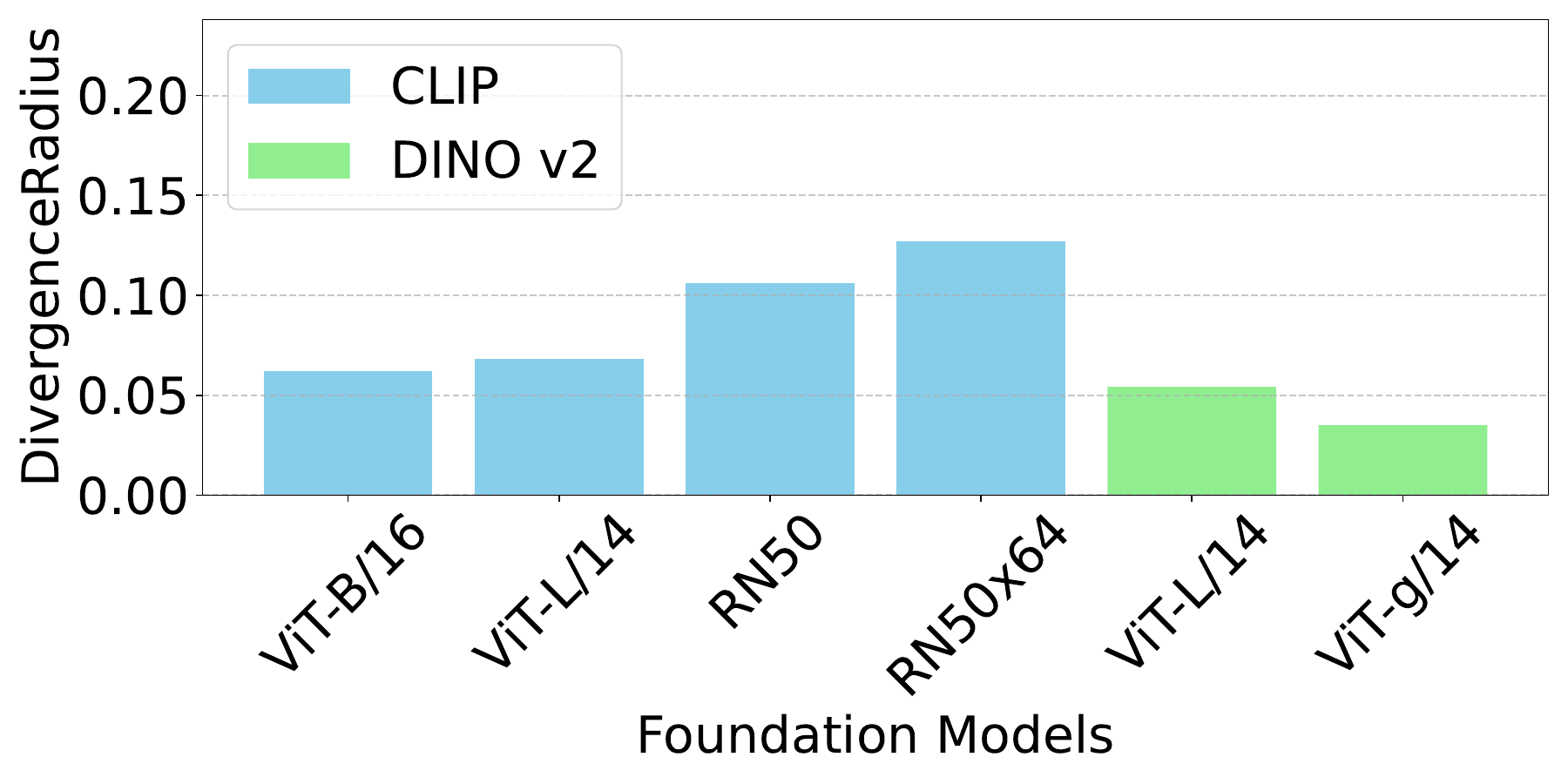}\label{fig:8}}
    \hfill
    \subfloat[Glass blurring]{\includegraphics[width=0.33\textwidth]{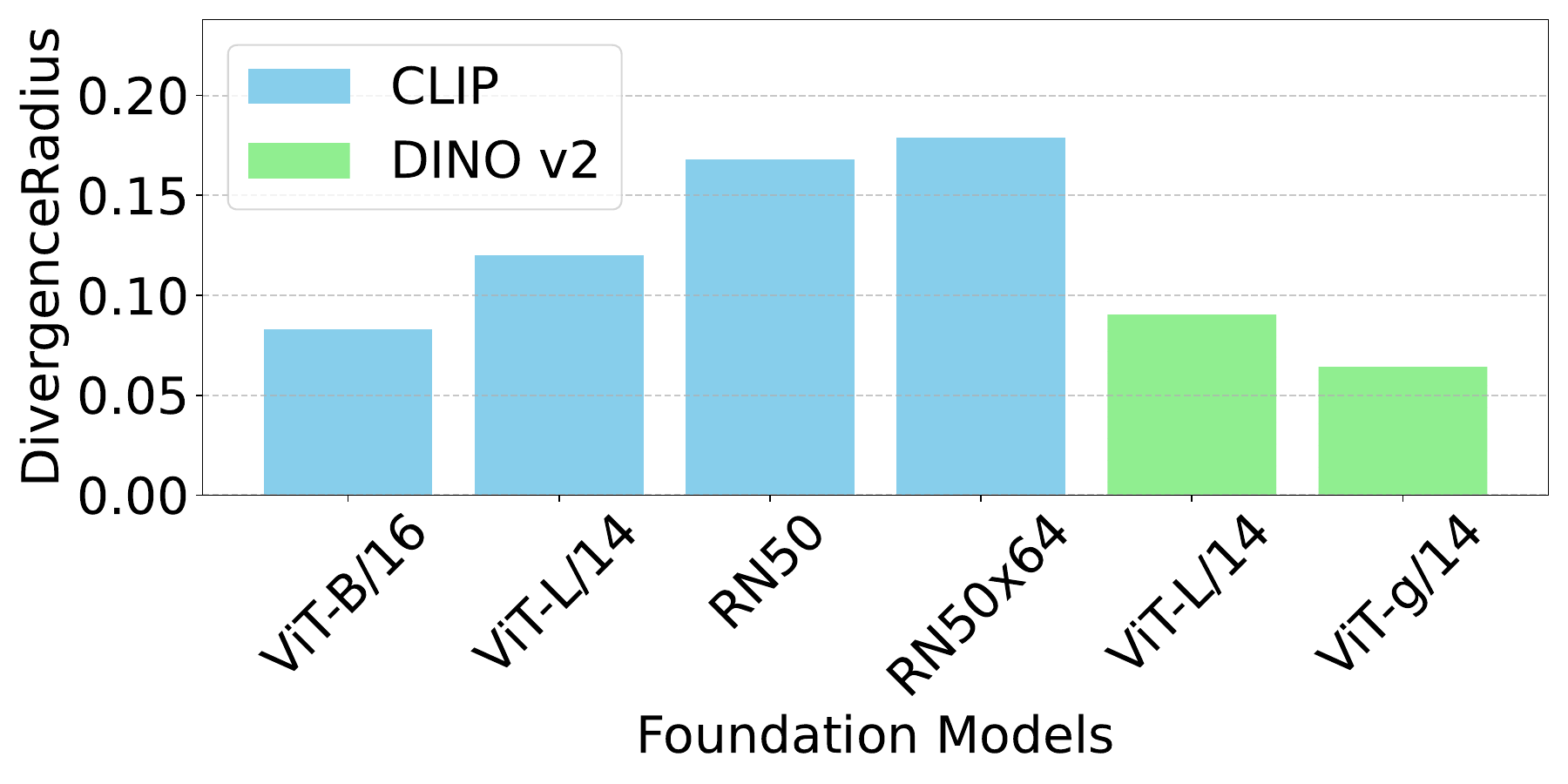}}
    \caption{Average \alg{} of ImageNet testing images for different foundation models and  perturbation functions.}
    \label{fig:divergence_radius_vision_imagenet}
    \vspace{-2mm}
\end{figure*}

\subsection{Measurement Results}

\myparatight{Random sampling v.s. equally-spaced sampling} When computing robustness values (i.e., $\mathcal{R}_{cs}$ or $\mathcal{R}_{dr}$), we can use \emph{random sampling}, which selects $m$ discrete values from $\mathbb{K}$ at random, or \emph{equally-spaced sampling}, which selects $m$ evenly spaced values. We compare these methods by computing $\mathcal{R}_{dr}$ for a foundation model, perturbation function, and image. Figure~\ref{fig:impact_sampling} shows the average $\mathcal{R}_{dr}$ on ImageNet images for CLIP ViT-L/14 under JPEG compression as $m$ varies.

We observe that \alg{} initially increases and then saturates with both sampling methods. As $m$ grows, more diverse parameters from $\mathbb{K}$ yield a higher $\mathcal{R}_{dr}$, approaching the true robustness value. Equally-spaced sampling converges more quickly, reaching near-saturation at $m \geq 5$, whereas random sampling requires $m \geq 20$. This indicates that equally-spaced sampling provides a more efficient approximation of $\mathcal{R}_{dr}$. Thus, for computational efficiency, we use equally-spaced sampling with $m=5$ for each perturbation function in other experiments.

\myparatight{Comparing model architectures} Figure~\ref{fig:divergence_radius_vision_imagenet}, Figure~\ref{fig:divergence_radius_vision_Food101} (Appendix), and Figure~\ref{fig:divergence_radius_vision_nyu} (Appendix) show the average $\mathcal{R}_{dr}$ on ImageNet, Food101, and NYU-Depth V2 datasets across different foundation models and perturbations. Figures~\ref{fig:cosine_vision_imagenet}, \ref{fig:cosine_vision_Food101}, and \ref{fig:cosine_vision_nyu} (Appendix) present the corresponding $\mathcal{R}_{cs}$ results. To compare architectures, we examine models with the same pre-training algorithm and similar sizes: specifically, CLIP ViT-B/16 vs. CLIP RN50 and CLIP ViT-L/14 vs. CLIP RN50×64. We observe that ViT-B/16 (or ViT-L/14) is consistently more robust than RN50 (or RN50×64) across all perturbations and datasets, as evidenced by lower \alg{} values. This suggests that Vision Transformers are generally more robust to common image perturbations than ResNet architectures. Similar results are also observed when using cosine similarity.

\myparatight{Comparing pre-training algorithms} To compare pre-training algorithms, we evaluate the robustness of CLIP ViT-L/14 and DINO v2 ViT-L/14, which share the same architecture and model size. We find no consistent trend in robustness across perturbation types. For example, DINO v2 ViT-L/14 has lower \alg{} (i.e., greater robustness) than CLIP ViT-L/14 on JPEG compression, Brightness adjustment, Contrast adjustment, Fog blurring, Gaussian noise, and Glass blurring across all datasets. However, DINO v2 ViT-L/14 shows higher \alg{} (i.e., lower robustness) on other perturbations. For instance, under JPEG compression, the average \alg{} for ImageNet images is 0.038 for DINO v2 ViT-L/14 and 0.084 for CLIP ViT-L/14, while under Defocus blurring, it is 0.146 for DINO v2 ViT-L/14 and 0.108 for CLIP ViT-L/14.

\myparatight{Comparing model sizes} For model size comparisons, we examine foundation models within the same architecture and pre-training algorithm: CLIP ViT-B/16 vs. CLIP ViT-L/14, CLIP RN50 vs. CLIP RN50x64, and DINO v2 ViT-L/14 vs. DINO v2 ViT-g/14. In the CLIP family, larger models are generally less robust to common perturbations than smaller ones. For example, ViT-L/14 shows a higher average $\mathcal{R}_{dr}$ (or $\mathcal{R}_{cs}$) than ViT-B/16 across all perturbations and datasets, and RN$50\times64$ has higher robustness metrics than RN50, except for a few cases on Food101 (e.g., JPEG compression, Defocus blurring). Conversely, in the DINO v2 family, larger models are more robust: ViT-g/14 consistently shows lower $\mathcal{R}_{dr}$ values than ViT-L/14 across perturbations and datasets. This contrast suggests that pre-training settings—multi-modal self-supervised learning for CLIP vs. image-only self-supervised learning for DINO v2—impact robustness differently as model size increases.

%% file: 4.3_exp_downstream.tex
\section{\label{sec:exp_downstream}Measuring Performance of Downstream Applications}

In this section, we measure the robustness of downstream applications to common perturbations and demonstrate that an image's robustness value (i.e., $\mathcal{R}_{dr}$ or $\mathcal{R}_{cs}$) can predict the performance of downstream applications on its perturbed versions. We omit Euclidean distance results due to its equivalence with cosine similarity.

\subsection{Experimental Setup}

\myparatight{Downstream applications} Given a pre-trained vision foundation model, we consider the following three popular downstream applications: zero-shot classification, linear-probe classification, and depth estimation. The Appendix shows more details of these applications.

\begin{figure*}[t!]
    \centering
    \subfloat[Zero-shot classification]{\includegraphics[width=0.25\textwidth]{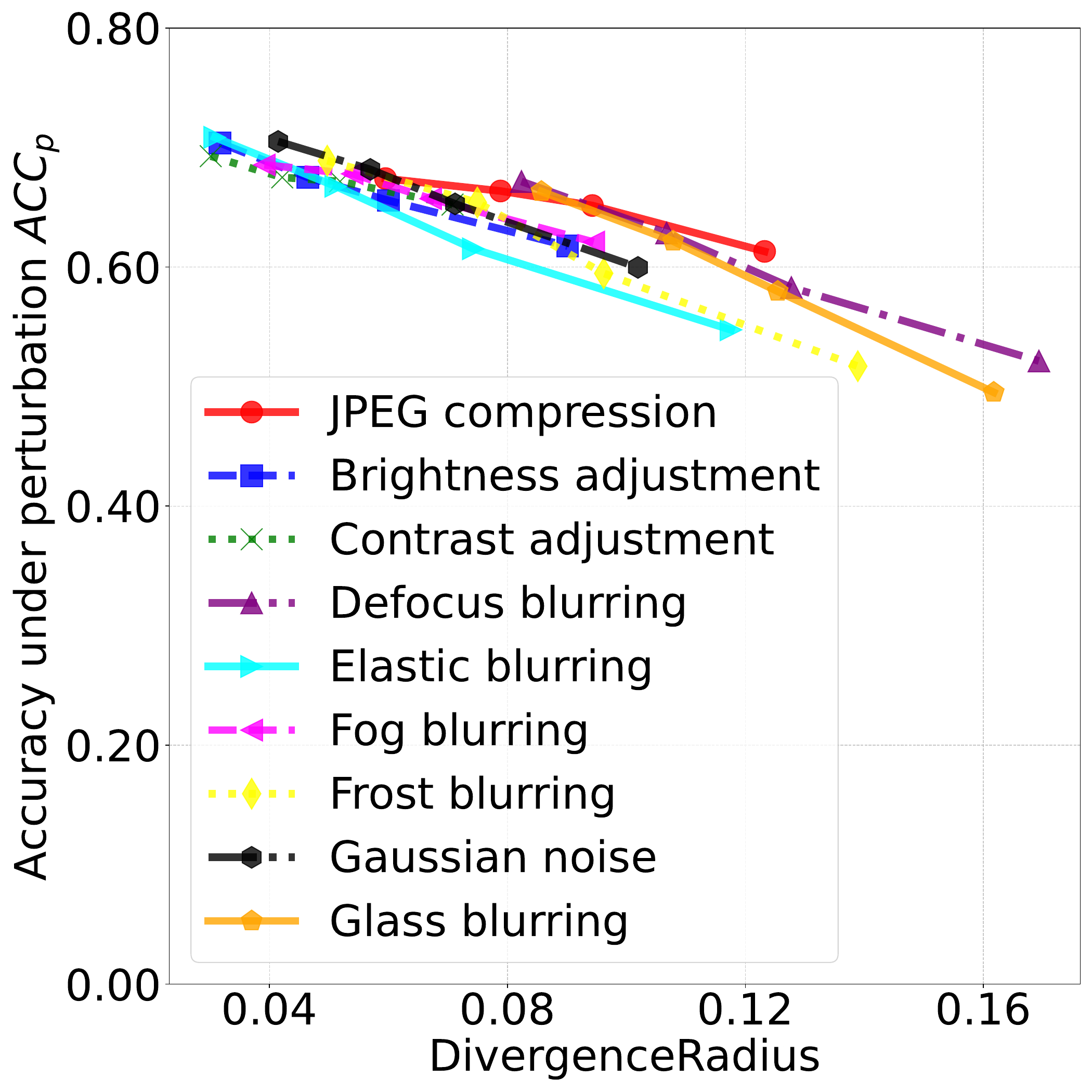}\label{fig:vision_downstream_zero_shot_imagenet}}
    \quad \quad \quad
    \subfloat[Linear-probe classification]{\includegraphics[width=0.25\textwidth]{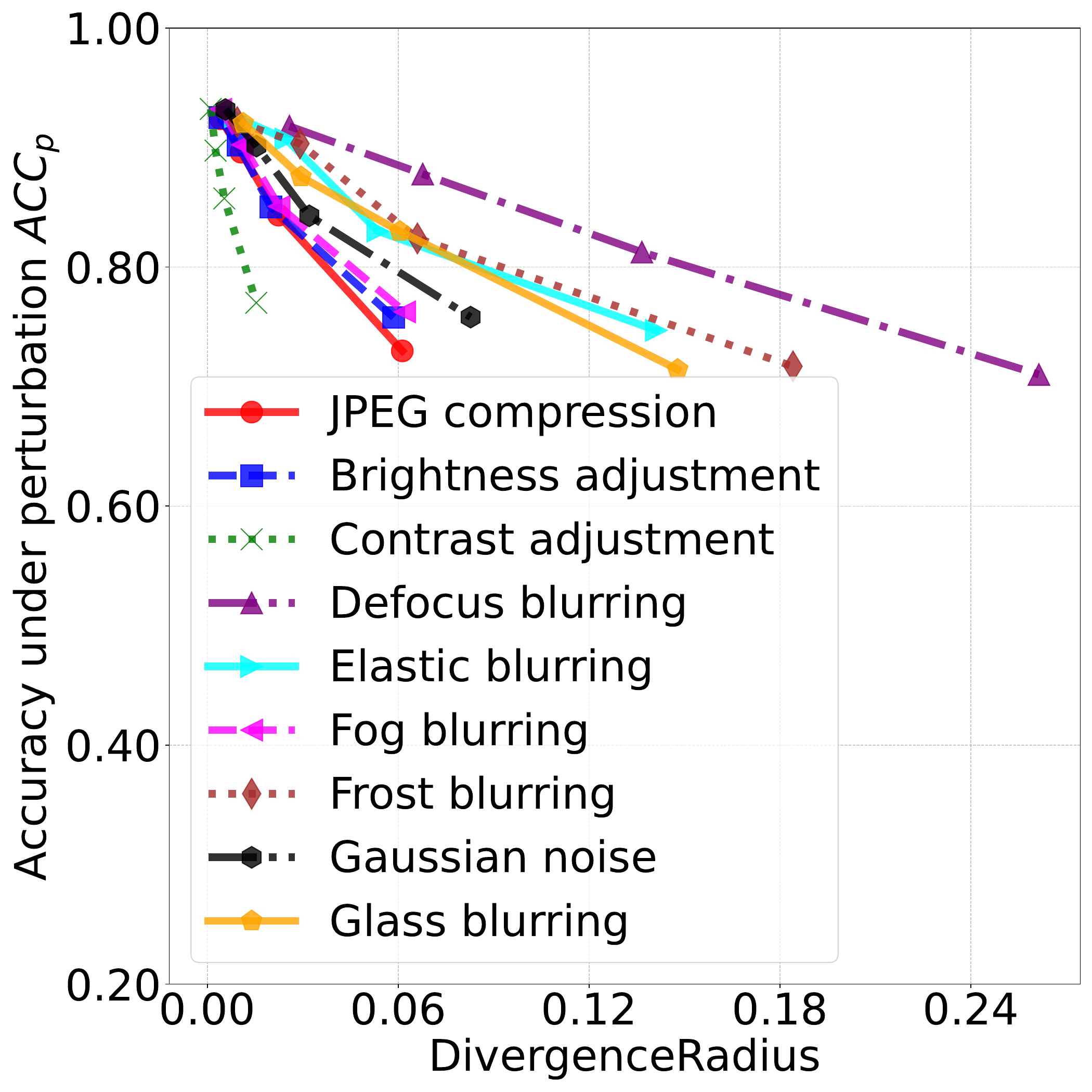}\label{fig:vision_downstream_linear_probe_imagenet}}
    \caption{Accuracy under perturbation $ACC_p$ vs.~\alg{}  of ImageNet testing images for (a) zero-shot classification and (b) linear-probe classification when different perturbation functions are used.   Zero-shot classification is based on the CLIP ViT-L/14 foundation model and linear-probe classification is based on the DINO v2 ViT-g/14 foundation model.} 
    \label{fig:vision_downstream_imagenet}
    \vspace{-4mm}
\end{figure*}

\begin{table}[t!]
\centering
\fontsize{7}{11}\selectfont 
\caption{ $ACC$ and average $ACC_p$ of ImageNet's testing images for two downstream classifiers. Zero-shot classification is based on the CLIP ViT-L/14 foundation model and linear-probe classification is based on the DINO v2 ViT-g/14 foundation model.}
\label{tab:acc_p_imagenet}
\begin{tabular}{|c|c|c|c|}
\hline
\multicolumn{2}{|c|}{ }  &    \makecell{Zero-shot\\Classification} &   \makecell{Linear-probe\\Classification} \\ \hline
\multicolumn{2}{|c|}{$ACC$ (\%)} & 68.4    &     86.6 \\ \hline
\multirow{9}{*}{\makecell{$ACC_p$\\(\%)}} &  JPEG compression    & 65.1 ($\downarrow$ 3.3)    &     84.6 ($\downarrow$ 2.0)   \\ \cline{2-4}
&  Brightness adjustment  &  66.3 ($\downarrow$ 2.1)    &  85.7 ($\downarrow$ 0.9)  \\ \cline{2-4}
&  Contrast adjustment  &  67.3 ($\downarrow$ 1.1)   &  86.4 ($\downarrow$ 0.2)  \\ \cline{2-4}
&  Defocus blurring  &  60.1 ($\downarrow$ 8.3)   &  82.2 ($\downarrow$ 4.4)  \\ \cline{2-4}
&  Elastic blurring  &  63.5 ($\downarrow$ 4.9)   &  85.0 ($\downarrow$ 1.6)  \\ \cline{2-4}
&  Fog blurring  &  66.0 ($\downarrow$ 2.4)   &  86.1 ($\downarrow$ 0.5)  \\ \cline{2-4}
&  Frost blurring  &  61.4 ($\downarrow$ 7.0)   &  83.6 ($\downarrow$ 3.0)  \\ \cline{2-4}
&  Gaussian  noise  &  66.0 ($\downarrow$ 2.4)   &  85.7 ($\downarrow$ 0.9)  \\ \cline{2-4}
&  Glass blurring  &  59.0 ($\downarrow$ 9.4)   &  82.9 ($\downarrow$ 3.7)  \\ \hline
\end{tabular}
\vspace{-2mm}
\end{table}

\myparatight{Evaluation metrics ($ACC$, $ACC_p$, {$RMSE$, and $RMSE_p$})} For a downstream classifier $g \circ f$, where $f$ is a foundation model and $g$ a classifier head, \emph{accuracy} ($ACC$) is the fraction of correctly predicted labels. \emph{Accuracy under perturbation} ($ACC_p$) is calculated for each perturbed image. For depth estimation heads, we use $RMSE$ and $RMSE_p$ (definitions in the Appendix).

\myparatight{Parameter settings} We evaluate the classifiers with the highest accuracy on zero-shot or linear-probe tasks. For zero-shot classification, we use CLIP ViT-L/14; for linear-probing, we use DINO v2 ViT-g/14, training a one-layer classifier for Food101.

\subsection{Experimental Results}
Due to space limits, we discuss results for downstream classifiers and defer depth estimation results to the Appendix.

\myparatight{$ACC$ vs. $ACC_p$}
Table~\ref{tab:acc_p_imagenet} shows $ACC$ and average $ACC_p$ under different perturbations for ImageNet; results for Food101 are in Table~\ref{tab:acc_p_food101}. Common perturbations reduce $ACC_p$ compared to $ACC$, indicating degraded accuracy. For example, Glass blurring reduces ImageNet zero-shot accuracy by 9.4\%, and Frost blurring reduces Food101 linear-probe accuracy by 3.3\%. The accuracy drop correlates with the average robustness value across perturbations. For instance, in Figure~\ref{fig:divergence_radius_vision_imagenet}, both CLIP ViT-L/14 and DINO v2 ViT-g/14 exhibit the highest \alg{} under Defocus blurring, corresponding to the largest $ACC_p$ drop in Table~\ref{tab:acc_p_imagenet}. This suggests that higher \alg{} values indicate greater embedding diversity, increasing the chance of misclassification.

\begin{table}[t!]
\centering
\fontsize{8pt}{10pt}\selectfont
\caption{$ACC$ and average $ACC_p$ of Food101's testing images for two downstream classifiers. Zero-shot classification is based on the CLIP ViT-L/14 foundation model and linear-probe classification is based on the DINO v2 ViT-g/14 foundation model.}
\label{tab:acc_p_food101}
\begin{tabular}{|c|c|c|c|}
\hline
\multicolumn{2}{|c|}{ }  &    \makecell{Zero-shot\\Classification} &   \makecell{Linear-probe\\Classification} \\ \hline
\multicolumn{2}{|c|}{$ACC$ (\%)} &  92.0   &   94.1   \\ \hline
\multirow{9}{*}{\makecell{$ACC_p$\\(\%)}} &  JPEG compression    &  87.4 ($\downarrow$ 4.6)    &  91.2 ($\downarrow$ 2.9)   \\ \cline{2-4}
&  Brightness adjustment  &   88.4 ($\downarrow$ 3.6)    &  90.8 ($\downarrow$ 3.3)  \\ \cline{2-4}
&  Contrast adjustment  &   90.9 ($\downarrow$ 1.1)   &  91.8 ($\downarrow$ 2.3)  \\ \cline{2-4}
&  Defocus blurring  &   81.7 ($\downarrow$ 10.3)   &  87.0 ($\downarrow$ 7.1)  \\ \cline{2-4}
&  Elastic blurring  &   87.2 ($\downarrow$ 4.8)   &  89.6 ($\downarrow$ 4.5)  \\ \cline{2-4}
&  Fog  blurring&   88.3 ($\downarrow$ 3.7)   &  90.9 ($\downarrow$ 3.2)  \\ \cline{2-4}
&  Frost  blurring&   79.3 ($\downarrow$ 12.7)   &  85.0 ($\downarrow$ 9.1)  \\ \cline{2-4}
&  Gaussian noise  &   86.6 ($\downarrow$ 5.4)   &  89.7 ($\downarrow$ 4.4)  \\ \cline{2-4}
&  Glass blurring  &   82.9 ($\downarrow$ 9.1)   &  87.3 ($\downarrow$ 6.8)  \\ \hline
\end{tabular}
\vspace{-2mm}
\end{table}

\myparatight{$ACC_p$ vs. robustness value} Figure~\ref{fig:vision_downstream_imagenet} shows the relationship between $ACC_p$ and \alg{} for ImageNet images under various perturbations; results for Food101 are in Figure~\ref{fig:vision_downstream_food101}. Across datasets and perturbations, $ACC_p$ decreases approximately linearly as \alg{} or cosine similarity increases, indicating that greater embedding diversity leads to lower accuracy.

\begin{figure}[t!]
    \centering
    \subfloat[Zero-shot classification]{\includegraphics[width=0.24\textwidth]{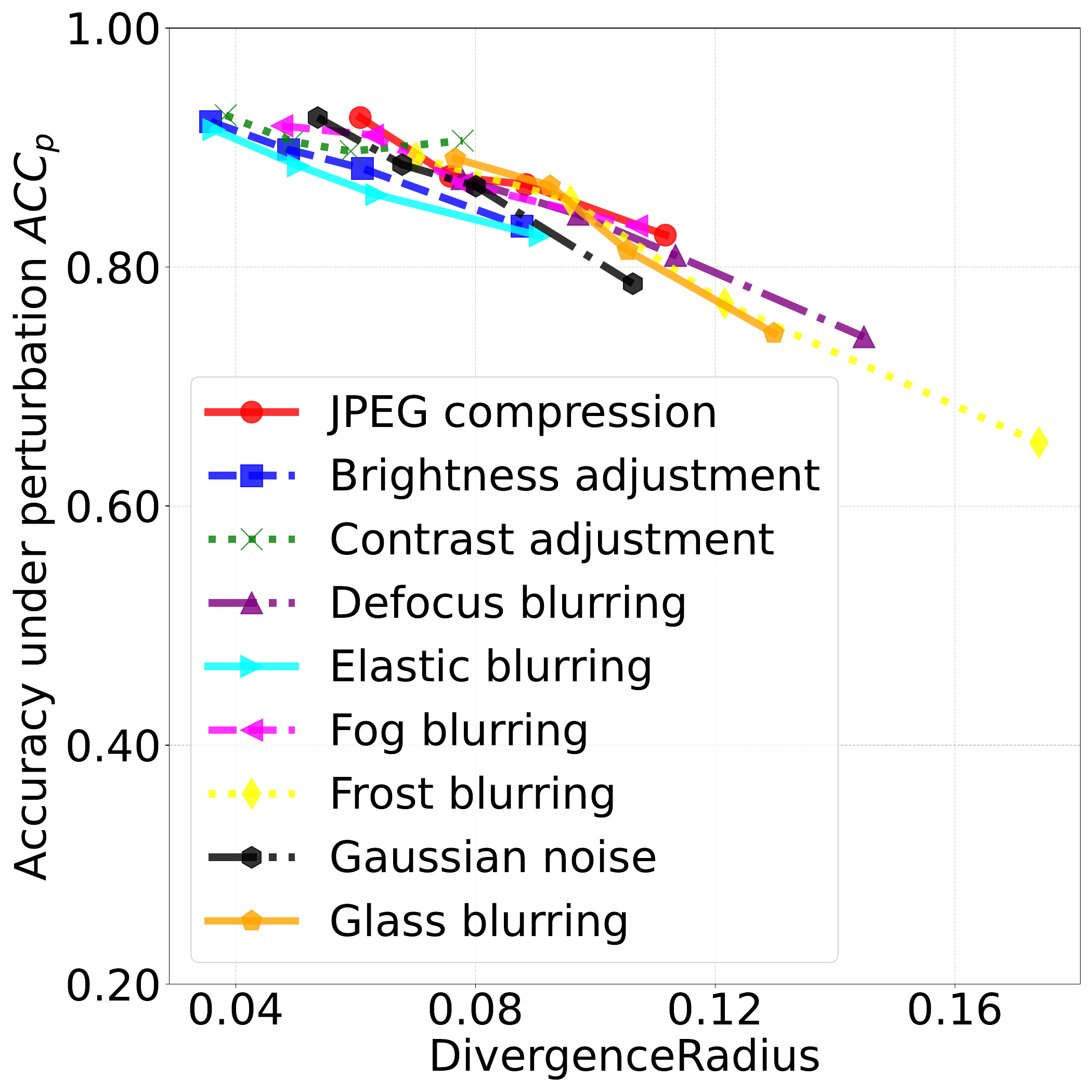} \label{fig:vision_downstream_zero_shot_food101}}
    \subfloat[Linear-probe classification]{\includegraphics[width=0.24\textwidth]{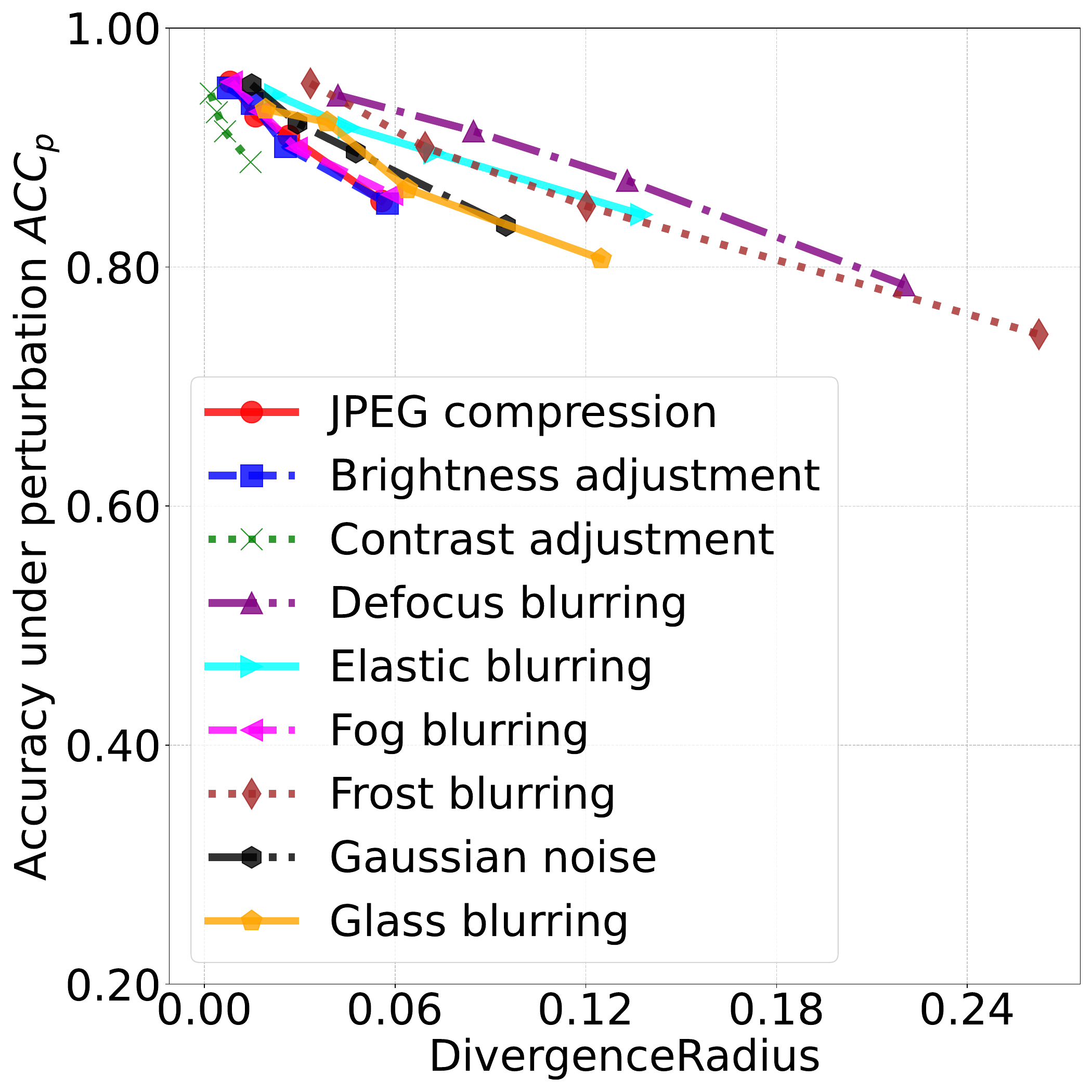}\label{fig:vision_downstream_linear_probe_food101}}
    \caption{Accuracy under perturbation $ACC_p$ vs.~\alg{}  of Food101 testing images for (a) zero-shot classification and (b) linear-probe classification when different perturbation functions are used.   Zero-shot classification is based on the CLIP ViT-L/14 foundation model and linear-probe classification is based on the DINO v2 ViT-g/14 foundation model.}
    \label{fig:vision_downstream_food101}
    \vspace{-2mm}
\end{figure}

\myparatight{Predicting $ACC_p$ using robustness values} The linear trend between $ACC_p$ and \alg{} (or cosine similarity) enables accurate $ACC_p$ prediction via linear regression. We divide the dataset, train a linear model on the first half (predicting $ACC_p$ from \alg{} or cosine similarity), and test it on the second half. Figures~\ref{fig:mse_downstream_prediction_imagenet} and~\ref{fig:mse_downstream_prediction_food101} in the Appendix show low mean squared errors, indicating that \alg{} or cosine similarity can reliably predict downstream accuracy under perturbations.

%% file: 4.4_exp_enhancement.tex
\section{Robustness Enhancement}
\subsection{Method}

\myparatight{Robustness and utility goals} We propose a fine-tuning method to enhance a foundation model's robustness against common perturbations. Given a foundation model $f$, our goal is to produce a model $f'$ that meets both a \emph{robustness goal} (increased robustness to perturbations) and a \emph{utility goal} (maintaining performance on unperturbed images).

\myparatight{Formulating an optimization problem} We define two loss terms for robustness and utility. Minimizing these terms involves optimizing a weighted sum. For a set of unlabeled images $\mathcal{D}$, we quantify robustness by the cosine similarity between an image’s embedding and those of its perturbed versions. Specifically, the robustness loss term $\mathcal{L}_1$ is: $\mathcal{L}_1 = -\frac{1}{\|\mathcal{D}\|} \sum\limits_{x \in \mathcal{D}} \operatorname{cos}\left(f'(x), f'(P(x,k))\right)$, where $P$ is a perturbation function and $k$ is sampled from $\mathbb{K}$ each epoch. To ensure utility, we use cosine similarity between embeddings of unperturbed images from $f$ and $f'$. The utility loss term $\mathcal{L}_2$ is: $\mathcal{L}_2 = -\frac{1}{\|\mathcal{D}\|} \sum\limits_{x \in \mathcal{D}} \operatorname{cos}\left(f(x), f'(x)\right)$.

The optimization problem then minimizes the weighted sum of these terms: $\min_{f'} \mathcal{L}_1 + \lambda \mathcal{L}_2$, where $\lambda$ balances robustness and utility.

\myparatight{Solving the optimization problem} We use a gradient-based method, initializing $f'$ as $f$, and iteratively updating $f'$ using the gradient computed over mini-batches from $\mathcal{D}$.

\subsection{Experimental Setup}

We use CLIP ViT-L/14 for zero-shot classification. Fine-tuning data $\mathcal{D}$ includes 20,000 randomly selected ImageNet training images, with JPEG compression as the default perturbation. Models are fine-tuned for 50 epochs with a learning rate of $1\times10^{-5}$ and $\lambda = 1$ unless otherwise stated.

\subsection{Experimental Results}

\myparatight{Achieving the robustness goal} Figures~\ref{fig:enhancement_cosine_similarity_clip} and~\ref{fig:enhancement_divergence_radius_clip} show the average cosine similarity and \alg{} for ImageNet/Food101 testing images before and after fine-tuning. We observe decreases in both metrics, indicating that embedding vectors of perturbed images become closer to those of unperturbed images, thus enhancing robustness.

\begin{figure}[!h]
    \centering
    \subfloat[Cosine similarity]{\label{fig:enhancement_cosine_similarity_clip}\includegraphics[width=0.2\textwidth]{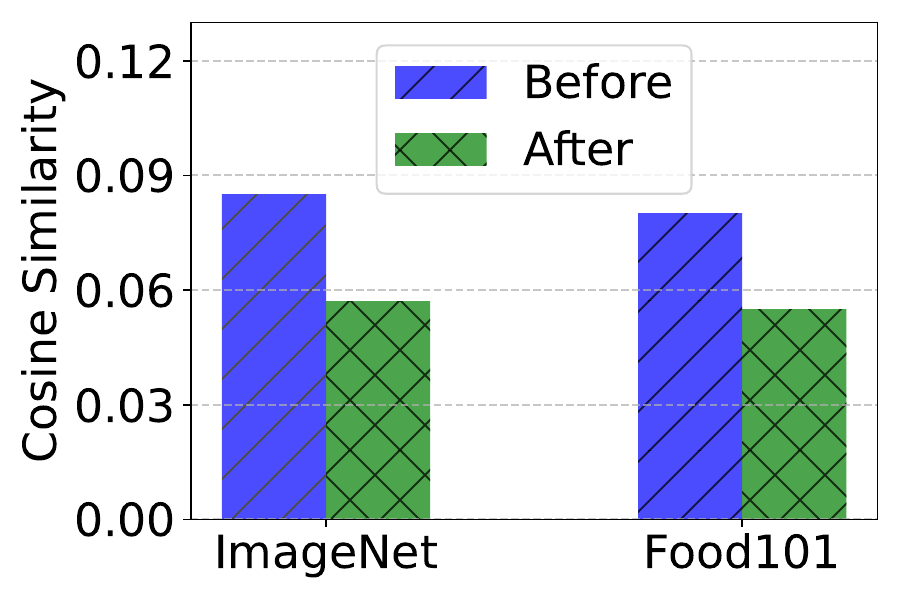}}
    \\
    \subfloat[\alg{}]{\label{fig:enhancement_divergence_radius_clip}\includegraphics[width=0.2\textwidth]{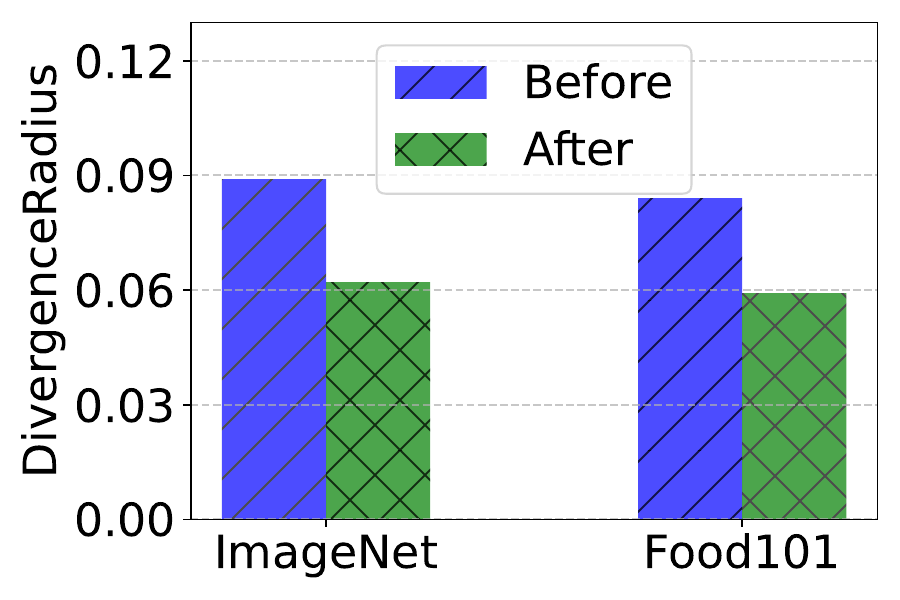}}
    \\
    \subfloat[$ACC$]{\label{fig:enhancement_acc_clip}\includegraphics[width=0.2\textwidth]{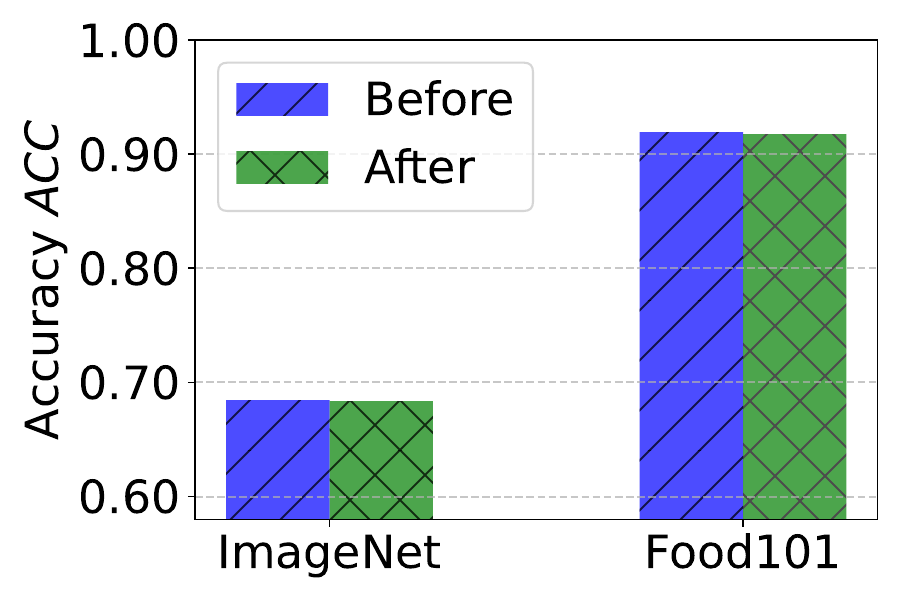}}
    \caption{(a) Average cosine similarity, (b) average~\alg{}, and (c) $ACC$ of zero-shot classification for the two datasets before and after robustness enhancement, where the foundation model is CLIP ViT-L/14.}
    \label{fig:enhancement_results_clip}
\end{figure}

\myparatight{Achieving the utility goal} Figure~\ref{fig:enhancement_acc_clip} shows that zero-shot classification accuracy $ACC$ remains nearly unchanged on unperturbed images after fine-tuning, indicating that the utility goal is met due to the inclusion of $\mathcal{L}_2$.

\myparatight{Impact of $\lambda$} Table~\ref{tab:enhancement_impact_lambda} in the Appendix shows how $\lambda$ affects the robustness-utility trade-off. When $\lambda$ is small (e.g., 0), robustness is achieved but utility declines, whereas large $\lambda$ (e.g., 5) favors utility over robustness.

%% file: 5_related.tex
\section{\label{sec:preliminaries}Related Work}

\myparatight{Foundation models} 
Foundation models~\citep{radford2021learning,oquab2023dinov2,chen2020simple,li2022blip,he2020momentum} are pre-trained neural networks used as general-purpose feature extractors, often for vision tasks. Vision foundation models are typically pre-trained on large datasets of unlabeled images~\citep{oquab2023dinov2,chen2020simple} or image-text pairs~\citep{radford2021learning}. For example, Meta’s DINO v2~\citep{oquab2023dinov2} is trained on 142 million images, while OpenAI’s CLIP~\citep{radford2021learning} is trained on 400 million image-text pairs.

\myparatight{Common perturbations} 
Common perturbations frequently arise in real-world, non-adversarial settings. While adversarial robustness of foundation models has been widely studied, robustness to common perturbations is less explored. Zhu et al.~\citep{zhu2023promptbench} investigated language models’ robustness to prompt perturbations but did not focus on vision models. Hendrycks and Dietterich~\citep{hendrycks2019benchmarking} assessed classifier robustness to common perturbations, and studies~\citep{bai2021transformers,bhojanapalli2021understanding,wang2022can,paul2022vision} compared robustness of vision transformers and CNNs, though they focused on classifiers, not foundation models. Hendrycks et al.~\citep{hendrycks2021many} found that common perturbations could improve out-of-distribution robustness for classifiers, while our work centers on in-distribution robustness.

%% file: 6_conclusion.tex
\section{\label{sec:conclusion}Conclusion and Future Work}

In this work, we introduce three metrics—cosine similarity, Euclidean distance, and \alg{}—to quantify foundation model robustness to common perturbations. Theoretically, we showed that \alg{} meets all five desired mathematical properties, while cosine similarity and Euclidean distance do not satisfy the worst-robustness property. Using these metrics, we empirically evaluated the robustness of industry-scale foundation models and downstream applications, finding limited robustness to common perturbations, which impacts downstream performance. We also demonstrated that fine-tuning models to align perturbed embeddings with the original can enhance robustness without affecting utility. Future work includes extending these analyses to language models and investigating robustness against adversarial perturbations.

%% file: 7_appendix.tex
\clearpage
\appendix

\section{Impact Statements}

This work strengthens the reliability of vision foundation models by systematically evaluating their robustness to common perturbations and proposing fine-tuning strategies to enhance stability without sacrificing utility. By introducing principled robustness metrics and analyzing industry-scale models, our findings highlight vulnerabilities that can impact real-world applications. Addressing these weaknesses improves the robustness of AI systems in practical settings, ensuring consistent performance across diverse conditions. Our work provides a foundation for future research on enhancing model robustness, including extending robustness analyses to language models and defending against adversarial perturbations.

\section{\label{sec:proof_theorem_equivalence}Proof of Theorem~\ref{theorem_equivalence}}

\begin{proof}
Since $||f(x)||_2=1$ for any $x$, we have the following:
\begin{align}
    &||f(P(x,k_1)) - f(P(x,k_2))||_2^2 \nonumber \\
    =&||f(P(x,k_1))||_2^2 + ||f(P(x,k_2))||_2^2 \nonumber\\&- 2 f(P(x,k_1))^T\cdot f(P(x,k_2)) \nonumber \\
    =&2 - 2\cdot cos (f(P(x,k_1)), f(P(x,k_2))),
\end{align}
where $T$ represents transpose and $f(P(x,k_1))^T\cdot f(P(x,k_2))$ is the inner product between two embedding vectors. Therefore, we have:
\begin{align}
    &\mathcal{R}_{ed}(f,x,P,\mathbb{K}) \nonumber \\=& \frac{\max_{k_1,k_2 \in \mathbb{K}} ||f(P(x,k_1)) - f(P(x,k_2))||_2}{2} \nonumber \\
    =&\frac{\max_{k_1,k_2 \in \mathbb{K}} \sqrt{2 - 2 cos (f(P(x,k_1)), f(P(x,k_2)))}}{2} \nonumber \\
    =& \sqrt{\frac{1 - \min_{k_1,k_2 \in \mathbb{K}} cos(f(P(x,k_1)), f(P(x,k_2)))}{2}} \nonumber \\
    =& \sqrt{\mathcal{R}_{cs}(f,x,P,\mathbb{K})}.
\end{align}
Therefore, by Theorem~\ref{theorem_imposibility_cosine_function}, the Euclidean distance-based metric $\mathcal{R}_{ed}$ does not satisfy the worst-robustness property. 
\end{proof}

\section{\label{sec:proof_our_alg}\alg{} Satisfies the Mathematical Properties}
\begin{proof}
\myparatight{Bounded domain} The radius of any ball is non-negative, and thus $\mathcal{R}_{dr}(f,x,P,\mathbb{K})\geq 0$. Moreover, since all embedding vectors outputted by a foundation model lie on  the unit hyper-sphere, the unit ball with radius 1 can enclose all embedding vectors. Therefore, we have $\mathcal{R}_{dr}(f,x,P,\mathbb{K})\leq 1$. Thus, we have $\mathcal{R}_{dr}(f,x,P,\mathbb{K})\in  [0, 1]$.

\myparatight{Monotonicity}  Suppose the perturbation parameter domain $\mathbb{K}_1$ expands to $\mathbb{K}_2$.  If the embedding vectors of the perturbed images corresponding to the expanded perturbation parameters (i.e., parameters in $\mathbb{K}_2$-$\mathbb{K}_1$) fall outside of the minimum enclosing ball for $\mathbb{K}_1$, the minimum enclosing ball for $\mathbb{K}_2$ expands to have a larger radius, i.e., $\mathcal{R}_{dr}(f,x,P,\mathbb{K}_1) < \mathcal{R}_{dr}(f,x,P,\mathbb{K}_2)$; otherwise, the minimum enclosing ball remains the same, i.e., $\mathcal{R}_{dr}(f,x,P,\mathbb{K}_1) = \mathcal{R}_{dr}(f,x,P,\mathbb{K}_2)$. Therefore, we have $\mathcal{R}_{dr}(f,x,P,\mathbb{K}_1) \leq \mathcal{R}_{dr}(f,x,P,\mathbb{K}_2)$ if $\mathbb{K}_1 \subseteq \mathbb{K}_2$, which satisfies the monotonicity property.

\myparatight{Best robustness} If all the perturbed versions of the image have the same embedding vector, the minimum enclosing ball has a radius 0. Thus, we have $\mathcal{R}_{dr}(f,x,P,\mathbb{K})=0$ in such case, achieving the best-robustness property.

\myparatight{Worst robustness} In the worst-robustness case, there exists a subdomain $\mathbb{K}' \subseteq \mathbb{K}$, where $\mathbb{K}'$ consists of discrete values and the corresponding embedding vectors are equally distributed in the embedding space, i.e., $\sum_{k \in \mathbb{K}'} f(P(x,k)) = 0$. Without loss of generality, we assume the subdomain $\mathbb{K}^{\prime}$ contains $n$ discrete values $k_1, \cdots, k_n$. Then, we have the following:
\begin{align}
\sum_{i=1}^{n} f(P(x,k_i)) = \mathbf{0}.
\end{align}
Based on Equation~\ref{divergence_radius}, we have the following equation group:
\begin{align}
\label{equation_group_worst_robustness}
    \left\{
        \begin{array}{l}
            ||f(P(x, k_1))||_2^2 - 2f^T(P(x, k_1))\cdot c + ||c||_2^2 \leq r^2 \\
            ||f(P(x, k_2))||_2^2 - 2f^T(P(x, k_2))\cdot c + ||c||_2^2 \leq r^2  \\
            \cdots \\
            ||f(P(x, k_n))||_2^2 - 2f^T(P(x, k_n))\cdot c + ||c||_2^2 \leq r^2 
        \end{array}
    \right., 
\end{align}
where $T$ indicates transpose of a vector.  
Since the embedding vectors lie on the unit hyper-sphere, we have $||f(P(x,k_i))||_2^2=1$ for $i=1,\cdots,n$. After summing up the $n$ inequalities in the equation group~\ref{equation_group_worst_robustness}, we have the following:
\begin{align}
& n \cdot 1 -2 (\sum_{i=1}^n f^T(P(x,k_i)) )\cdot c + n||c||_2^2 \leq n r^2, \nonumber\\
\Leftrightarrow & r^2 \geq 1+||c||_2^2 \geq 1, 
\end{align}
where $r=1$ when $c=\mathbf{0}$. Since  $\mathcal{R}_{dr}(f,x,P,\mathbb{K})$ is the smallest $r$, we have $\mathcal{R}_{dr}(f,x,P,\mathbb{K})=1$, achieving the worst-robustness property.

\myparatight{Rotational invariance} Suppose the rotation matrix is $M$.  We observe that $||M \cdot f(P(x,k))  - M\cdot c||_2 = ||M \cdot (f(P(x,k))  - c)||_2 = ||f(P(x,k))  - c||_2 \leq \mathcal{R}_{dr}(f,x,P,\mathbb{K})$ for any $k \in \mathbb{K}$. Therefore, we have $\mathcal{R}_{dr}(M\cdot f,x,P,\mathbb{K}) \leq \mathcal{R}_{dr}(f,x,P,\mathbb{K})$.  We denote the inverse matrix of the rotation matrix as $M^{-1}$.  We have $\mathcal{R}_{dr}(f,x,P,\mathbb{K}) = \mathcal{R}_{dr}(M^{-1} \cdot M\cdot f,x,P,\mathbb{K}) \leq \mathcal{R}_{dr}(M \cdot f,x,P,\mathbb{K})$. Thus, we have $\mathcal{R}_{dr}(M\cdot f,x,P,\mathbb{K}) = \mathcal{R}_{dr}(f,x,P,\mathbb{K})$,  achieving the rotation-invariance property. Essentially,  the embedding vectors are enclosed in a minimum ball with center $c$ and radius $\mathcal{R}_{dr}(f,x,P,\mathbb{K})$ before rotation, and the embedding vectors are enclosed in a minimum ball with center $Mc$ and the same radius after rotation.  
\end{proof}

\begin{table}[t!]
\fontsize{8}{11}\selectfont 
\caption{Benchmark datasets.}
\label{tab:dataset}
\centering
\begin{tabular}{|c|c|c|c|}
\hline
Dataset & ImageNet & Food101 & NYU-Depth V2\\ \hline
\makecell{\#Training\\images} & 1,281,167 & 75,750 &  50,688\\ \hline
\makecell{\#Testing\\images} & 50,000 & 25,250 & 654 \\ \hline
\#Classes & 1,000 & 101 & - \\ \hline
\end{tabular}
\end{table}

\begin{table*}[!h]
\small
\centering
\caption{Foundation models.}
\label{tab:models}
\begin{tabular}{|c|c|c|c|c|}
\hline
Foundation Model & Model Family & Pre-training Algorithm & Architecture & \makecell{\# Parameters (M)} \\ 
\hline
CLIP ViT-B/16 & CLIP & Multi-modal self-supervised learning & Vision Transformer & 86 \\ \hline
CLIP ViT-L/14 & CLIP & Multi-modal self-supervised learning & Vision Transformer & 304 \\ \hline
CLIP RN50 & CLIP & Multi-modal self-supervised learning & ResNet & 38 \\ \hline
CLIP RN50$\times$64 & CLIP & Multi-modal self-supervised learning & ResNet & 420 \\ \hline
DINO v2 ViT-L/14 & DINO v2 & Self-supervised learning & Vision Transformer & 304 \\ \hline
DINO v2 ViT-g/14 & DINO v2 & Self-supervised learning & Vision Transformer & 1136 \\ \hline
\end{tabular}
\end{table*}

\section{\label{sec:applications}Details of Zero-shot Classification, Linear-probe Classification, and Depth Estimation}
\begin{itemize}
\item {\bf Zero-shot classification:} Jointly pre-trained image and text models like CLIP can perform zero-shot classification without labeled training data. Given an image $x$ and a set of text labels $\mathcal{Y}$, CLIP represents $x$ as an image embedding and the labels in $\mathcal{Y}$ as text embeddings, selecting the label most similar to $x$'s embedding. For instance, CLIP’s ViT-L/14 achieves 68.38\% accuracy on zero-shot ImageNet classification.

\item {\bf Linear-probe classification:} Vision foundation models can be adapted for linear-probe classification using labeled data to fine-tune a simple classification head. With foundation model parameters frozen, a feed-forward classification head is added and trained on labeled data. This approach achieves high accuracy with minimal additional parameters; for example, DINO-v2 reaches 86.6\% accuracy on ImageNet with a single-layer classification head.

\item {\bf Depth estimation:} Vision foundation models can also support depth estimation, predicting per-pixel depth in an image. A convolutional decoder head (with batch normalization, a $1\times1$ convolution, ReLU activation, and a sigmoid function) is appended to the frozen foundation model and fine-tuned on NYU-Depth V2 training data. For example, DINO-v2 achieves a root mean squared error of 0.35 on NYU-Depth V2.
\end{itemize}

\section{\label{sec:metrics_downstream_depth_estimation}Evaluation Metrics $RMSE$ and $RMSE_p$}
Given a downstream depth estimation model $g \circ f$, where $f$ is a foundation model and $g$ is a depth estimation head built on top of $f$. The  \emph{root mean squared error} (denoted as $RMSE$) for a testing dataset is the average root mean squared error between model's predicted depth map and the ground-truth depth map. We evaluate common perturbations to images, and thus we also consider \emph{root mean squared error under perturbation} (denoted as $RMSE_p$) for each testing image. 
Specifically, given a testing image $x$ with ground-truth depth map $y$, we evaluate the root mean squared error of a downstream depth estimation model $g \circ f$ for $x$ under a perturbation function $P$.  Without loss of generality, we assume the perturbation parameter domain $\mathbb{K}$ contains $m$ discrete values, where $m$ discrete values can be sampled via random sampling or equally-spaced sampling if $\mathbb{K}$ contains more than $m$ discrete values or is continuous. Formally, the $RMSE$ of the downstream classifier $g \circ f$ for $x$ under a perturbation function $P$ with a perturbation parameter domain $\mathbb{K}$ is defined as below:
\begin{align}
&RMSE_p(x, g, f, P, \mathbb{K}) \nonumber \\=& \frac{1}{m} \sum_{k\in\mathbb{K}} \sqrt{\frac{1}{n}\sum_{i=1}^{n}(y_{i} - g \circ f (P(x,k))_{i})^{2}},
\end{align}
where $g \circ f (P(x,k))$ is the predicted depth map for image $P(x,k)$ and $i$ is the index of each pixel in a depth map. Intuitively, the smaller $RMSE$ indicates the more accurate depth estimation.

\section{\label{sec:results_downstream_depth_estimation}Parameter Settings and Experimental Results of Depth Estimation}
\myparatight{Parameter settings}
We consider the downstream depth estimation model that achieves the lowest root mean squared error. Specifically, we use DINO v2 ViT-g/14 foundation model and the one-layer depth estimation head publicly released together with DINO v2 for NYU-Depth V2.

\myparatight{Root mean squared error $RMSE$ vs. root mean squared error under perturbation $RMSE_p$}
Table~\ref{tab:rmse_p_nyu} shows the $RMSE$ and average $RMSE_p$ under different perturbation functions of NYU-Depth V2 testing images for depth estimation model. First, we find that common perturbations degrade RMSE of downstream depth estimation models, i.e., average $RMSE_p$ is larger than $RMSE$. For instance, Frost blurring increases the $RMSE$ of depth estimation by 0.12.

Second, the increased $RMSE$ caused by a perturbation function is aligned with the average~\alg{}  under the perturbation function. For example, in Figure~\ref{fig:divergence_radius_vision_nyu}, DINO v2 ViT-g/14 have the largest average \alg{}  under Frost blurring compared to other perturbation functions. Correspondingly, in Table~\ref{tab:rmse_p_nyu},  average $RMSE_p$ under Frost blurring increases the most (0.12). This occurs because a higher average \alg{} indicates more diversity in the embedding vectors of perturbed images, making it more likely that the downstream depth estimation models predict inaccurate depth maps for them. 

\begin{table}[t!]
\centering
\fontsize{8pt}{10pt}\selectfont
\caption{$RMSE$ and average $RMSE_p$ of NYU-Depth V2's testing images for depth estimation when using DINO v2 ViT-g/14 foundation model.}
\label{tab:rmse_p_nyu}
\begin{tabular}{|c|c|c|}
\hline
\multicolumn{2}{|c|}{$RMSE$} & 0.35 \\ \hline
\multirow{9}{*}{\makecell{$RMSE_p$}} & JPEG compression & 0.37 ($\uparrow$ 0.02) \\ \cline{2-3}
& Brightness adjustment & 0.36 ($\uparrow$ 0.01) \\ \cline{2-3}
& Contrast adjustment & 0.35 ($\uparrow$ 0.00) \\ \cline{2-3}
& Defocus blurring & 0.39 ($\uparrow$ 0.04) \\ \cline{2-3}
& Elastic blurring & 0.38 ($\uparrow$ 0.03) \\ \cline{2-3}
& Fog blurring& 0.38 ($\uparrow$ 0.03) \\ \cline{2-3}
& Frost blurring& 0.47 ($\uparrow$ 0.12) \\ \cline{2-3}
& Gaussian noise & 0.37 ($\uparrow$ 0.02) \\ \cline{2-3}
& Glass blurring & 0.38 ($\uparrow$ 0.03) \\ \hline
\end{tabular}
\end{table}

\myparatight{Root mean squared error under perturbation $RMSE_p$ vs. robustness value} Figure~\ref{fig:vision_downstream_depth_estimation_nyud}  shows the relationship between an image's $RMSE_p$ under a perturbation function and the image's corresponding robustness value (cosine similarity or \alg{}) for the NYU-Depth V2 testing images and 9 perturbation functions. Specifically, given a perturbation function, for each testing image, we compute its robustness value and its $RMSE_p$ of a downstream depth estimation model, i.e., we obtain a pair (robustness value, $RMSE_p$). Then, we rank the pairs of all testing images in an increasing order according to the robustness values and divide the ranked pairs into 4 groups equally. For each group of pairs, we calculate the mean robustness value and mean $RMSE_p$, which are shown in Figure~\ref{fig:vision_downstream_depth_estimation_nyud}. Across perturbation functions, we find that $RMSE_p$ roughly increases \emph{linearly} when the robustness metrics increases. $RMSE_p$ increases as the robustness value increases because a larger robustness value indicates more diverse embedding vectors for the perturbed images, leading to less accurate depth estimation. 

\begin{figure}[!h]
    \centering
    \subfloat[Zero-shot classification]{\includegraphics[width=0.24\textwidth]{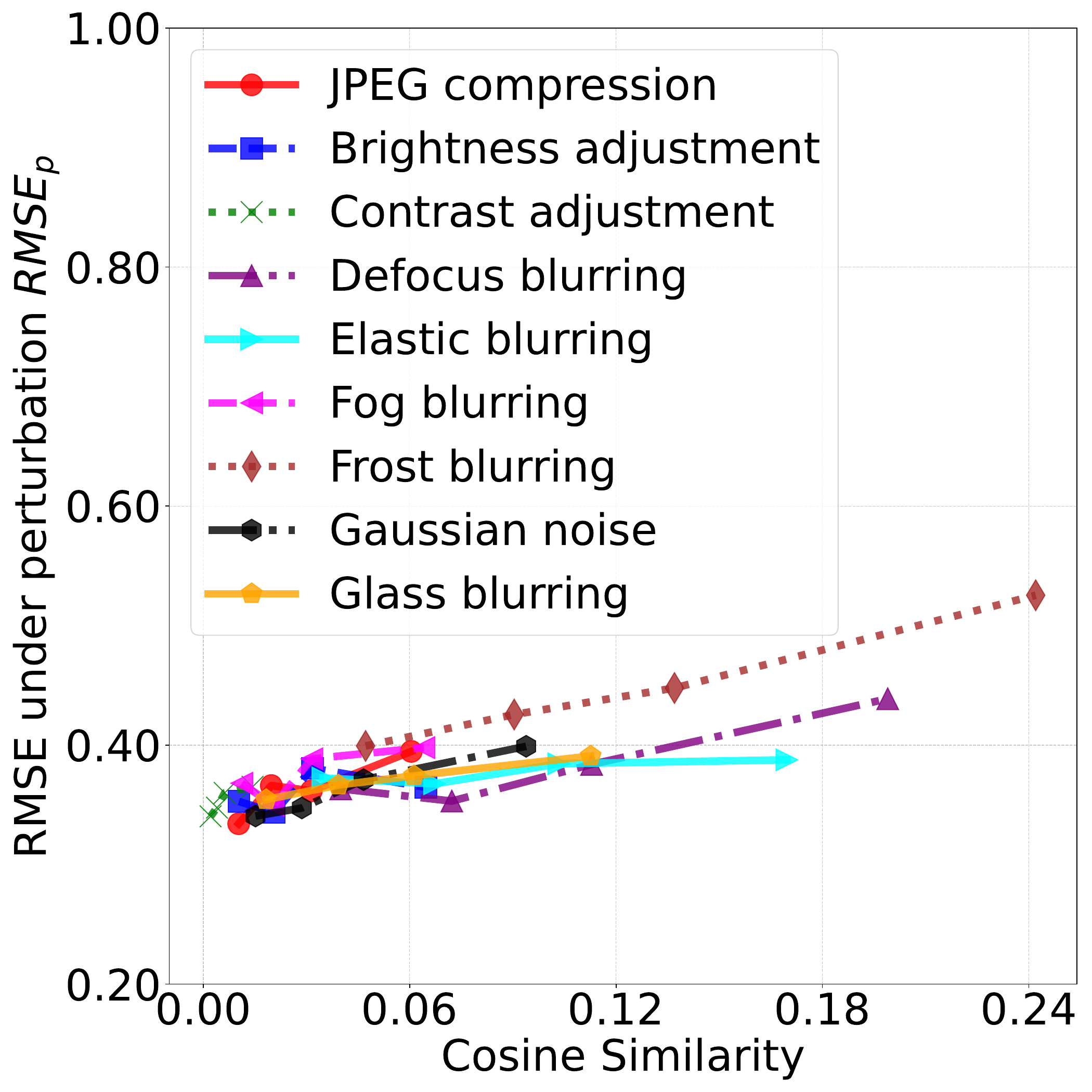} }
    \subfloat[Linear-probe classification]{\includegraphics[width=0.24\textwidth]{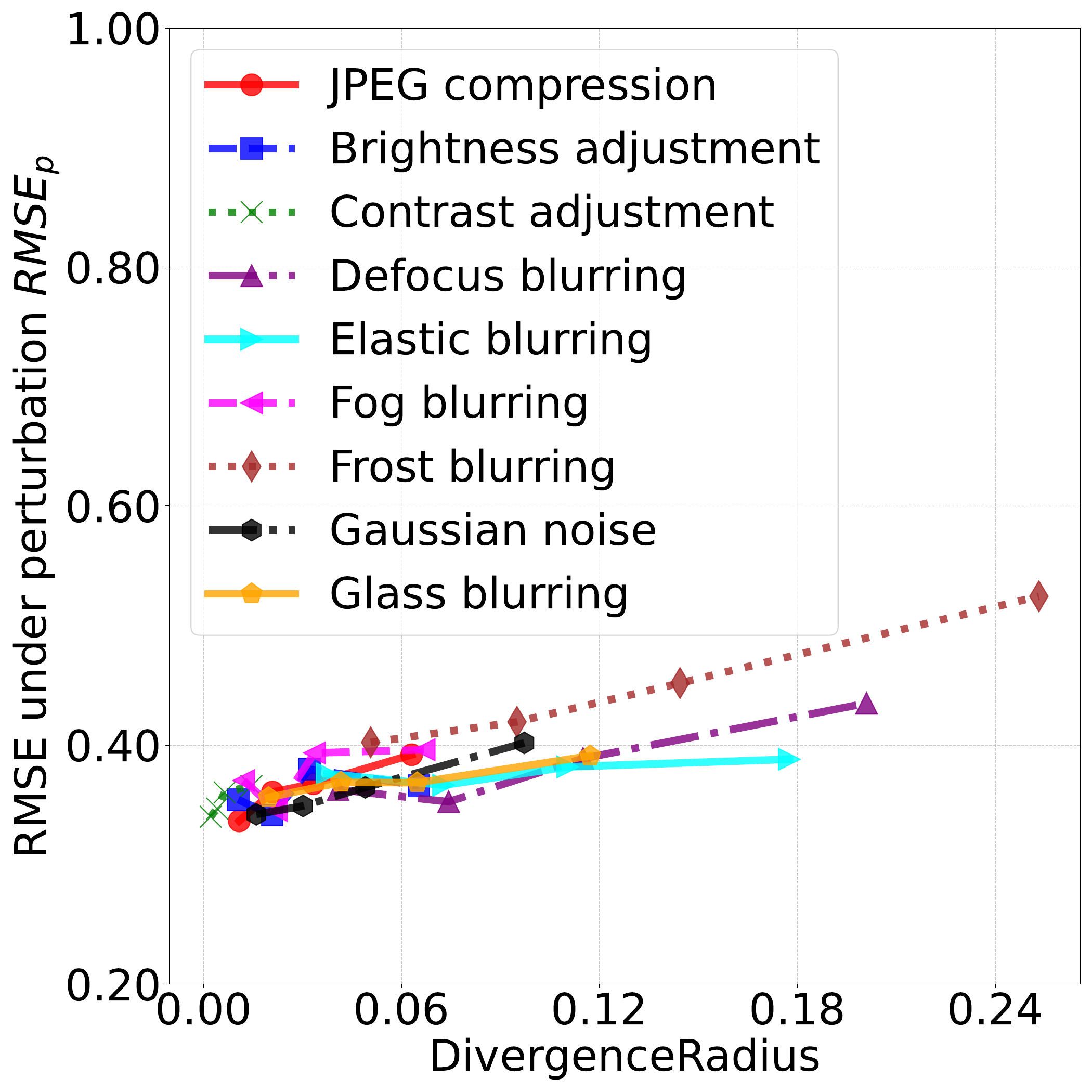}}
    \caption{Root mean squared error under perturbation $RMSE_p$ vs. (a) cosine similarity, (b) \alg{} of NYUd testing images for depth estimation when different perturbation functions are used. We use the DINO v2 ViT-g/14 foundation model.}
    \label{fig:vision_downstream_depth_estimation_nyud}
\end{figure}

\begin{figure}[!h]
    \centering
    \includegraphics[width=0.4\textwidth]{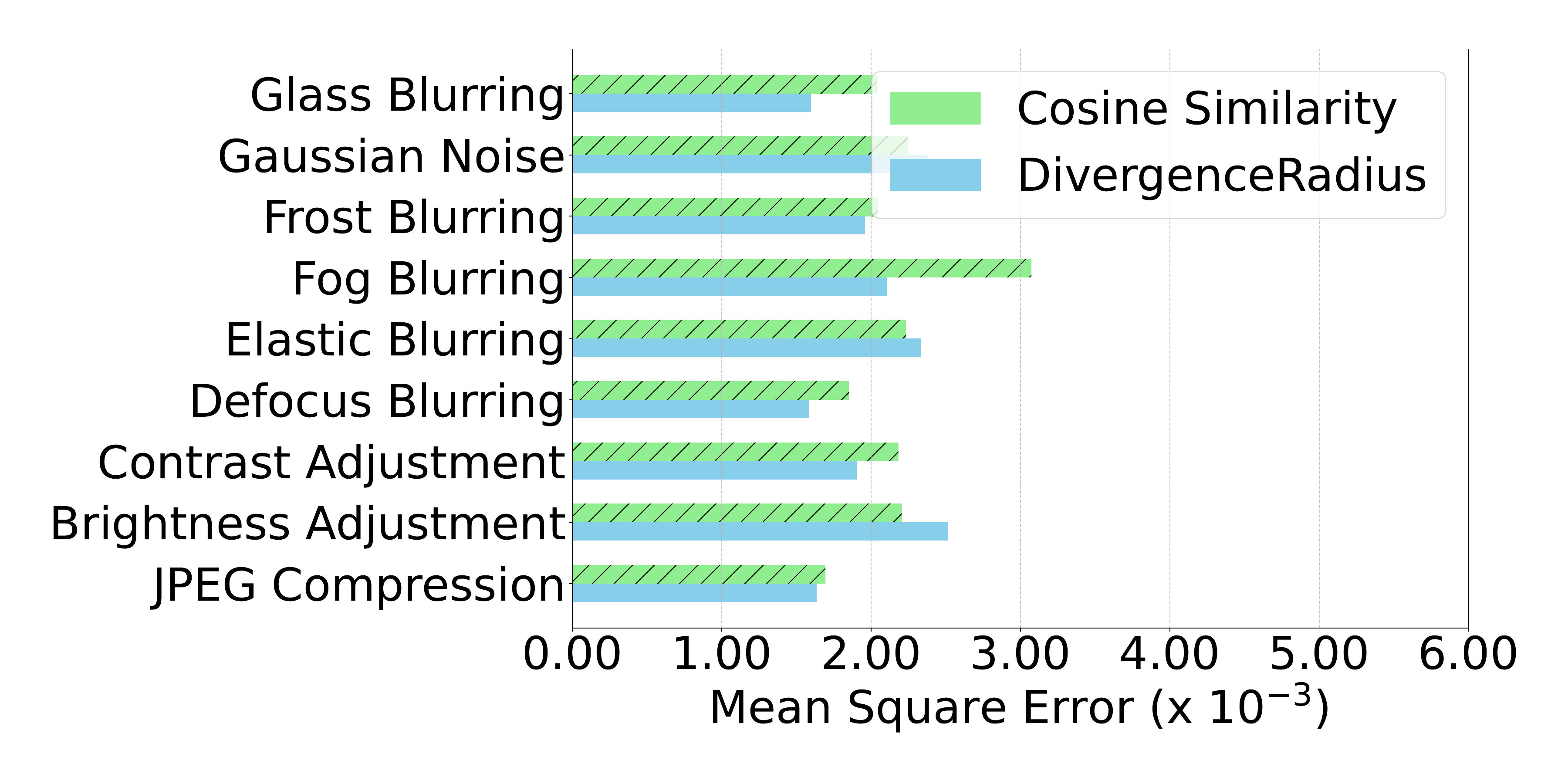}
    \caption{Mean squared error of predicting an image's $RMSE_p$ using its cosine similarity or \alg{} for NYU-Depth V2 under different perturbation functions.}
    \label{fig:mse_downstream_prediction_nyu}
\end{figure}

\begin{table*}[t!]
\fontsize{7}{11}\selectfont 
\caption{Pearson correlations between average $ACC$-average $ACC_p$ and average cosine similarity or \alg{} across the 9 perturbation functions for the two datasets and downstream classifiers.}
\label{tab:pearson_correlations}
\centering
\subfloat[Cosine similarity]{
\begin{tabular}{|c|c|c|}
\hline
 & \makecell{Zero-shot\\Classification} & \makecell{Linear-probe\\Classification}  \\ \hline
ImageNet  & 0.91   & 0.89   \\ \hline
Food101   & 0.94  & 0.92   \\ \hline
\end{tabular}
}
\subfloat[\alg{}]{
\begin{tabular}{|c|c|c|}
\hline
 & \makecell{Zero-shot\\Classification} & \makecell{Linear-probe\\Classification}  \\ \hline
ImageNet  & 0.92   & 0.89   \\ \hline
Food101   & 0.94  & 0.93   \\ \hline
\end{tabular}
}
\end{table*}

\myparatight{Predicting $RMSE_p$ of an image using its robustness value} The linear relationship between $RMSE_p$ and the robustness value indicates that we can predict $RMSE_p$ of an image using its robustness value by a linear regression model. We evaluate the performance of such prediction. Towards this goal, we divide the testing images of a dataset into two halves. We use  the pairs (robustness value, $RMSE_p$) of the first half of testing images to train a linear regression model, which takes a robustness value as input and outputs $RMSE_p$. Then, we evaluate this linear regression model on the pairs (\alg{}, $RMSE_p$) of the second half of the testing images.  Figure~\ref{fig:mse_downstream_prediction_nyu} shows the mean squared errors of the linear regression models under the 9 perturbation functions for the depth estimation model. The mean squared errors are very small, which indicates that an image's \alg{} under a perturbation function can be used to accurately predict a downstream depth estimation model's $RMSE$ for the image when it is perturbed by the perturbation function.

\begin{figure}[!h]
    \centering
    \subfloat[Cosine similarity]{\includegraphics[width=0.40\textwidth]{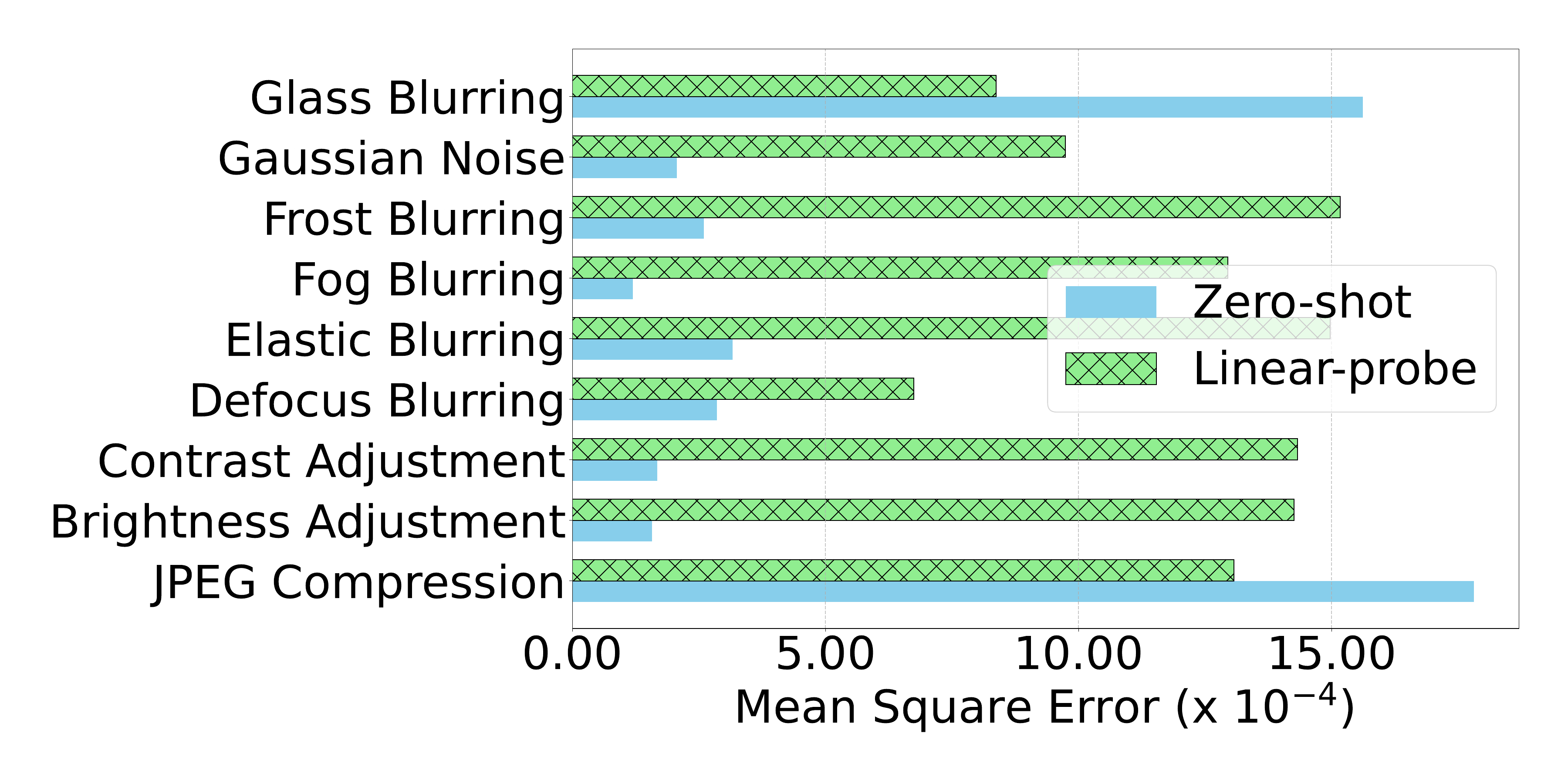}}
    \\
    \subfloat[\alg{}]{\includegraphics[width=0.40\textwidth]{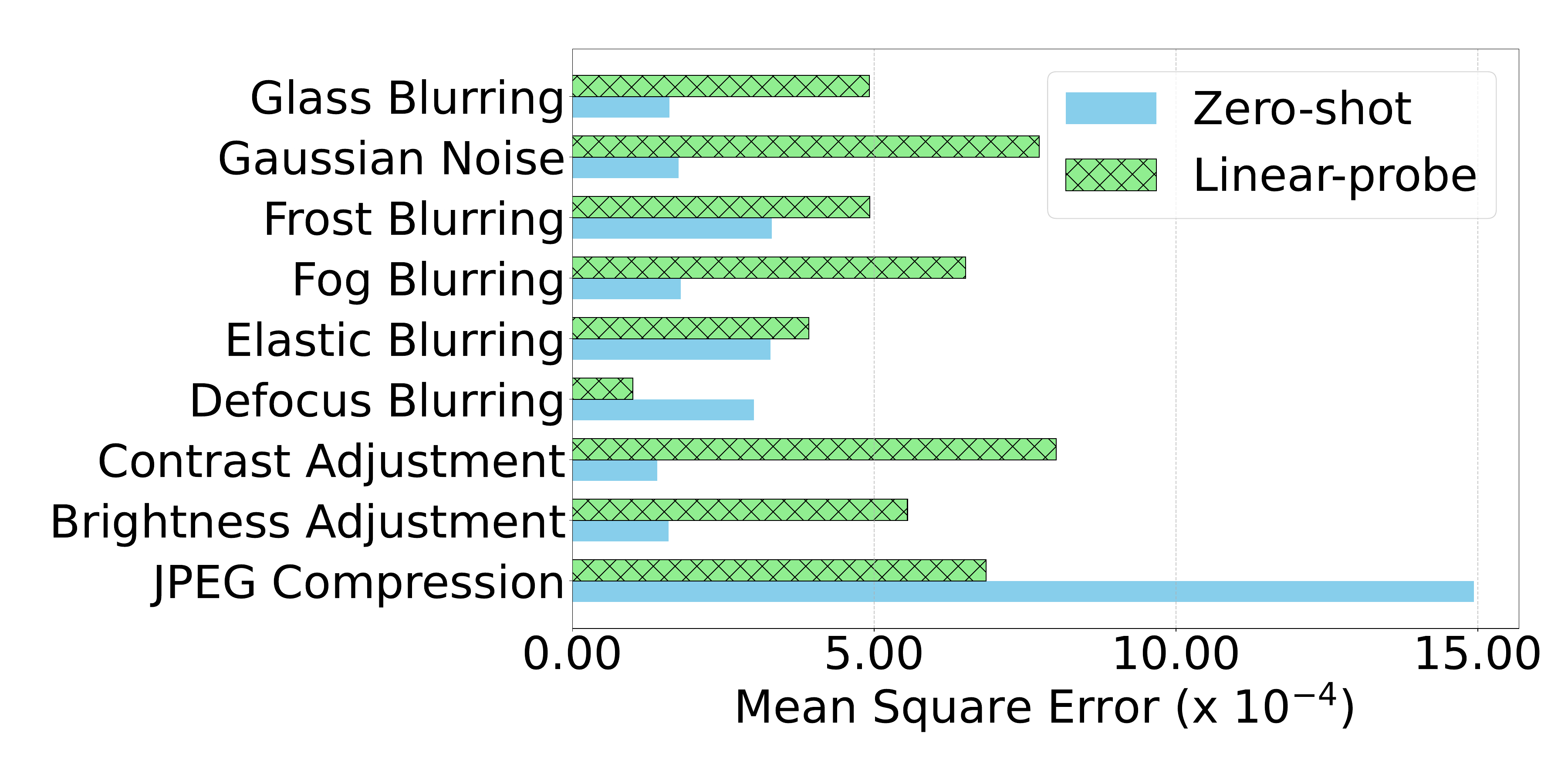}}
    \caption{Mean squared error of predicting an image's $ACC_p$ using its (a) cosine similarity or (b) \alg{} for ImageNet  under different perturbation functions in two downstream classifiers.}
    \label{fig:mse_downstream_prediction_imagenet}
\end{figure}

\begin{figure}[!h]
    \centering
    \subfloat[Cosine similarity]{\includegraphics[width=0.4\textwidth]{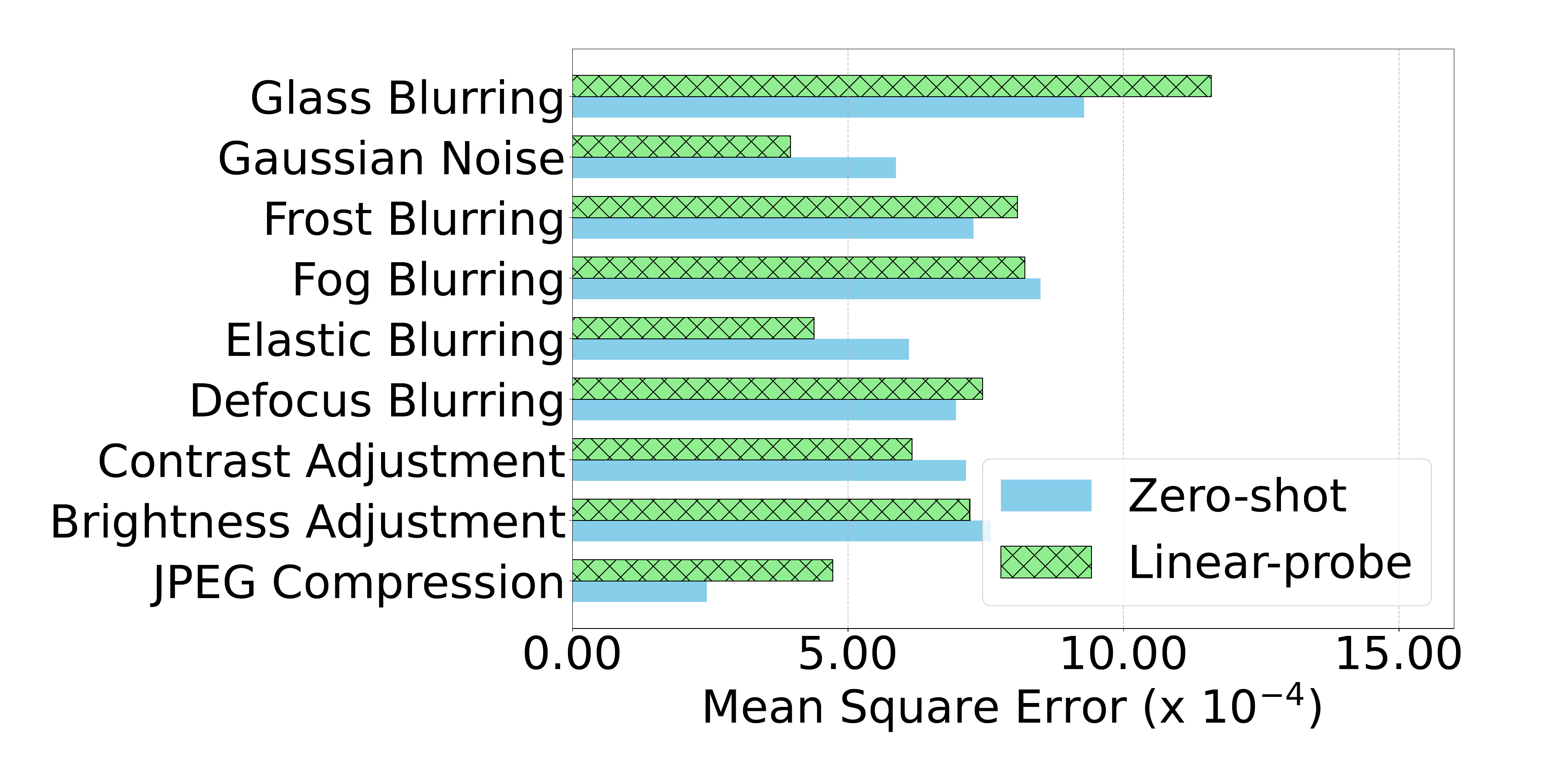}}\\
    \subfloat[\alg{}]{\includegraphics[width=0.4\textwidth]{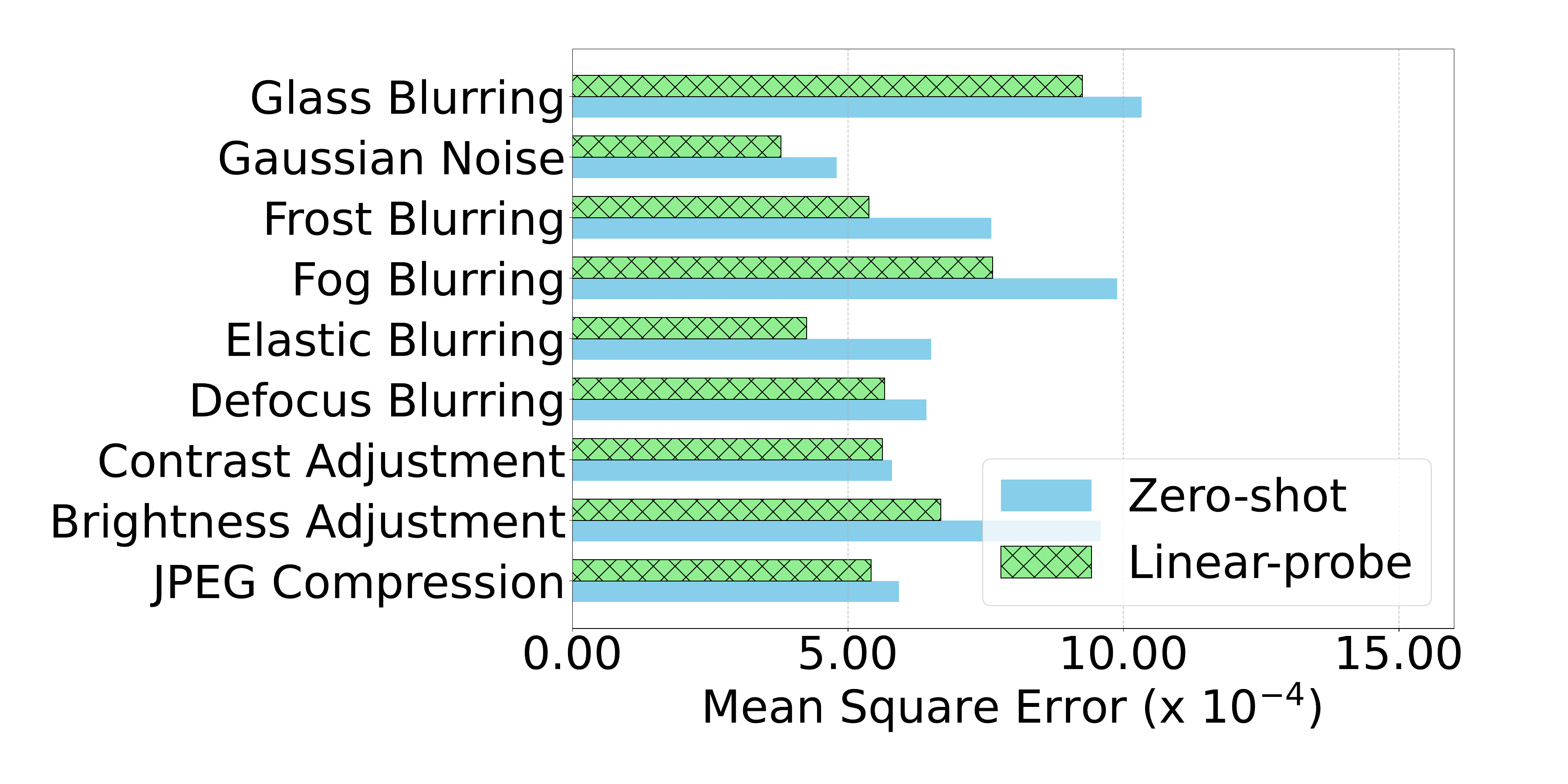}}
    \\
    \caption{Mean squared error of predicting an image's $ACC_p$ using its (a) cosine similarity or (b) \alg{}   for Food101  under different perturbation functions in two downstream classifiers.}
    \label{fig:mse_downstream_prediction_food101}
\end{figure}

\begin{table}[!h]
\fontsize{8pt}{10pt}\selectfont
\centering
\caption{Details of perturbation functions.}
\label{tab_robustness_radius_main}
\begin{tabular}{|>{\centering\arraybackslash}m{1.5cm}|c|c|>{\centering\arraybackslash}m{1.2cm}|}
\hline
\makecell{Perturbation\\Function} & \makecell{Key Perturbation\\Parameter $\mathbb{K}$} & Domain $\mathbb{K}$ & \makecell{Maximally\\Distorted\\ Example}\\
\hline
\makecell{JPEG\\Compression} & Quality factor & $[30,70]$ & \adjustbox{valign=c}{\includegraphics[width=0.05\textwidth]{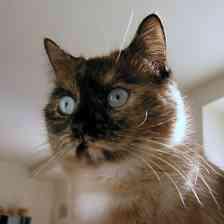}} \\
\hline
\makecell{Brightness\\Adjustment} & \makecell{Value in Hue-Saturation\\-Value space} & $[0.1, 0.5]$ & \adjustbox{valign=c}{\includegraphics[width=0.05\textwidth]{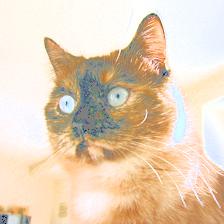}} \\
\hline
\makecell{Contrast\\Adjustment} & \makecell{Amplifying factor of\\deviations from the mean} & $[0.3, 0.7]$ & \adjustbox{valign=c}{\includegraphics[width=0.05\textwidth]{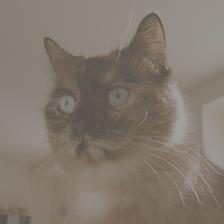}} \\
\hline
\makecell{Defocus\\Blurring} & Bluring disk kernel & $[1, 5]$ & \adjustbox{valign=c}{\includegraphics[width=0.05\textwidth]{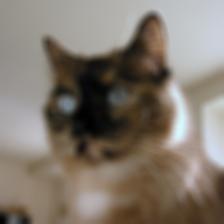}} \\
\hline
\makecell{Elastic\\Blurring} & Scaling factor & $[0.01, 0.05]$ & \adjustbox{valign=c}{\includegraphics[width=0.05\textwidth]{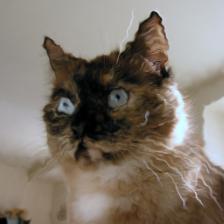}} \\
\hline
\makecell{Fog\\Blurring} & Density of fog & $[0.5, 2.5]$ & \adjustbox{valign=c}{\includegraphics[width=0.05\textwidth]{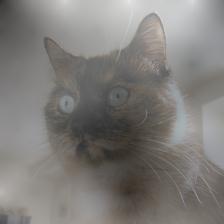}} \\
\hline
\makecell{Frost\\Blurring} & Weight of the frost & $[0.2, 0.6]$ & \adjustbox{valign=c}{\includegraphics[width=0.05\textwidth]{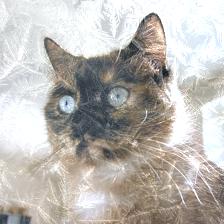}} \\
\hline
\makecell{Gaussian\\Noise} & Standard deviation & $[0.02, 0.10]$ & \adjustbox{valign=c}{\includegraphics[width=0.05\textwidth]{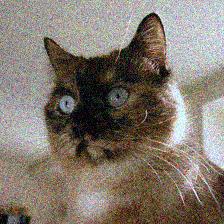}} \\
\hline
\makecell{Glass Blurring} & Standard deviation & $[0.2, 1.0]$ & \adjustbox{valign=c}{\includegraphics[width=0.05\textwidth]{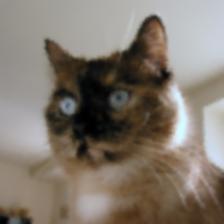}} \\
\hline
\end{tabular}
\end{table}

\begin{figure*}[t!]
    \centering
    \subfloat[JPEG compression]{\includegraphics[width=0.2\textwidth]{figures/bar_graphs/bar_graph_jpeg_mean.pdf}\label{fig:1}}
    \hfill
    \subfloat[Brightness adjustment]{\includegraphics[width=0.2\textwidth]{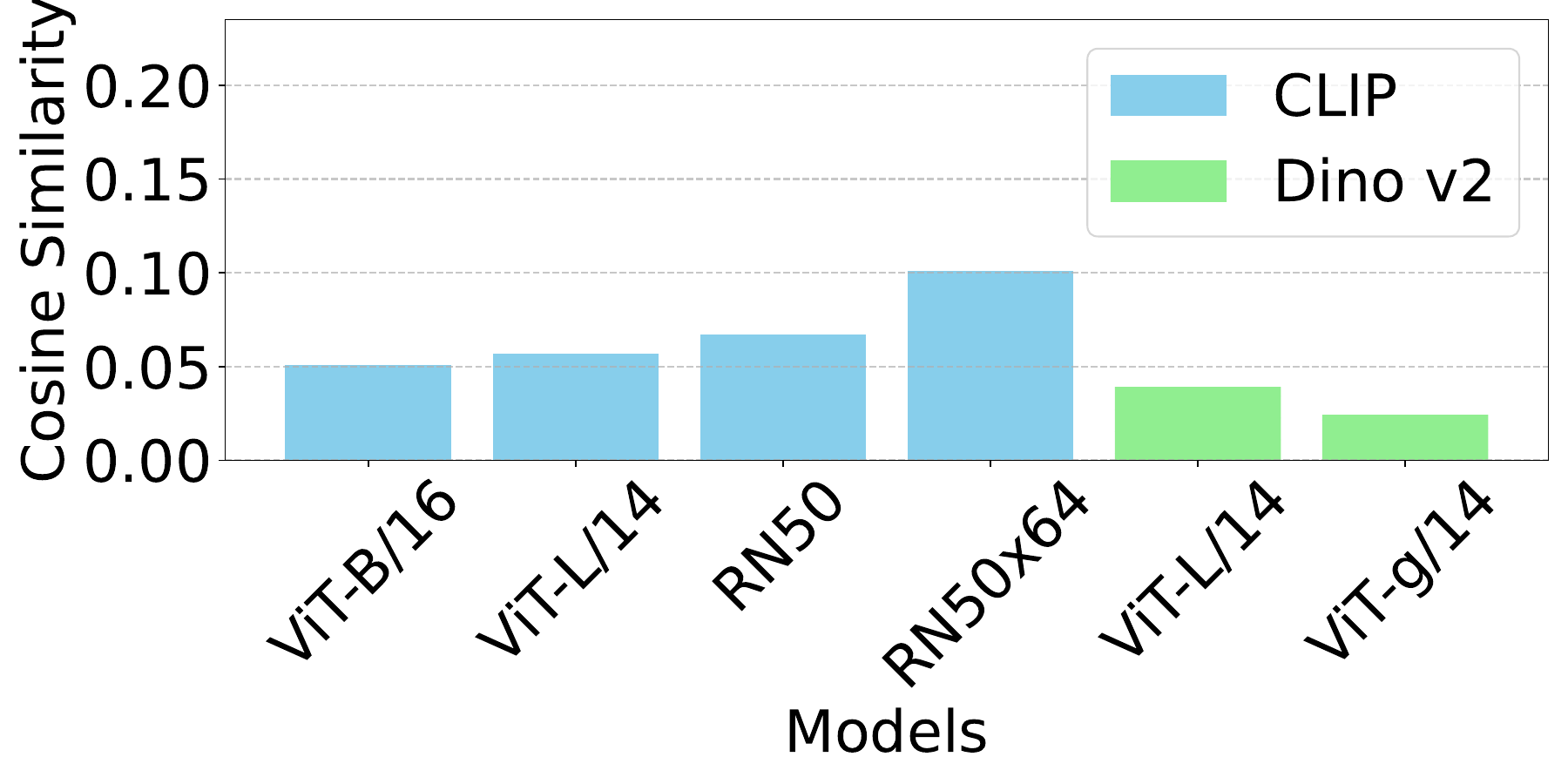}\label{fig:2}}
    \hfill
    \subfloat[Contrast adjustment]{\includegraphics[width=0.2\textwidth]{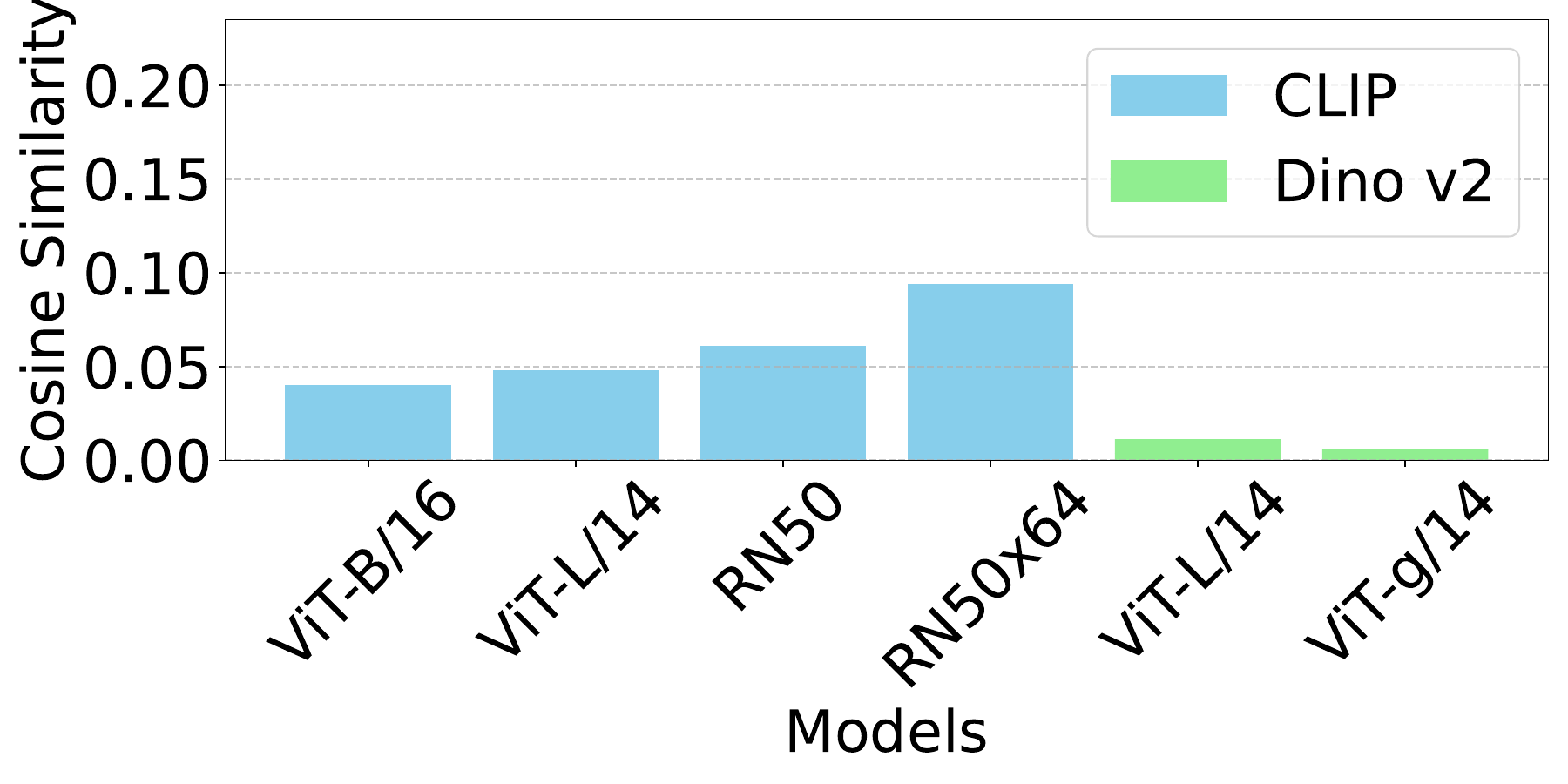}\label{fig:3}}
    \hfill
    \subfloat[Defocus blurring]{\includegraphics[width=0.2\textwidth]{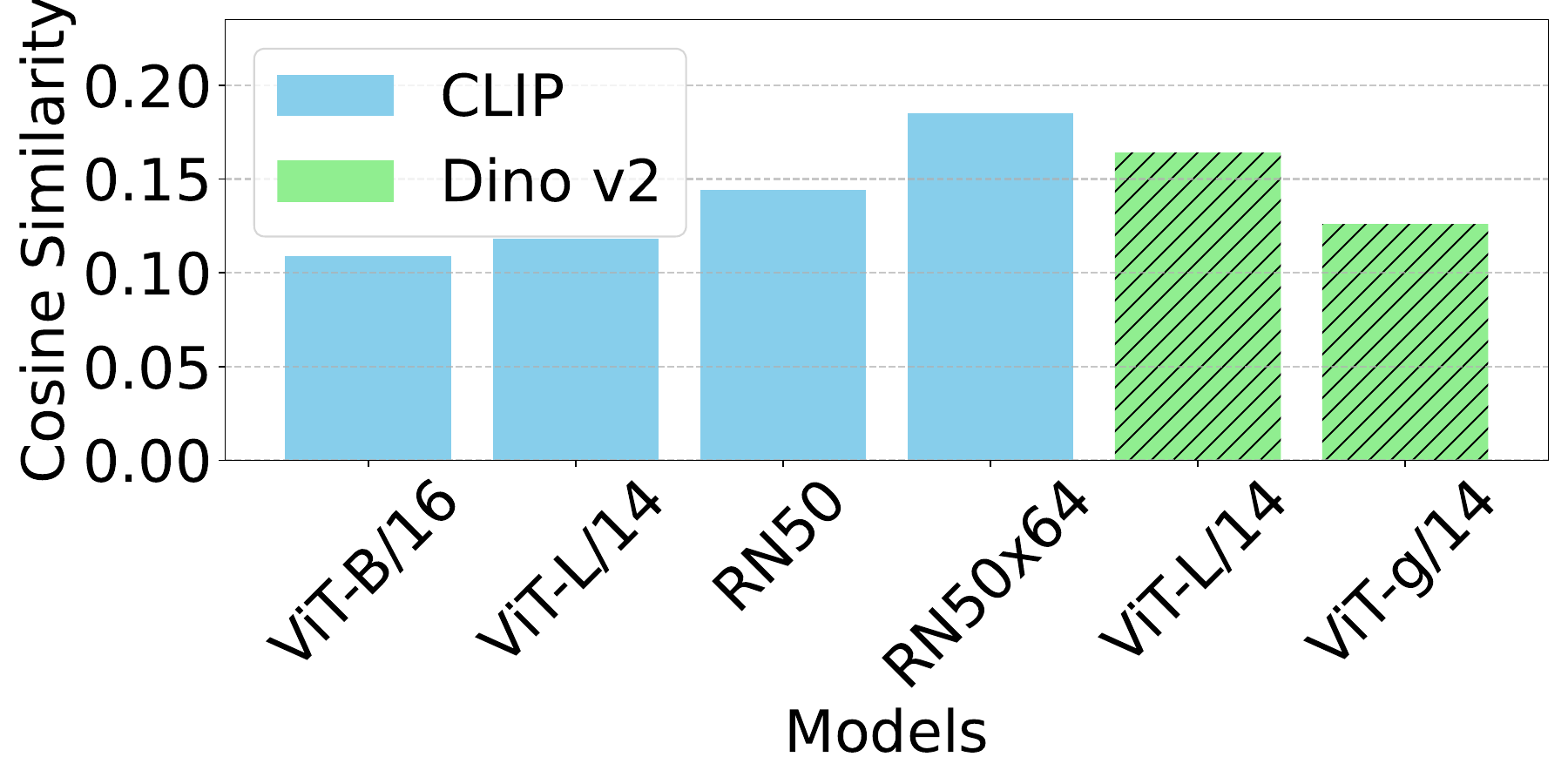}\label{fig:4}}
    \hfill
    \subfloat[Elastic blurring]{\includegraphics[width=0.2\textwidth]{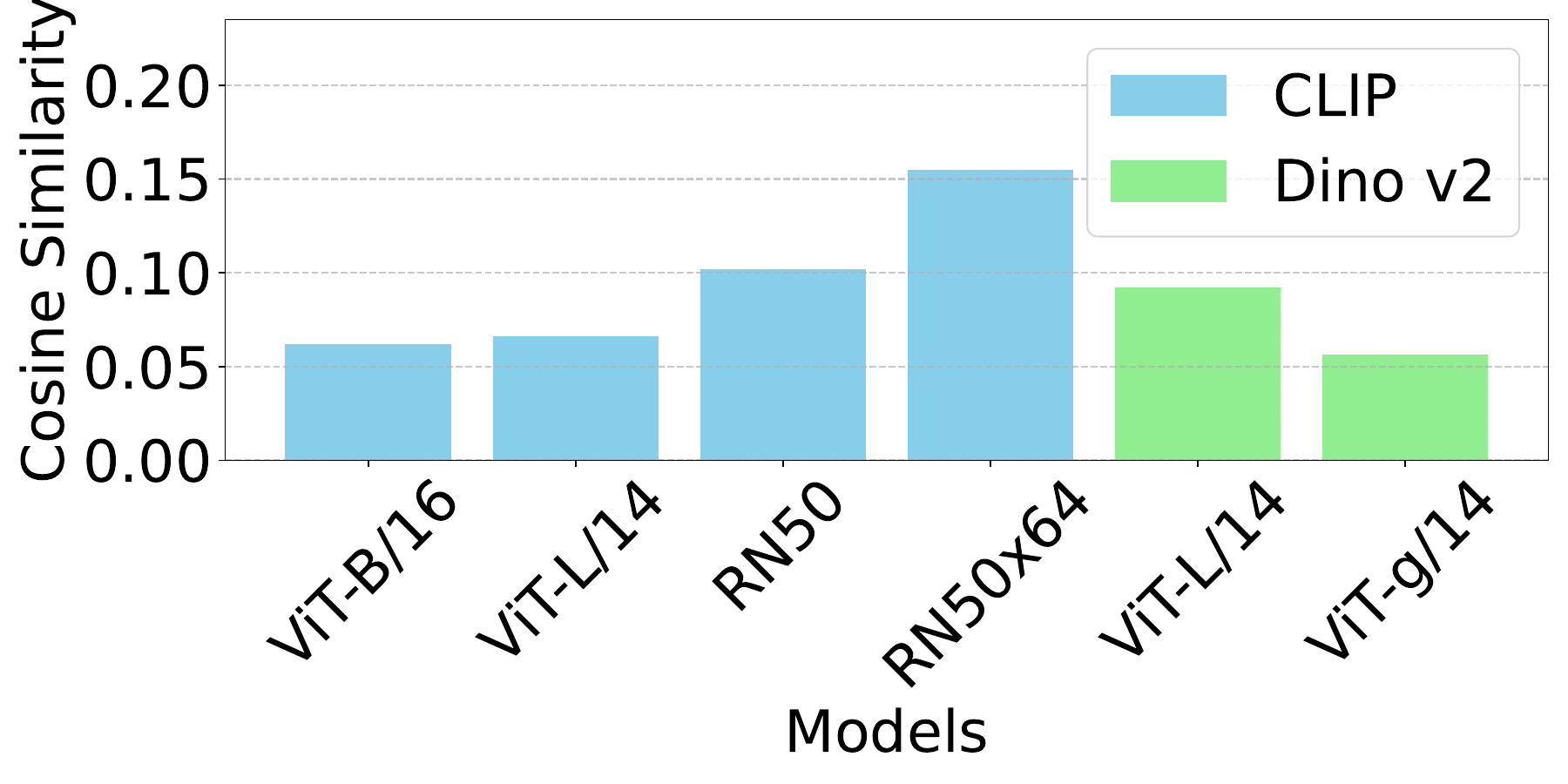}\label{fig:5}}
    \\
    \subfloat[Fog blurring]{\includegraphics[width=0.2\textwidth]{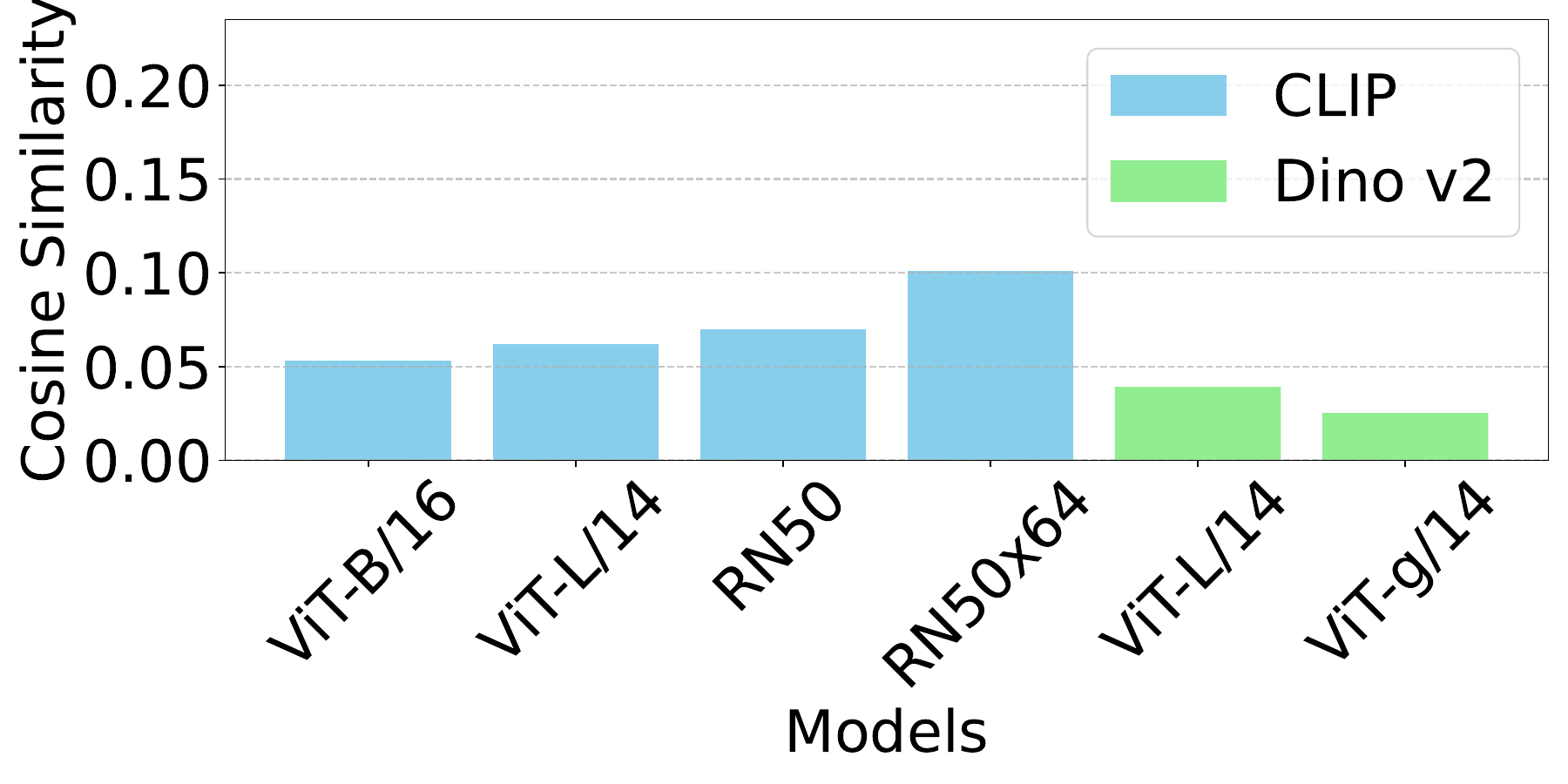}\label{fig:6}}    
    \hfill
    \subfloat[Frost blurring]{\includegraphics[width=0.2\textwidth]{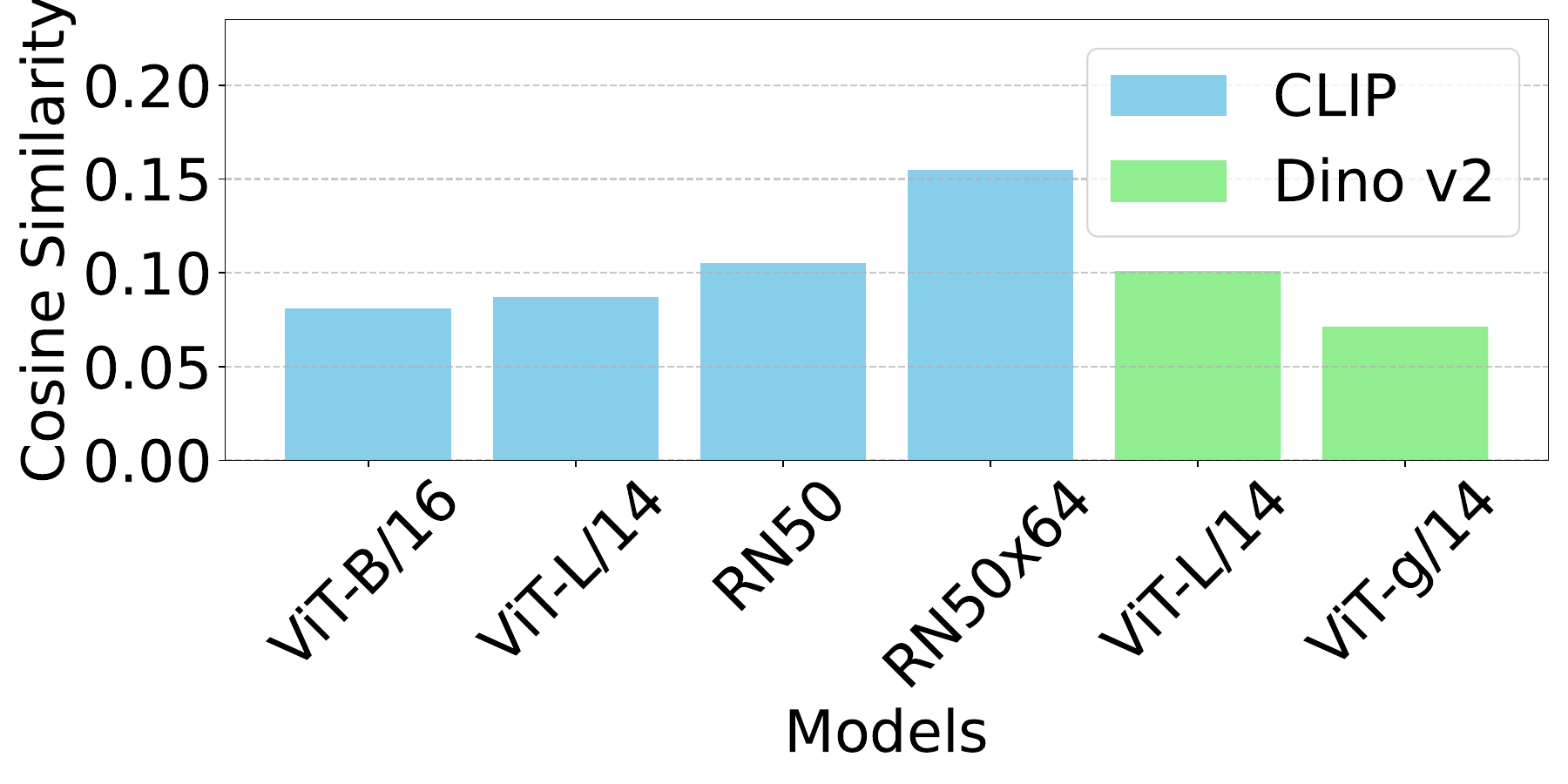}\label{fig:7}}
    \hfill
    \subfloat[Gaussian noise]{\includegraphics[width=0.2\textwidth]{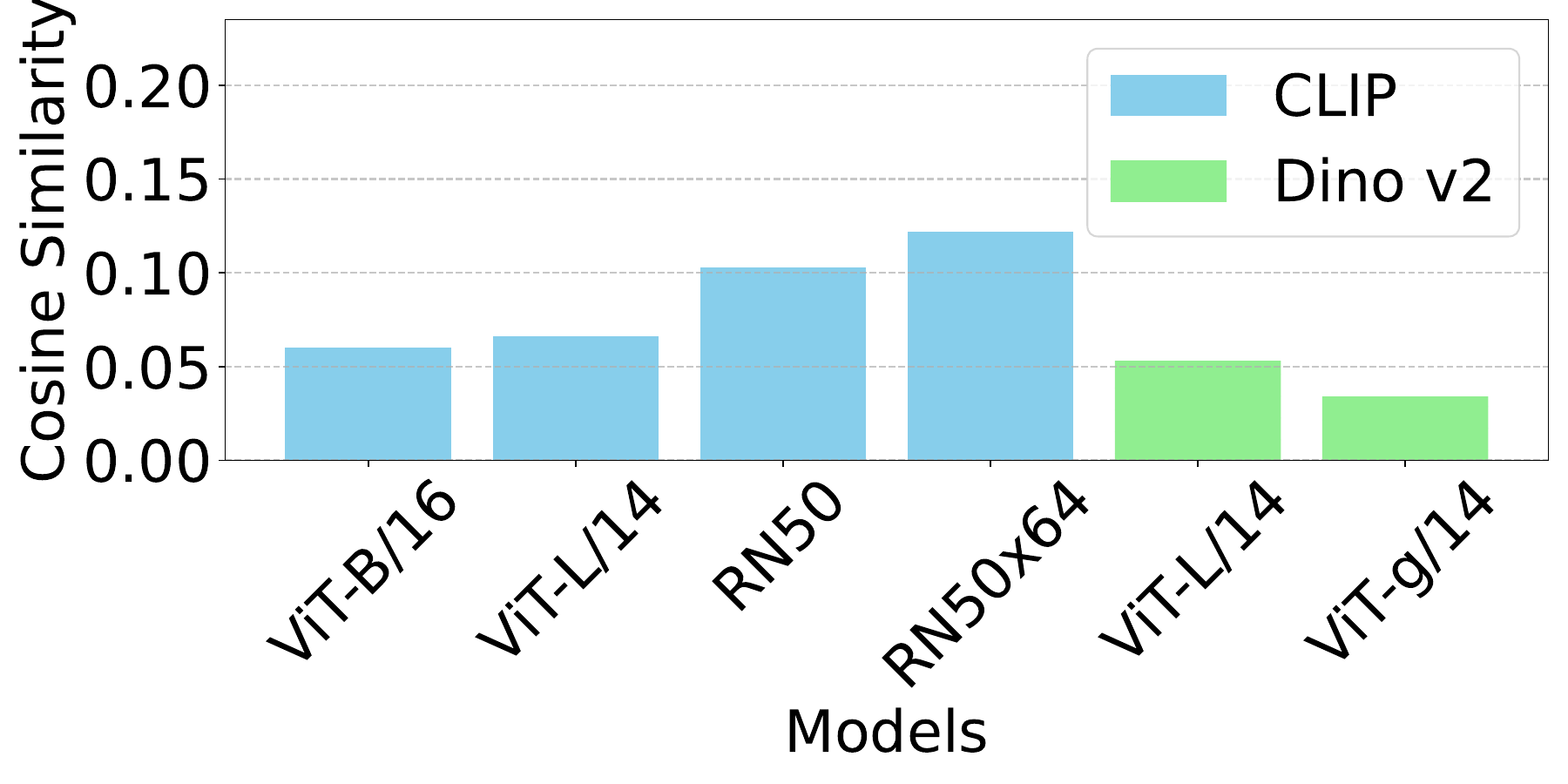}\label{fig:8}}
    \hfill
    \subfloat[Glass blurring]{\includegraphics[width=0.2\textwidth]{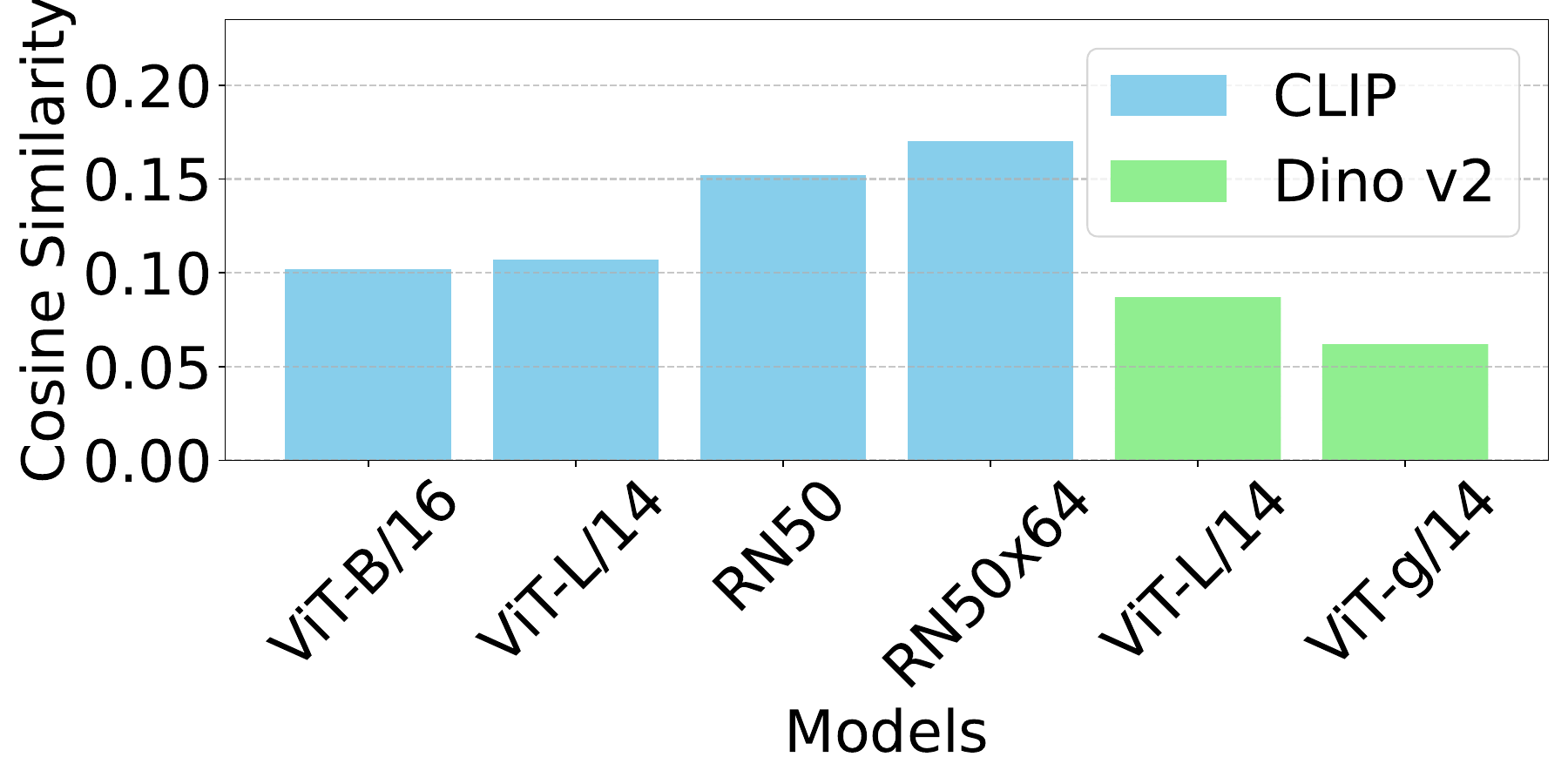}}
    \caption{Average cosine similarity of ImageNet testing images for different foundation models and  perturbation functions.}
    \label{fig:cosine_vision_imagenet}
\end{figure*}

\begin{figure*}[!h]
    \centering
    \subfloat[JPEG Compression]{\includegraphics[width=0.2\textwidth]{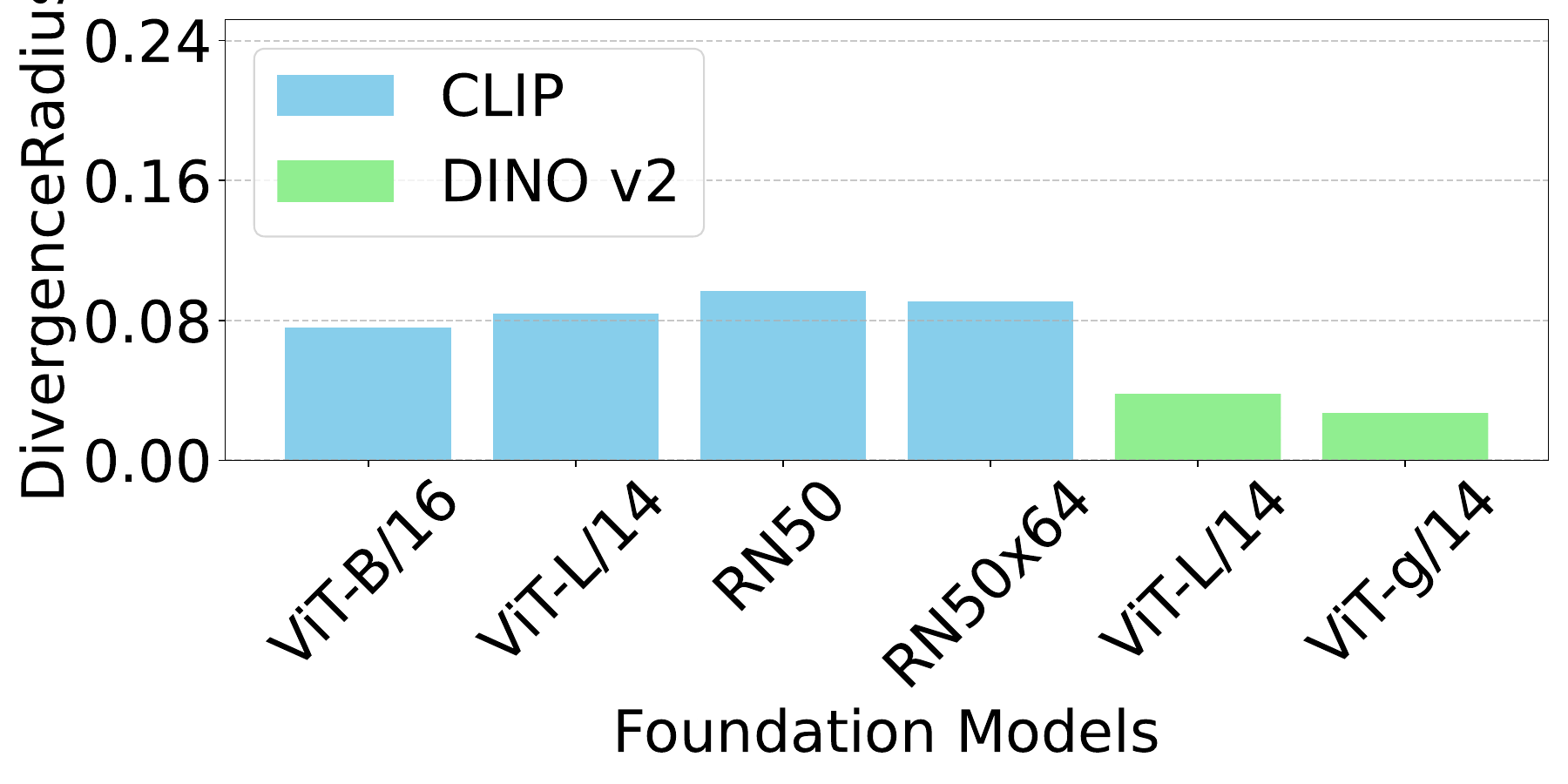}\label{fig:1}}
    \hfill
    \subfloat[Brightness adjustment]{\includegraphics[width=0.2\textwidth]{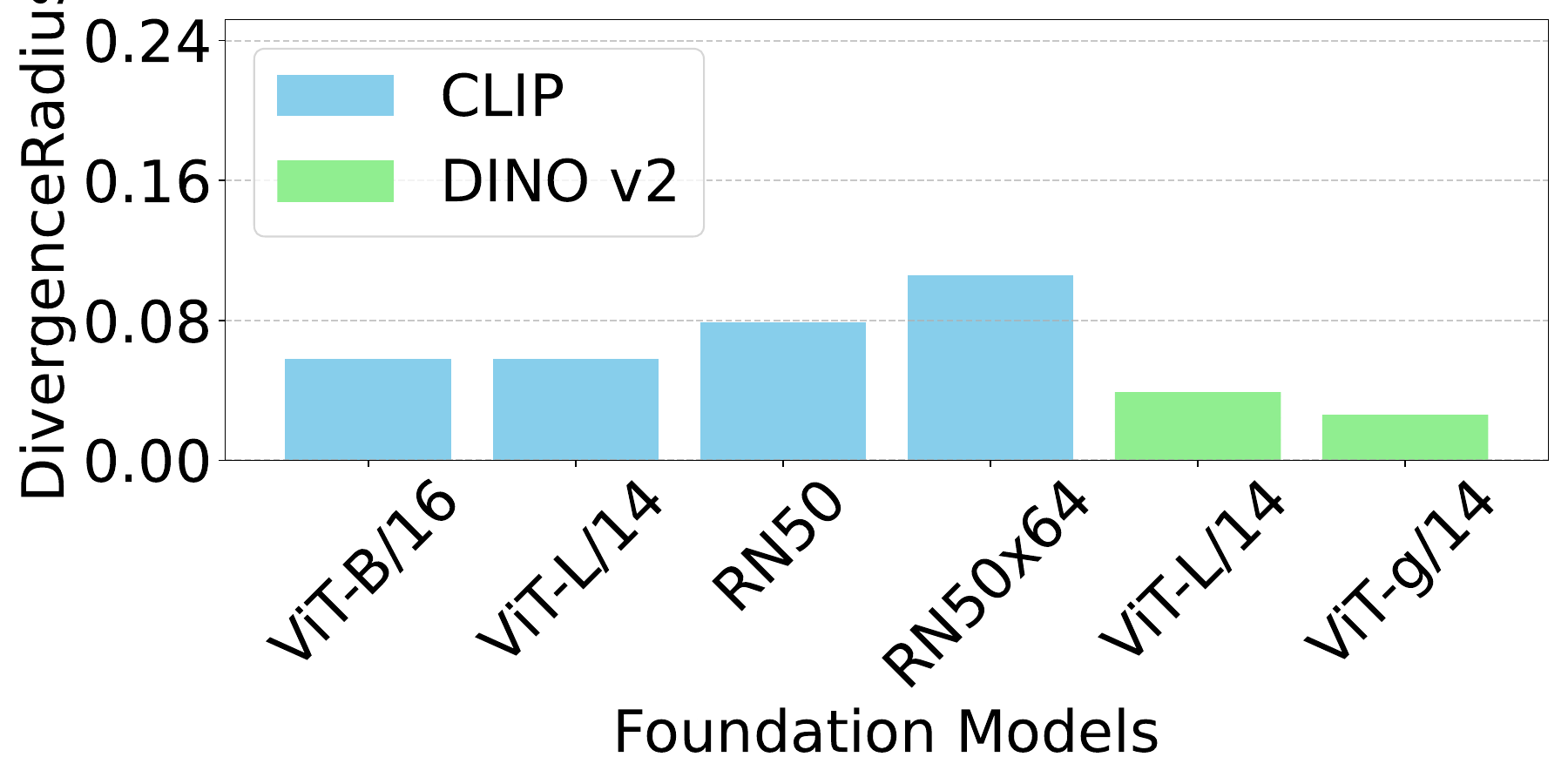}\label{fig:2}}
    \hfill
    \subfloat[Contrast adjustment]{\includegraphics[width=0.2\textwidth]{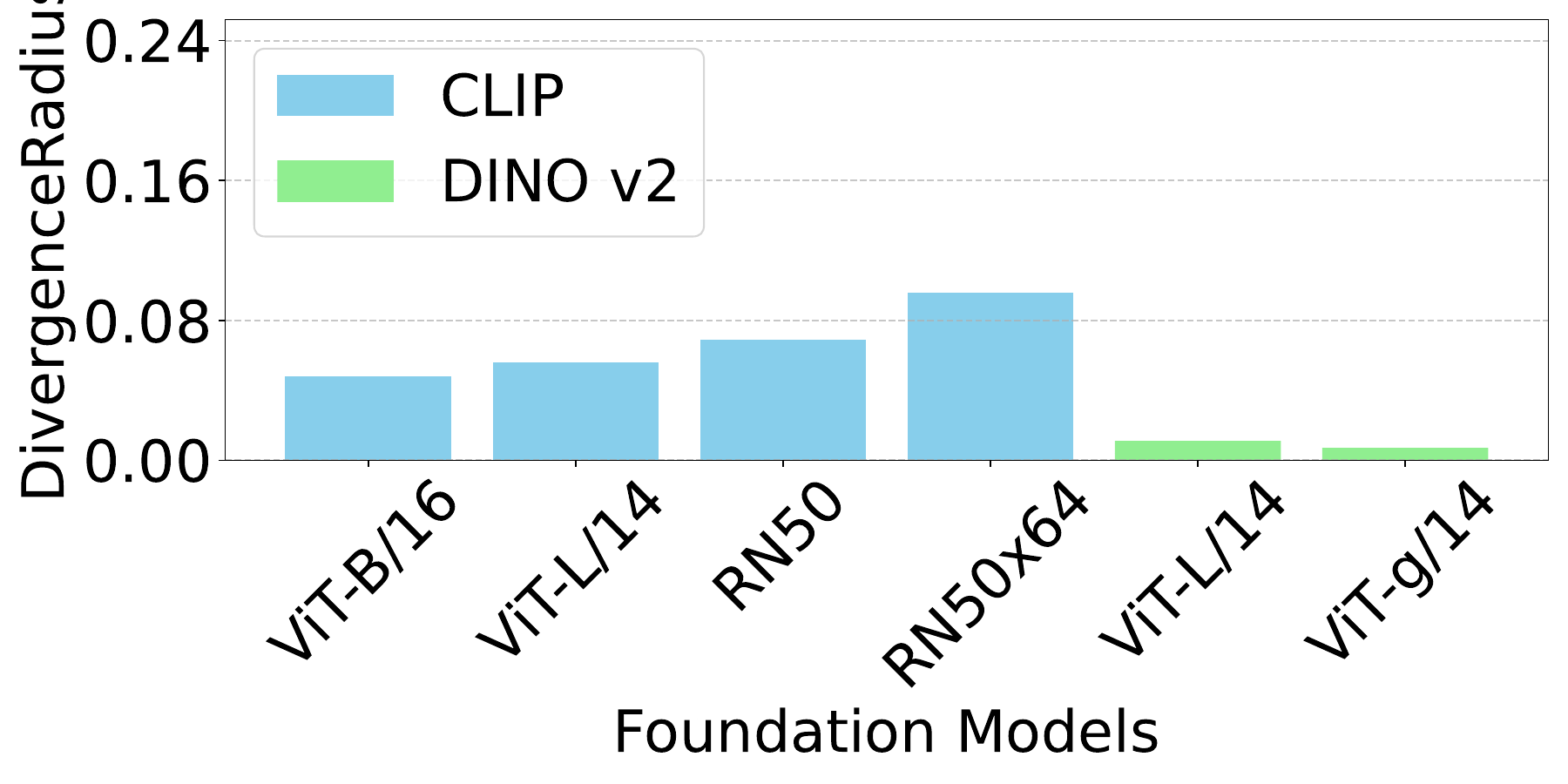}\label{fig:3}}
    \hfill
    \subfloat[Defocus blurring]{\includegraphics[width=0.2\textwidth]{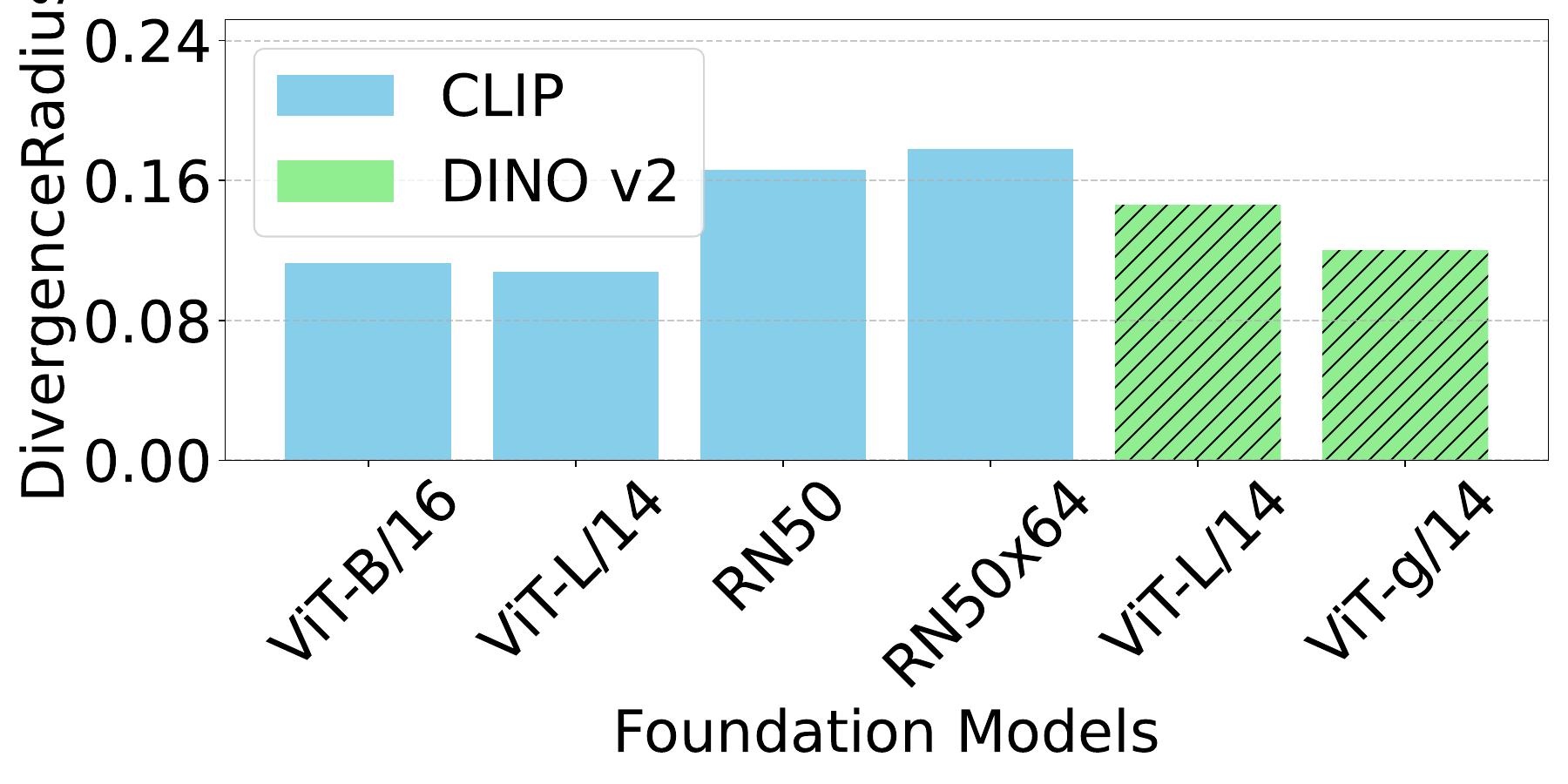}\label{fig:4}}
    \hfill
    \subfloat[Elastic blurring]{\includegraphics[width=0.2\textwidth]{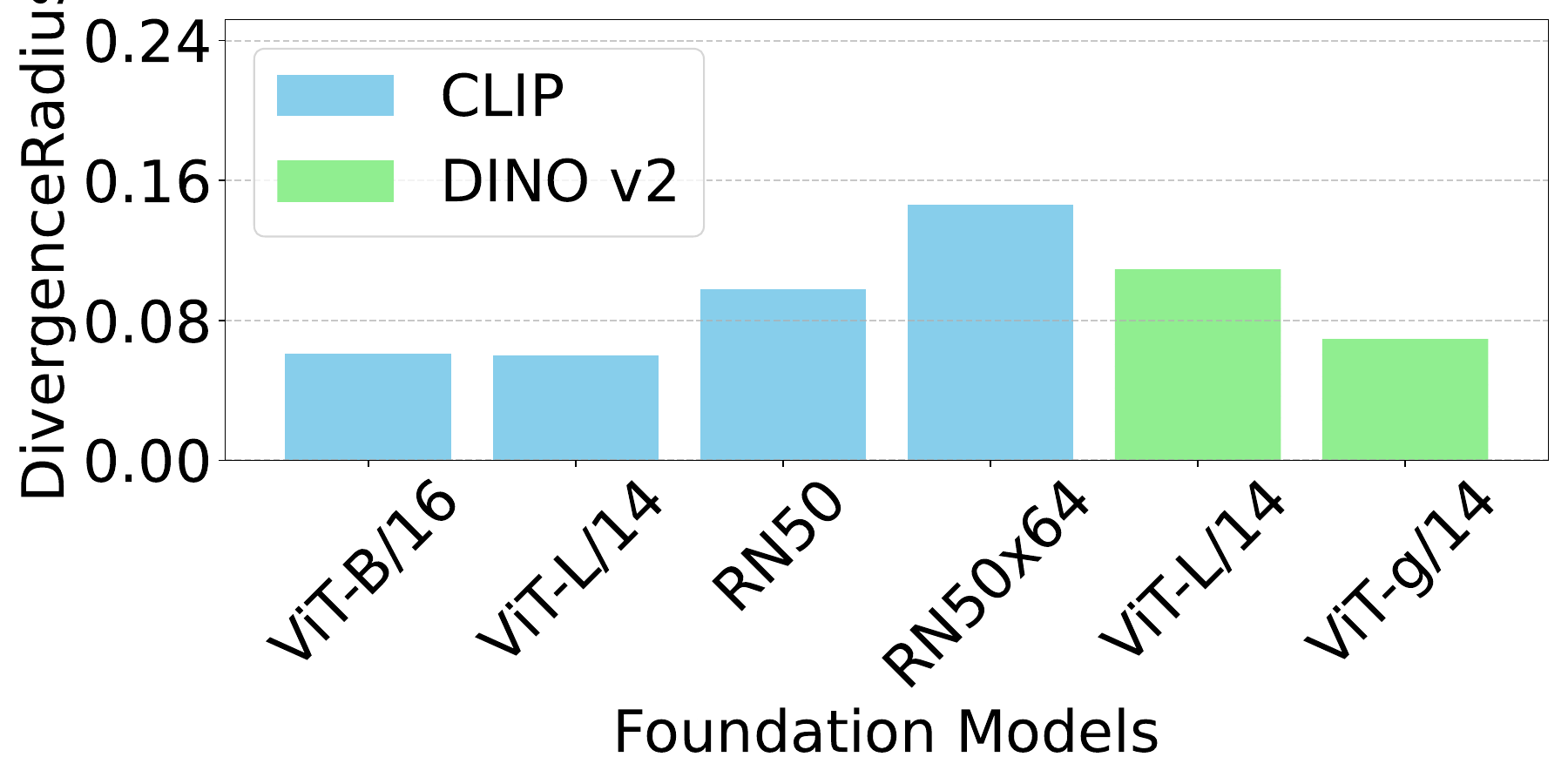}\label{fig:5}}
    \\
    \subfloat[Fog blurring]{\includegraphics[width=0.2\textwidth]{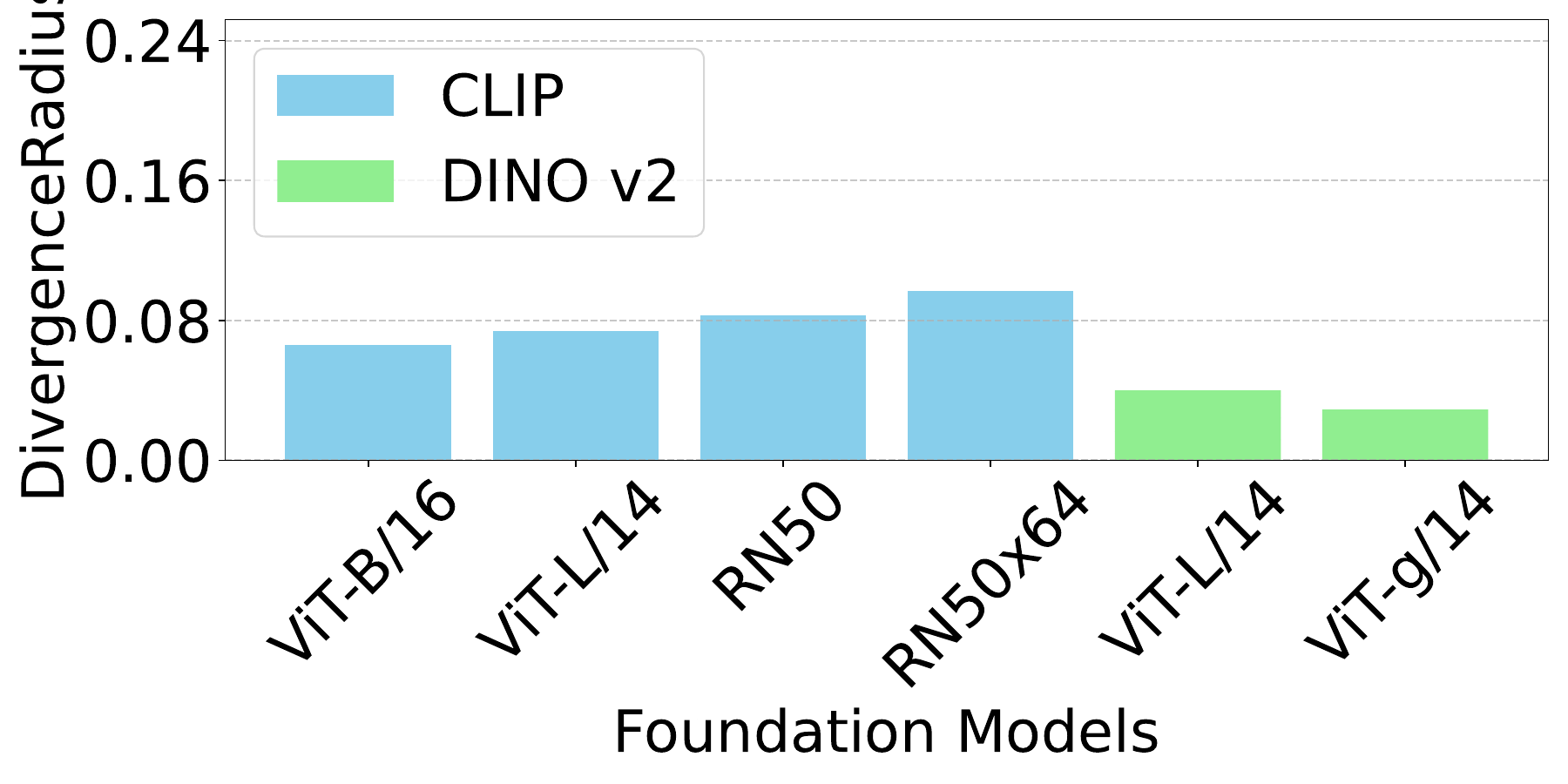}\label{fig:6}}
    \hfill
    \subfloat[Frost blurring]{\includegraphics[width=0.2\textwidth]{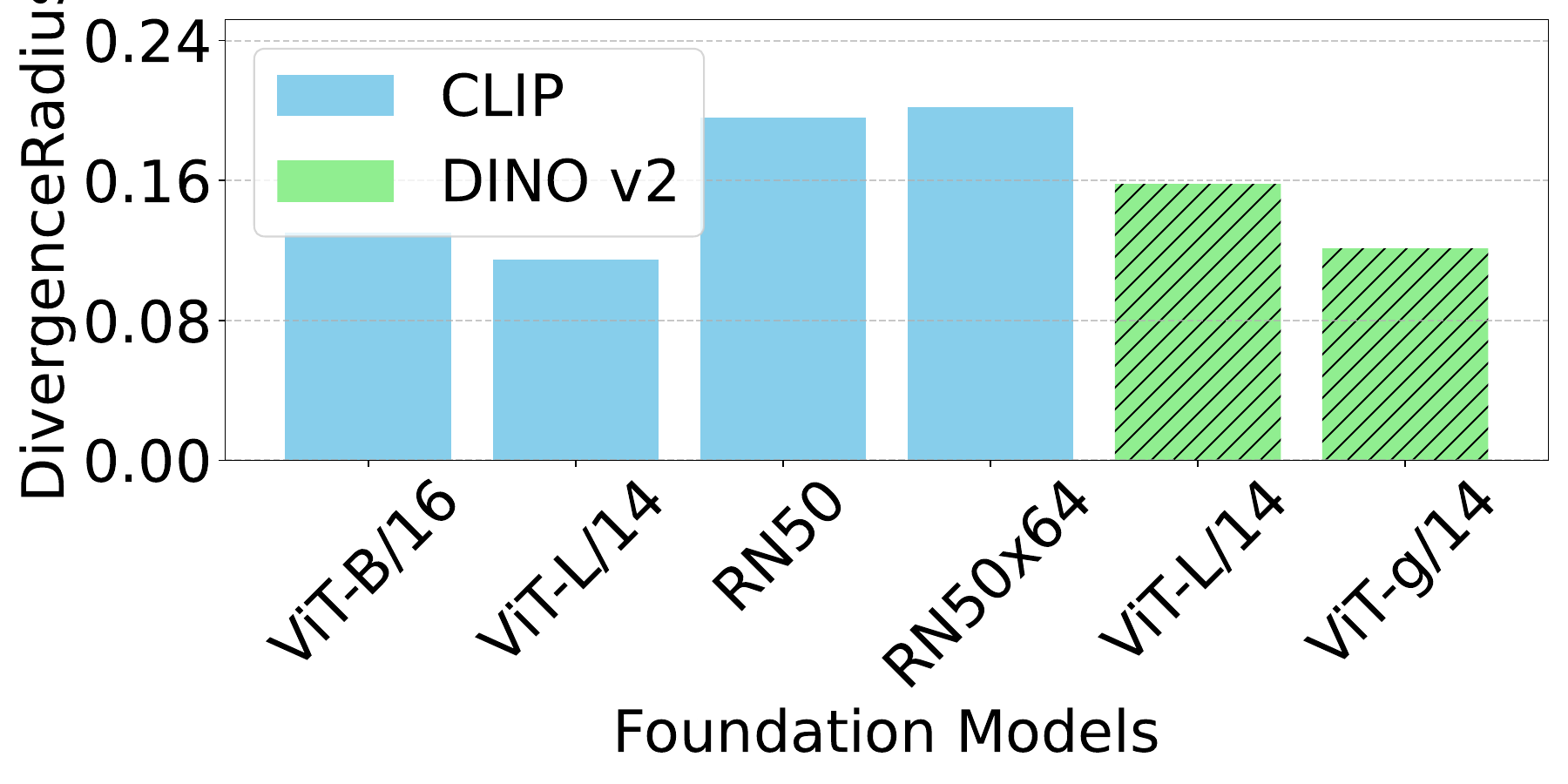}\label{fig:7}}
    \hfill
    \subfloat[Gaussian blurring]{\includegraphics[width=0.2\textwidth]{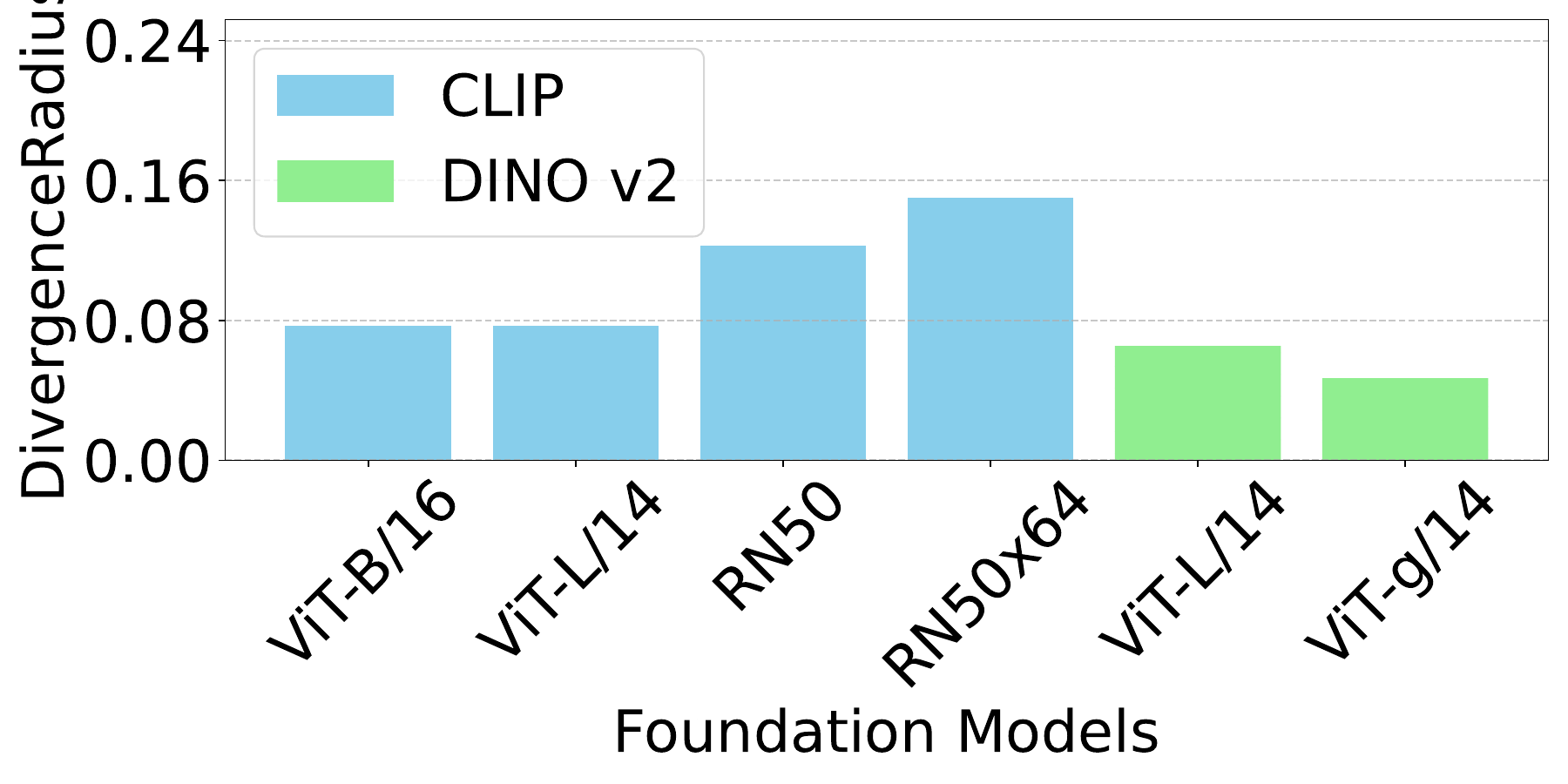}\label{fig:8}}
    \hfill
    \subfloat[Glass blurring]{\includegraphics[width=0.2\textwidth]{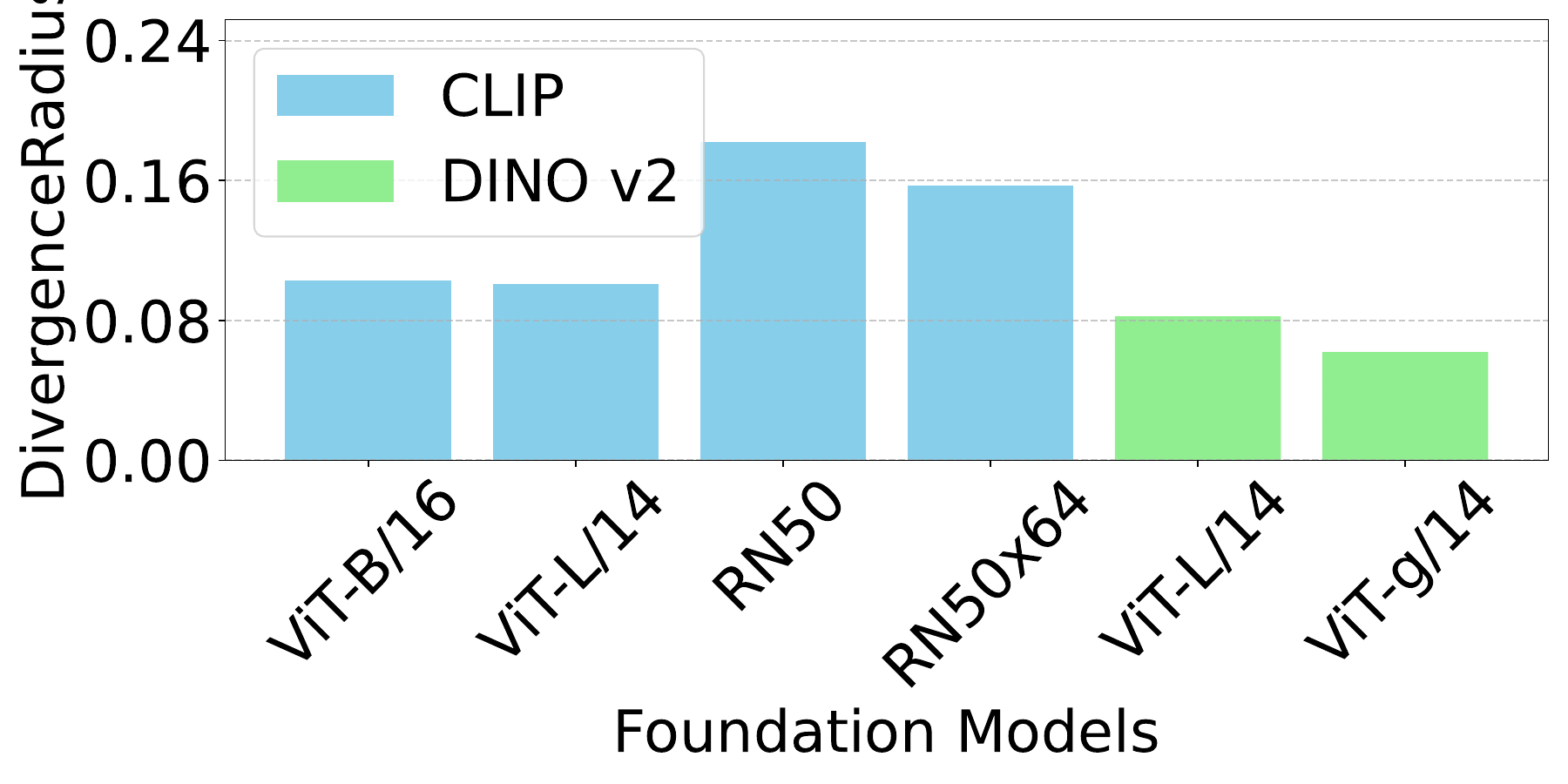}}
    \caption{Average \alg{} of Food101 testing images for different foundation models and  perturbation functions.}
    \label{fig:divergence_radius_vision_Food101}
\end{figure*}

\begin{figure*}[!h]
    \centering
    \subfloat[JPEG Compression]{\includegraphics[width=0.2\textwidth]{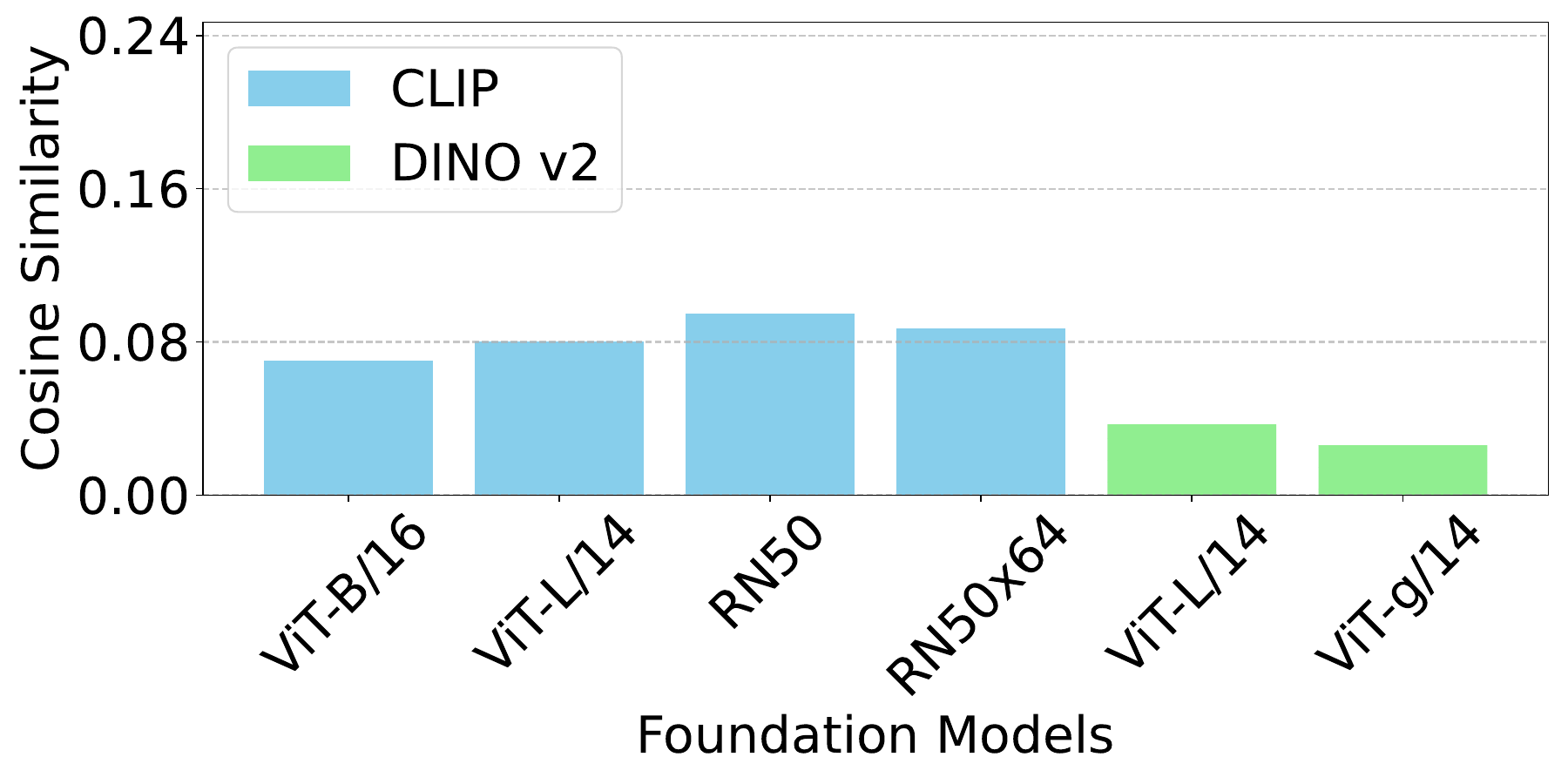}\label{fig:1}}
    \hfill
    \subfloat[Brightness adjustment]{\includegraphics[width=0.2\textwidth]{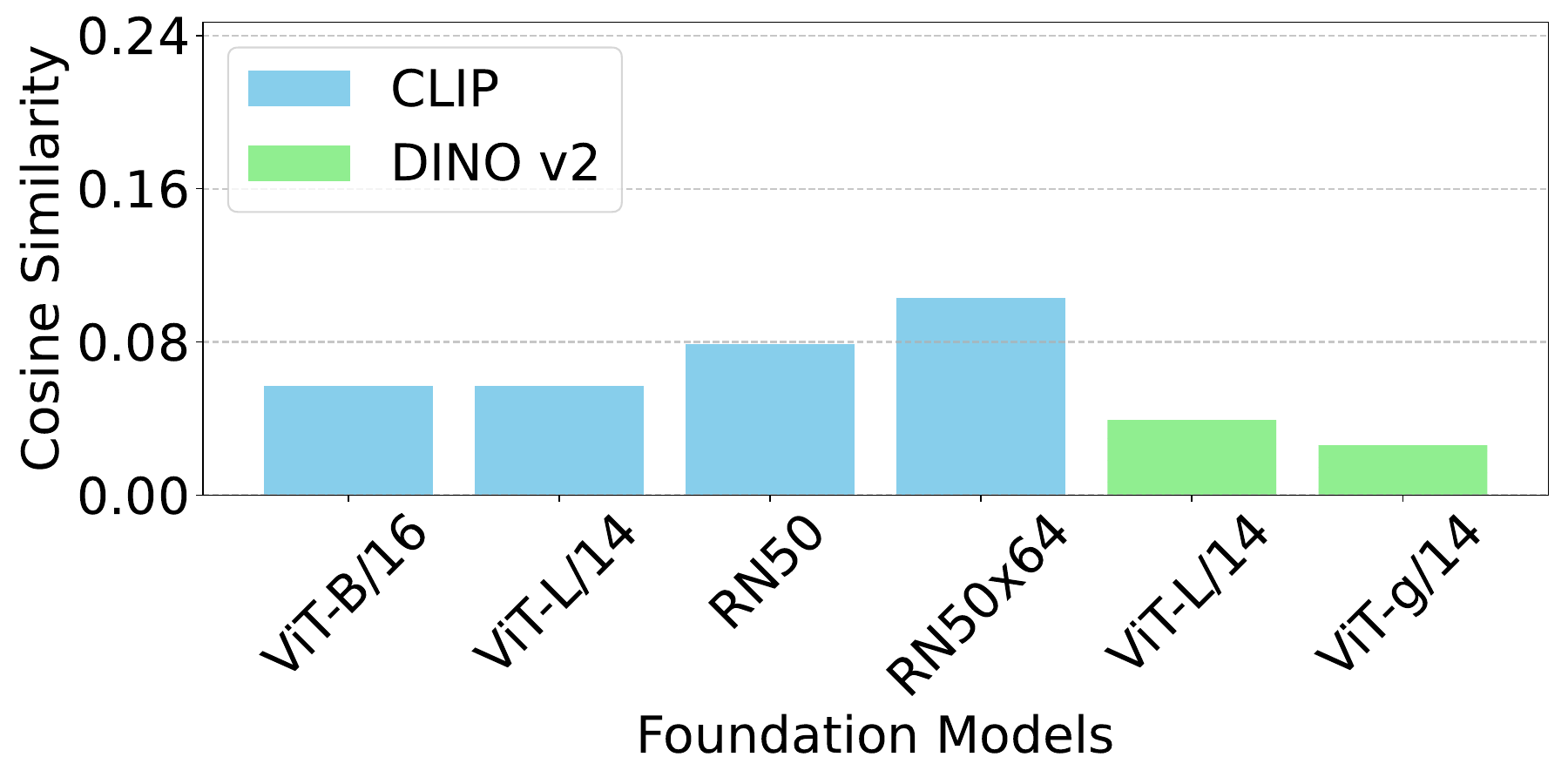}\label{fig:2}}
    \hfill
    \subfloat[Contrast adjustment]{\includegraphics[width=0.2\textwidth]{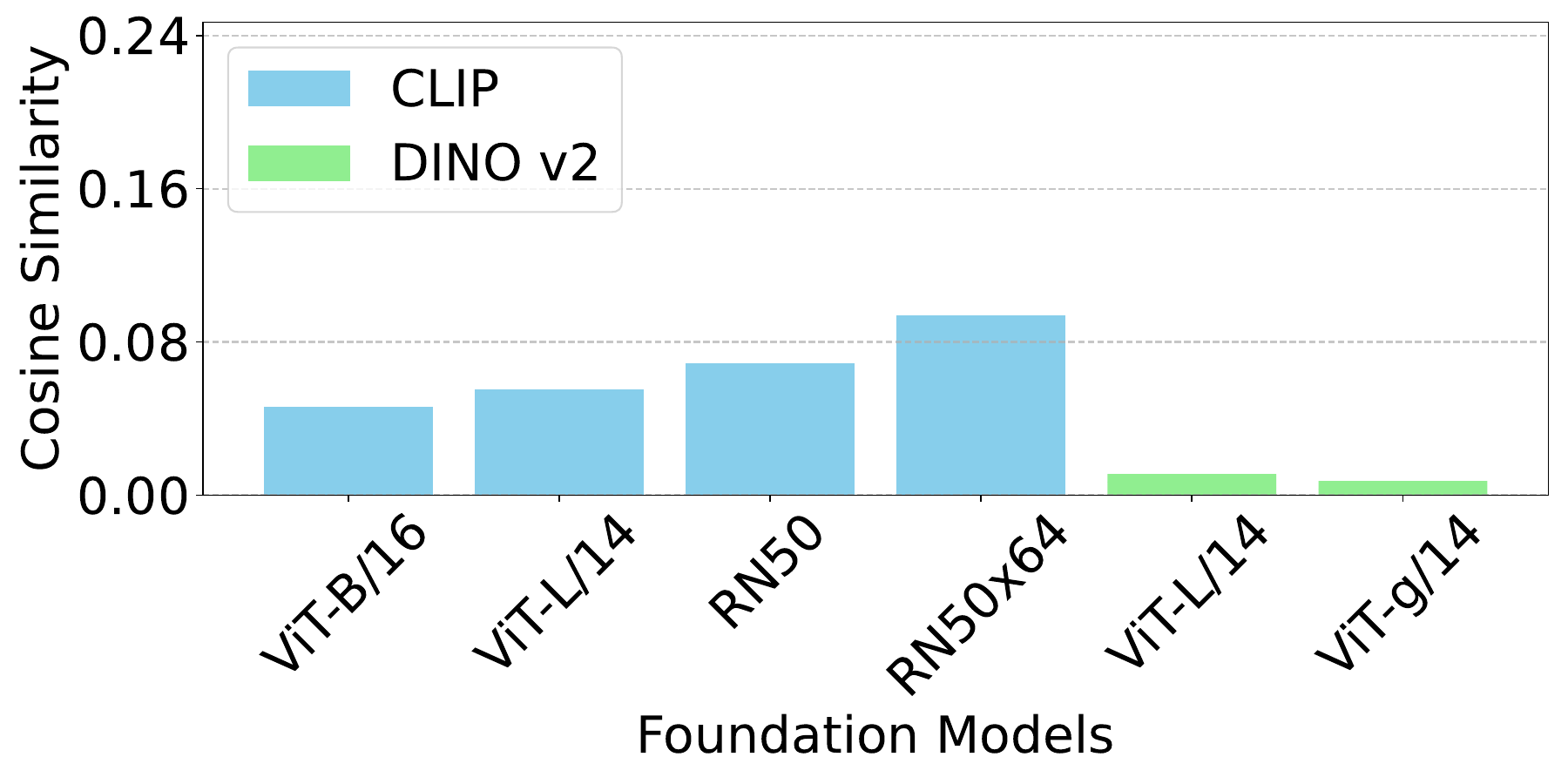}\label{fig:3}}
    \hfill
    \subfloat[Defocus blurring]{\includegraphics[width=0.2\textwidth]{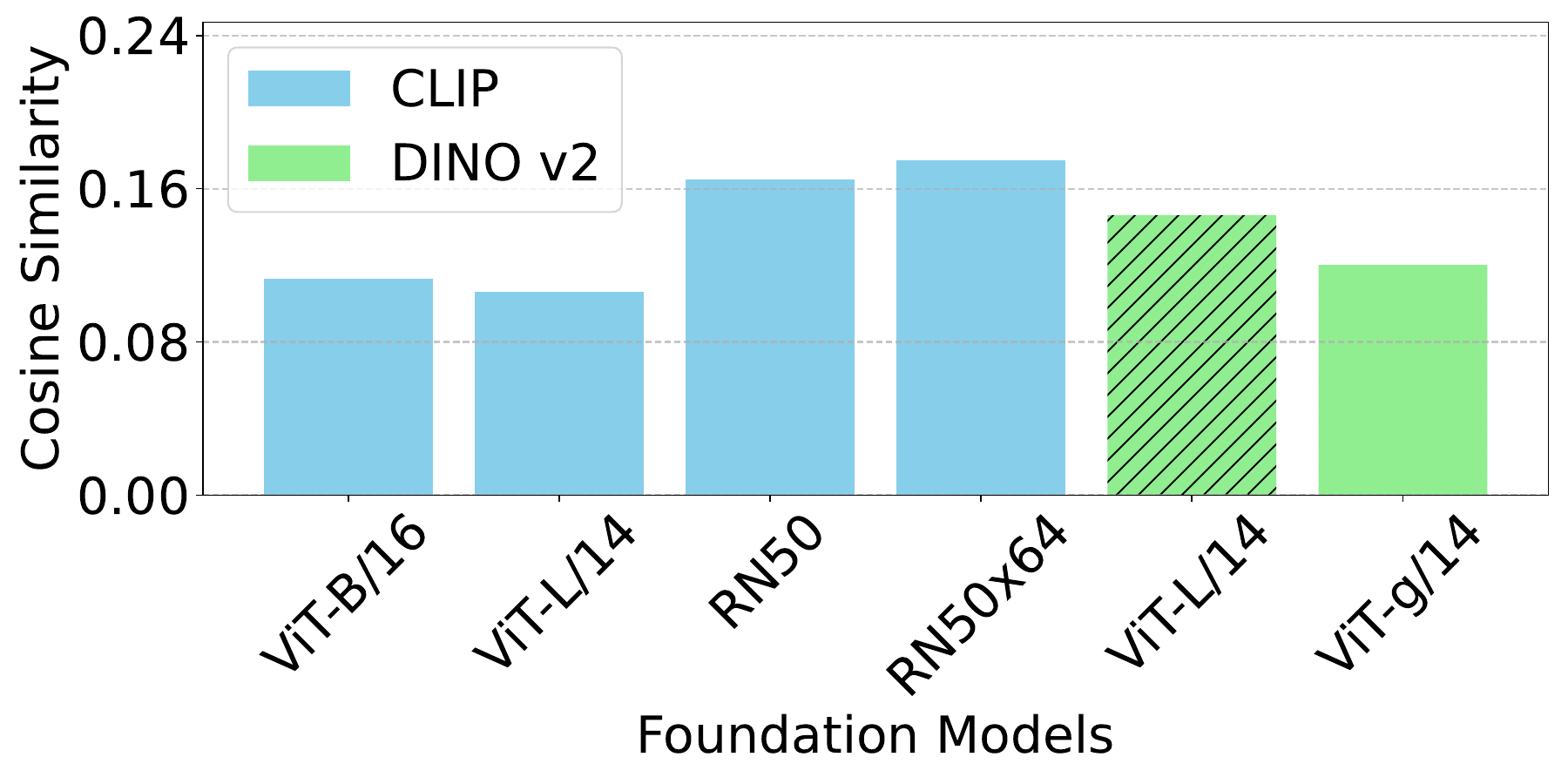}\label{fig:4}}
    \hfill
    \subfloat[Elastic blurring]{\includegraphics[width=0.2\textwidth]{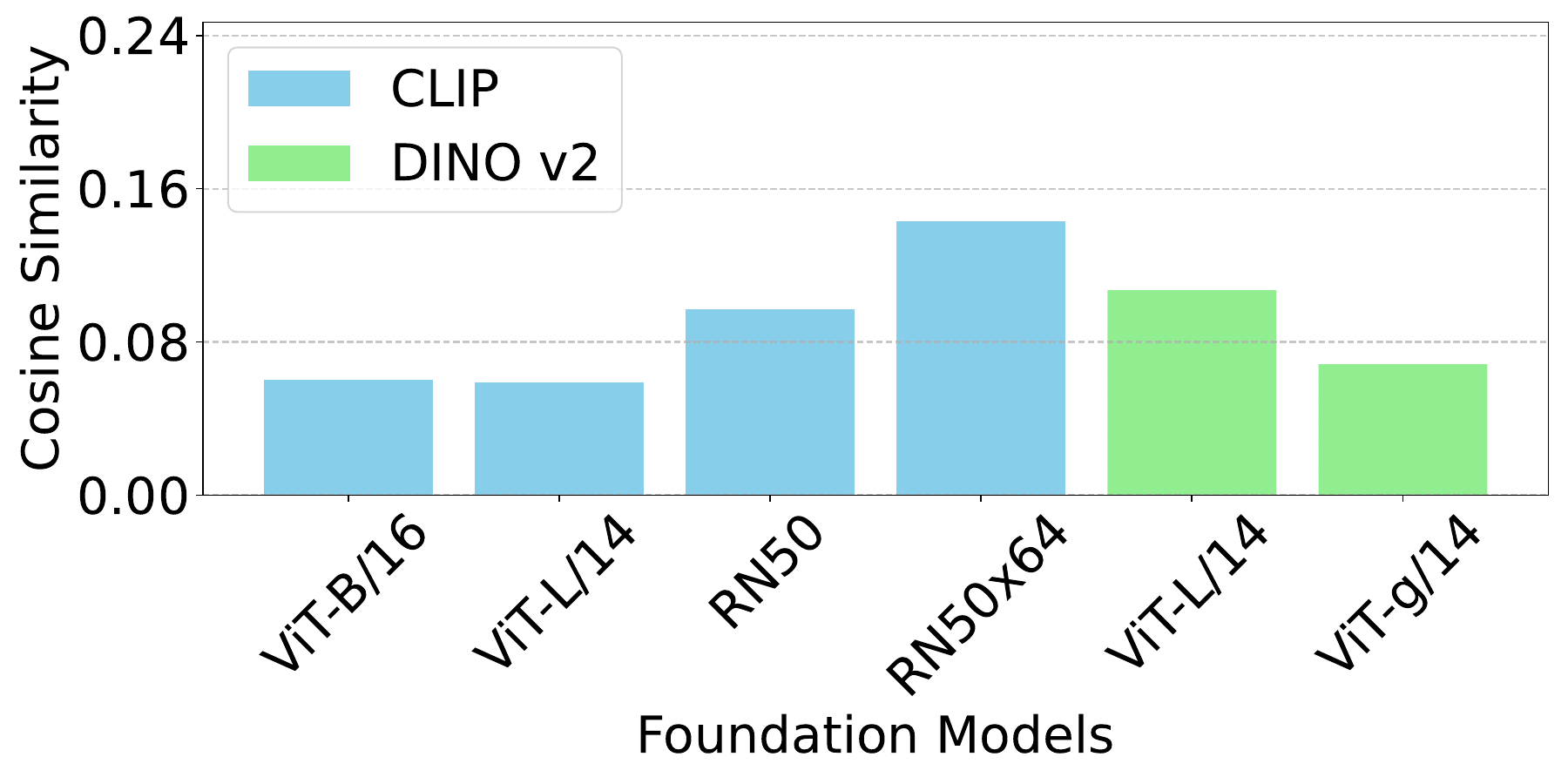}\label{fig:5}}
    \\
    \subfloat[Fog blurring]{\includegraphics[width=0.2\textwidth]{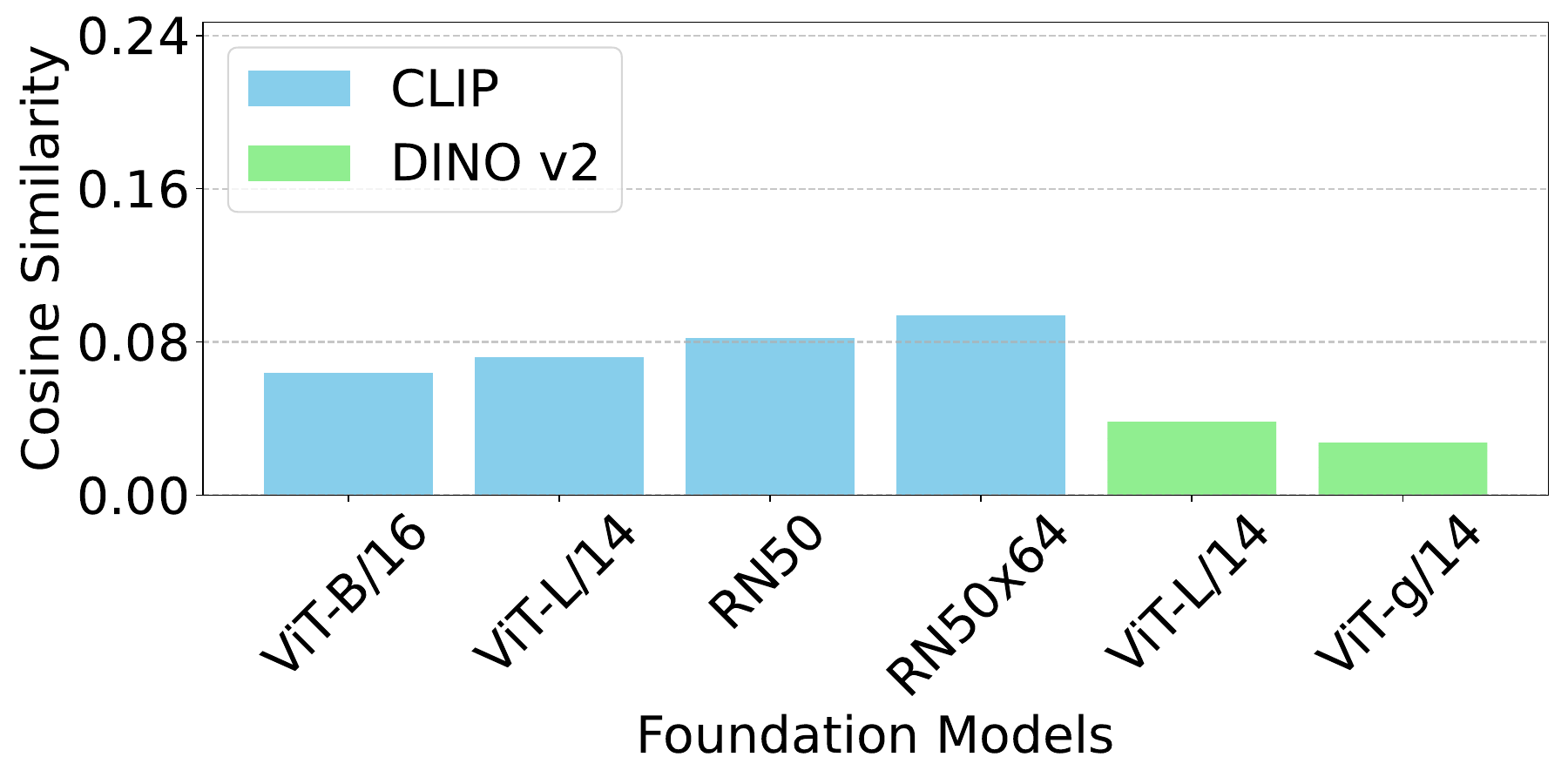}\label{fig:6}}
    \hfill
    \subfloat[Frost blurring]{\includegraphics[width=0.2\textwidth]{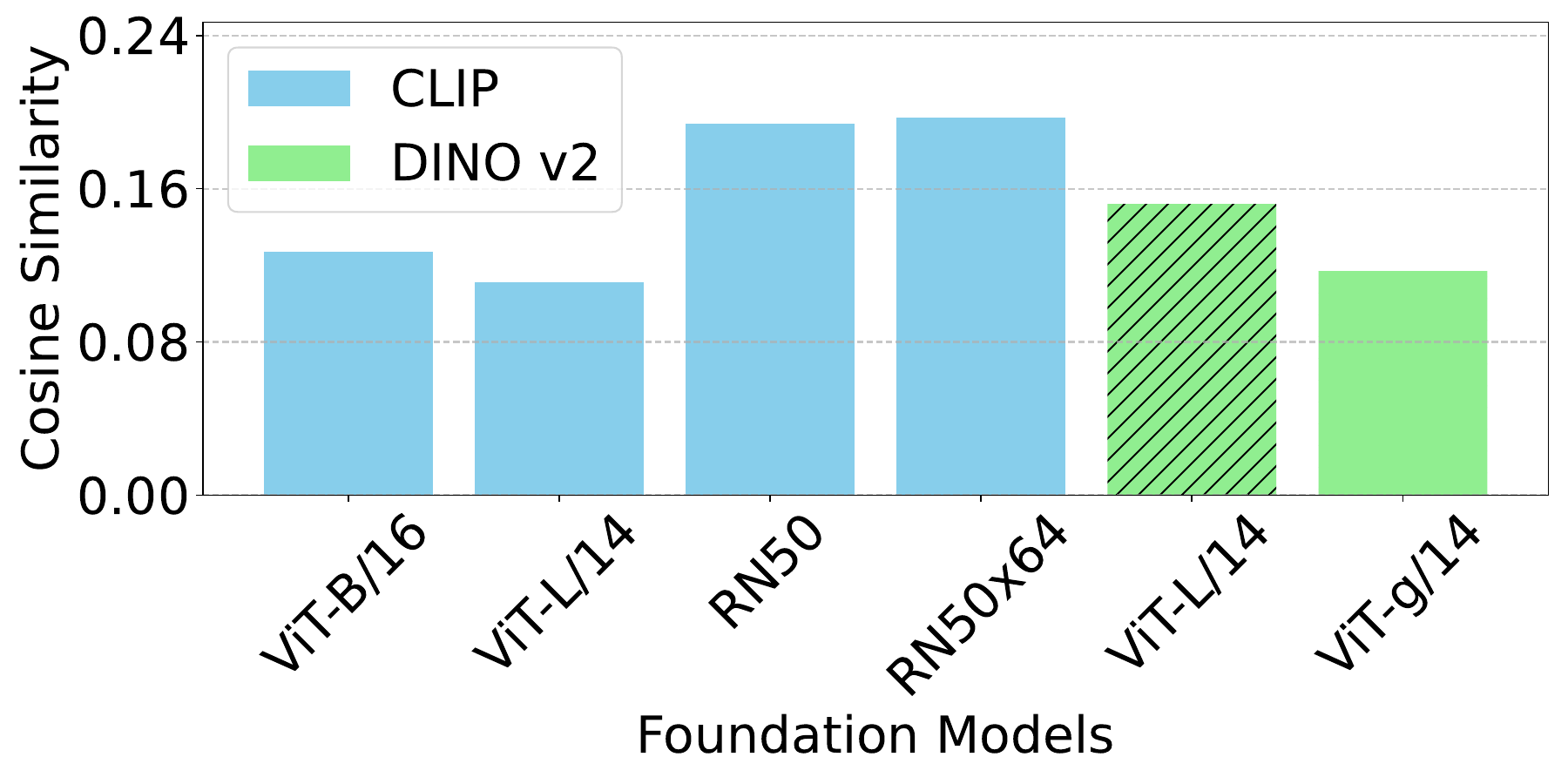}\label{fig:7}}
    \hfill
    \subfloat[Gaussian blurring]{\includegraphics[width=0.2\textwidth]{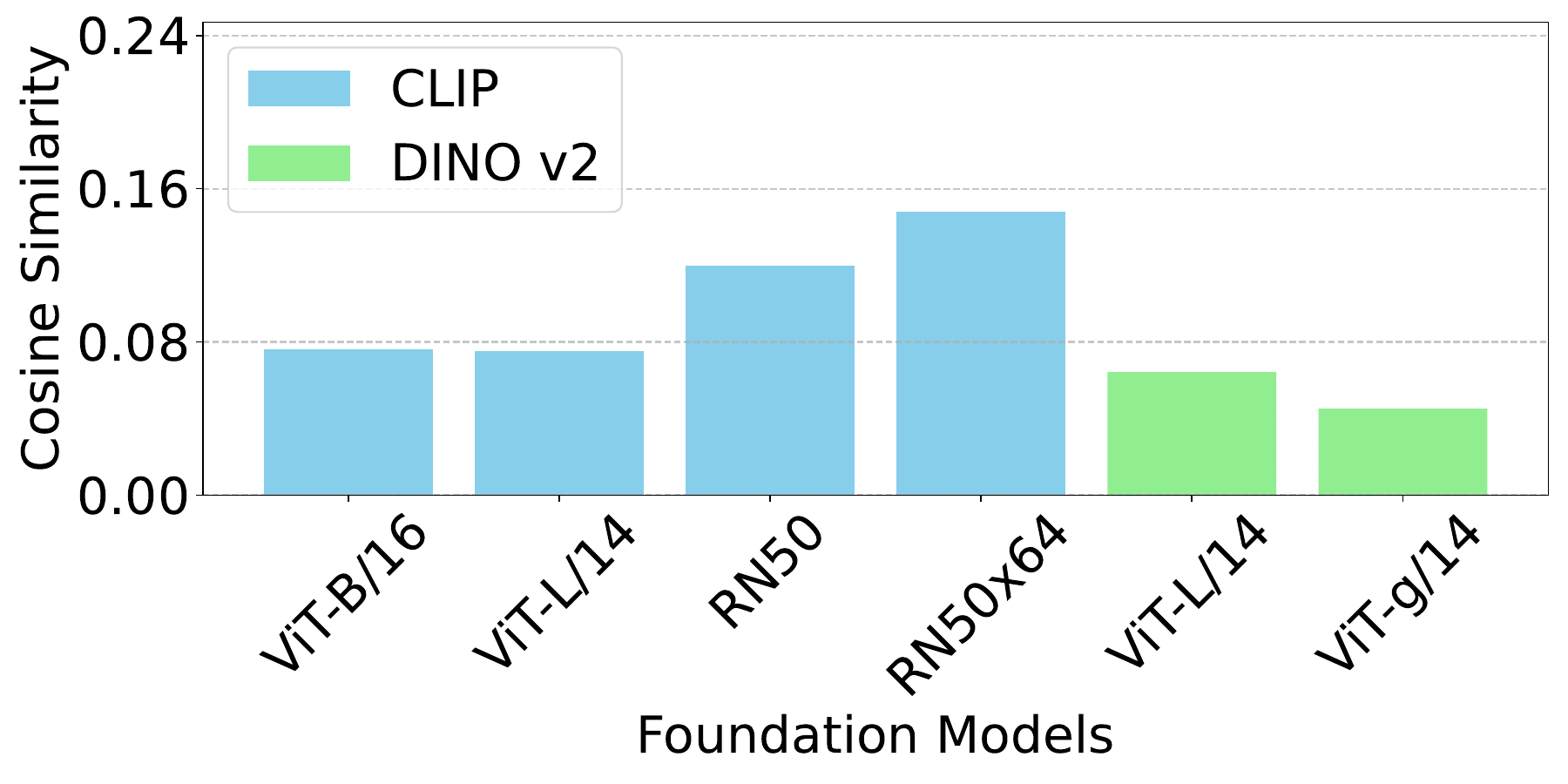}\label{fig:8}}
    \hfill
    \subfloat[Glass blurring]{\includegraphics[width=0.2\textwidth]{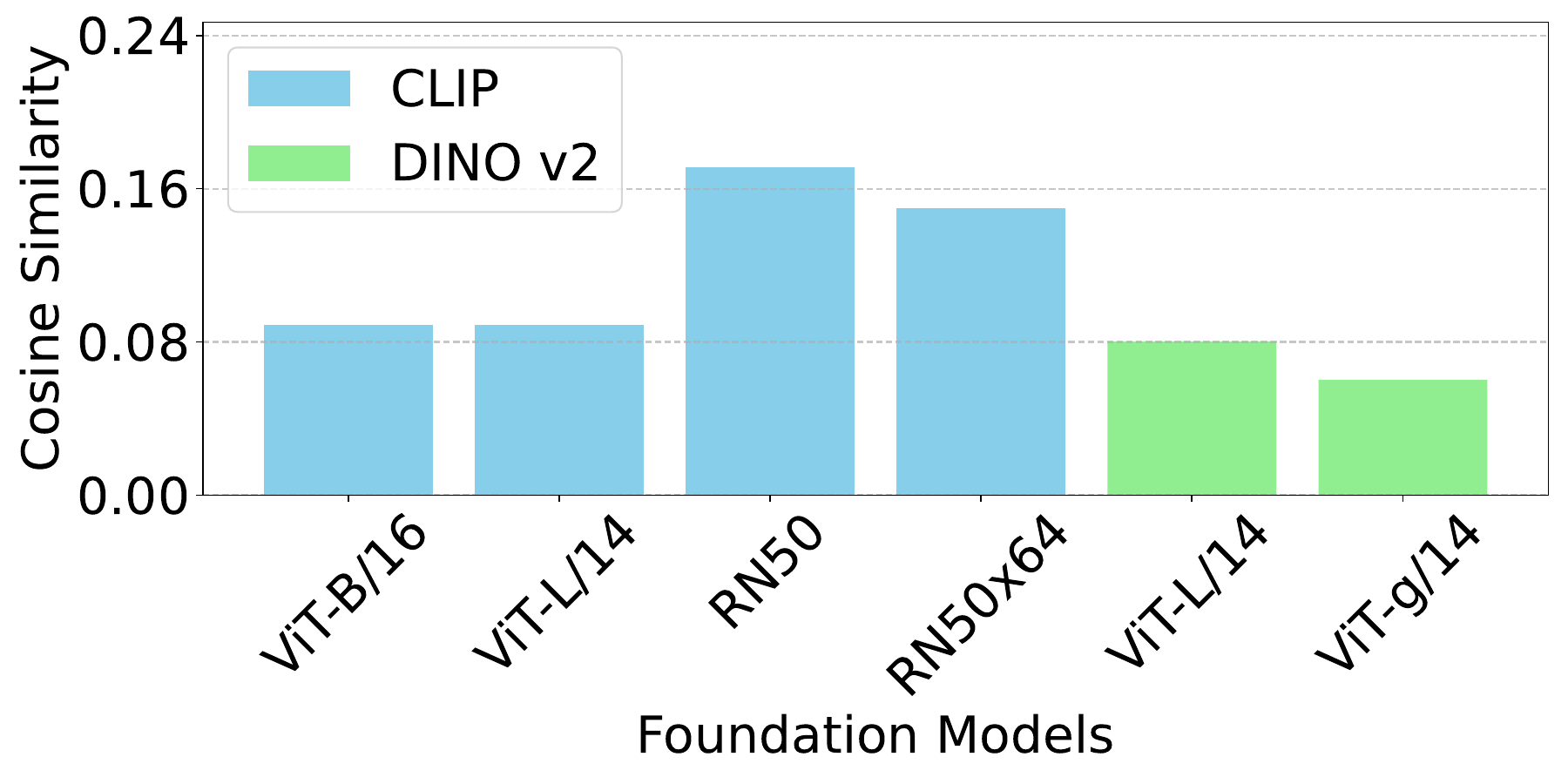}}
    \caption{Average cosine similarity of Food101 testing images for different foundation models and  perturbation functions.}
    \label{fig:cosine_vision_Food101}
\end{figure*}

\begin{figure*}[!h]
    \centering
    \subfloat[JPEG Compression]{\includegraphics[width=0.2\textwidth]{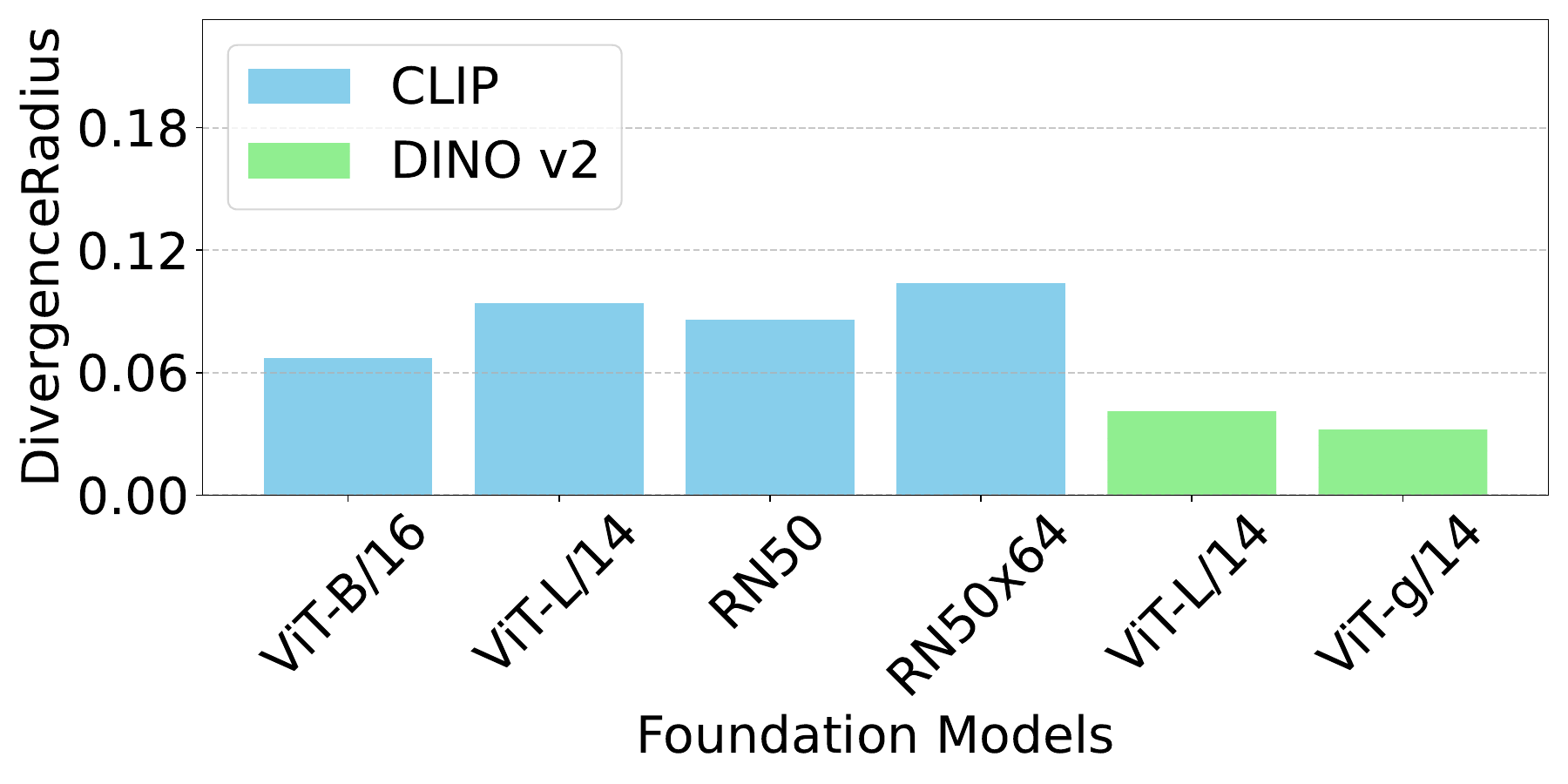}\label{fig:1}}
    \hfill
    \subfloat[Brightness adjustment]{\includegraphics[width=0.2\textwidth]{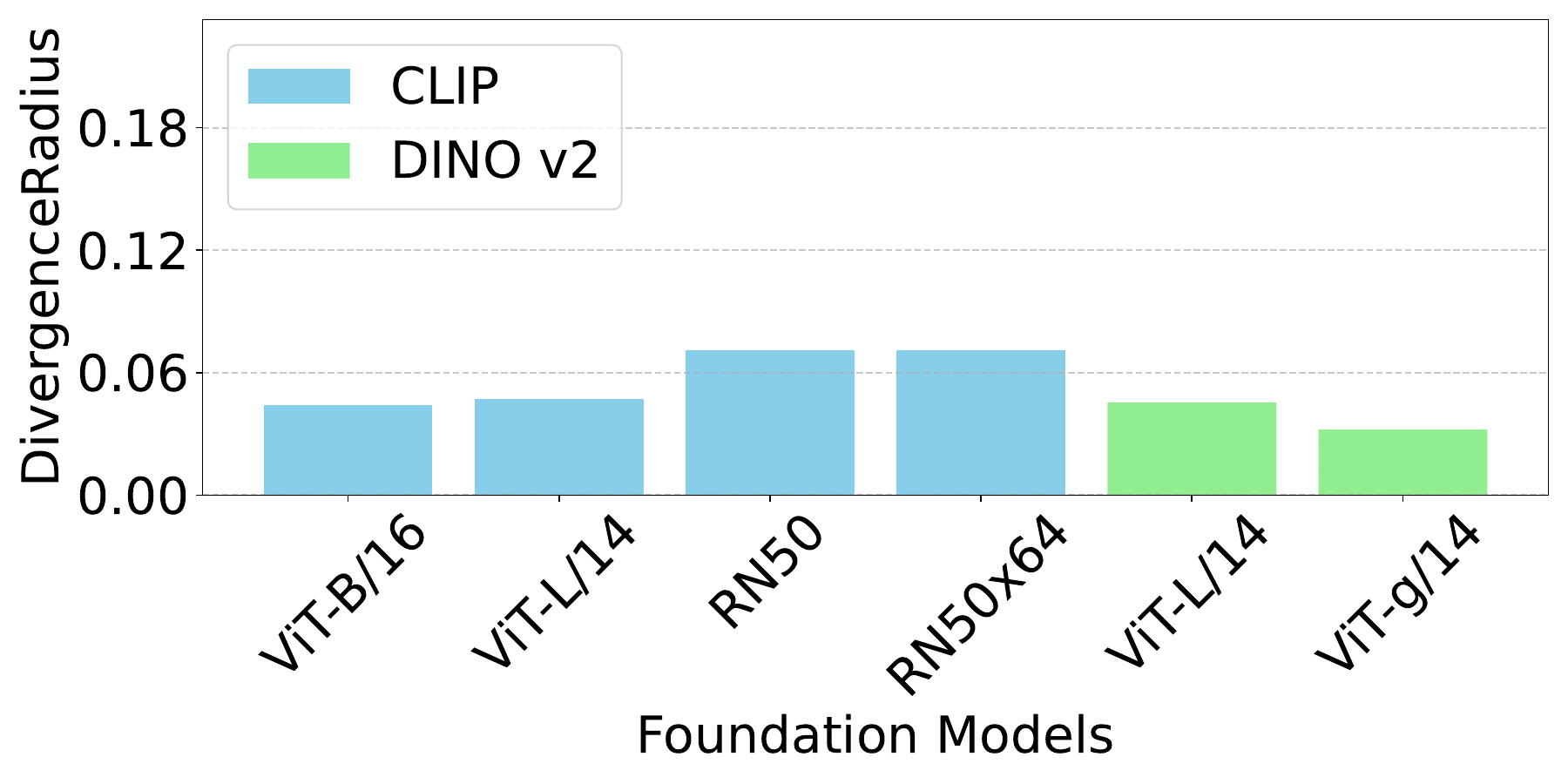}\label{fig:2}}
    \hfill
    \subfloat[Contrast adjustment]{\includegraphics[width=0.2\textwidth]{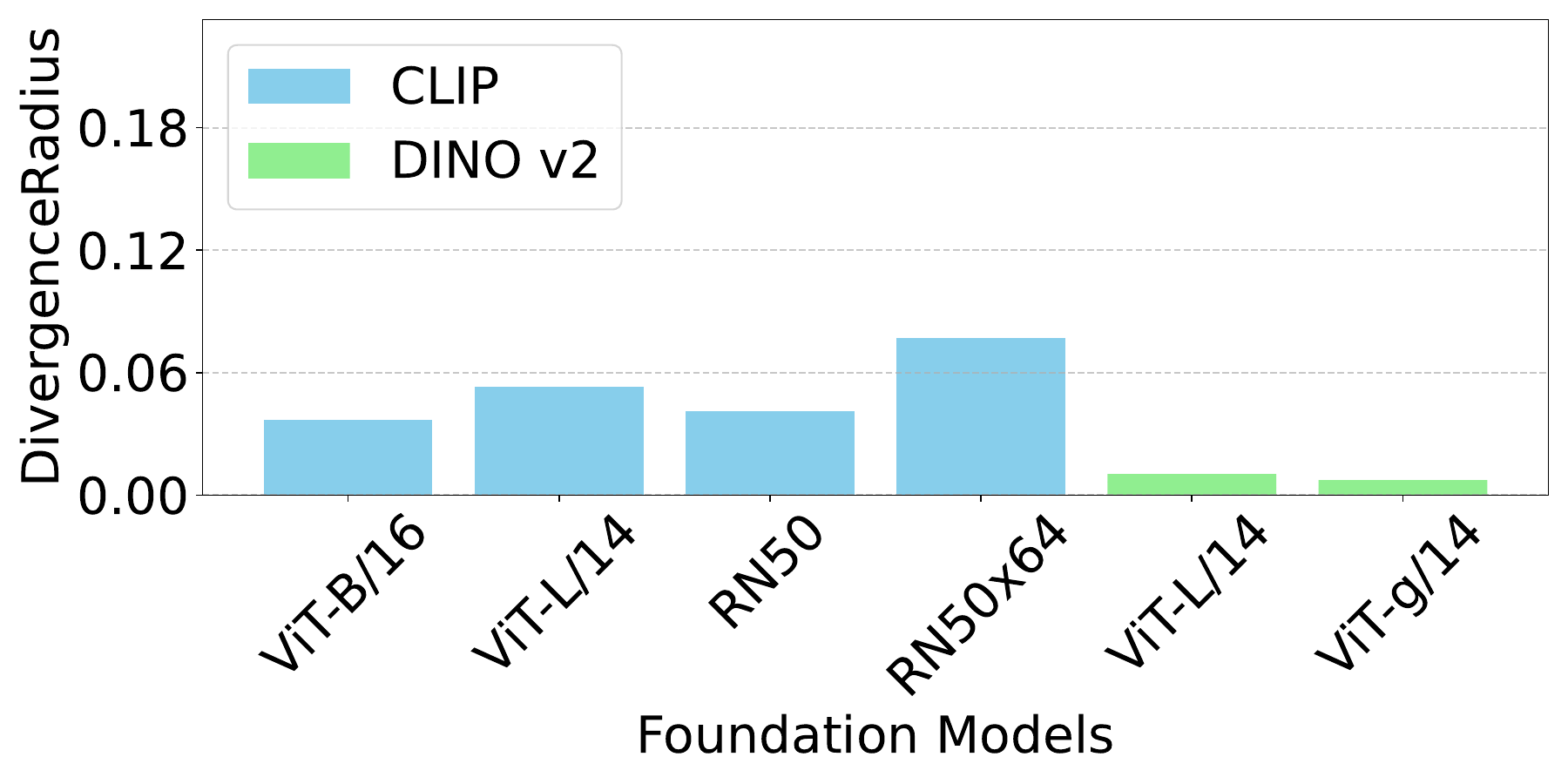}\label{fig:3}}
    \hfill
    \subfloat[Defocus blurring]{\includegraphics[width=0.2\textwidth]{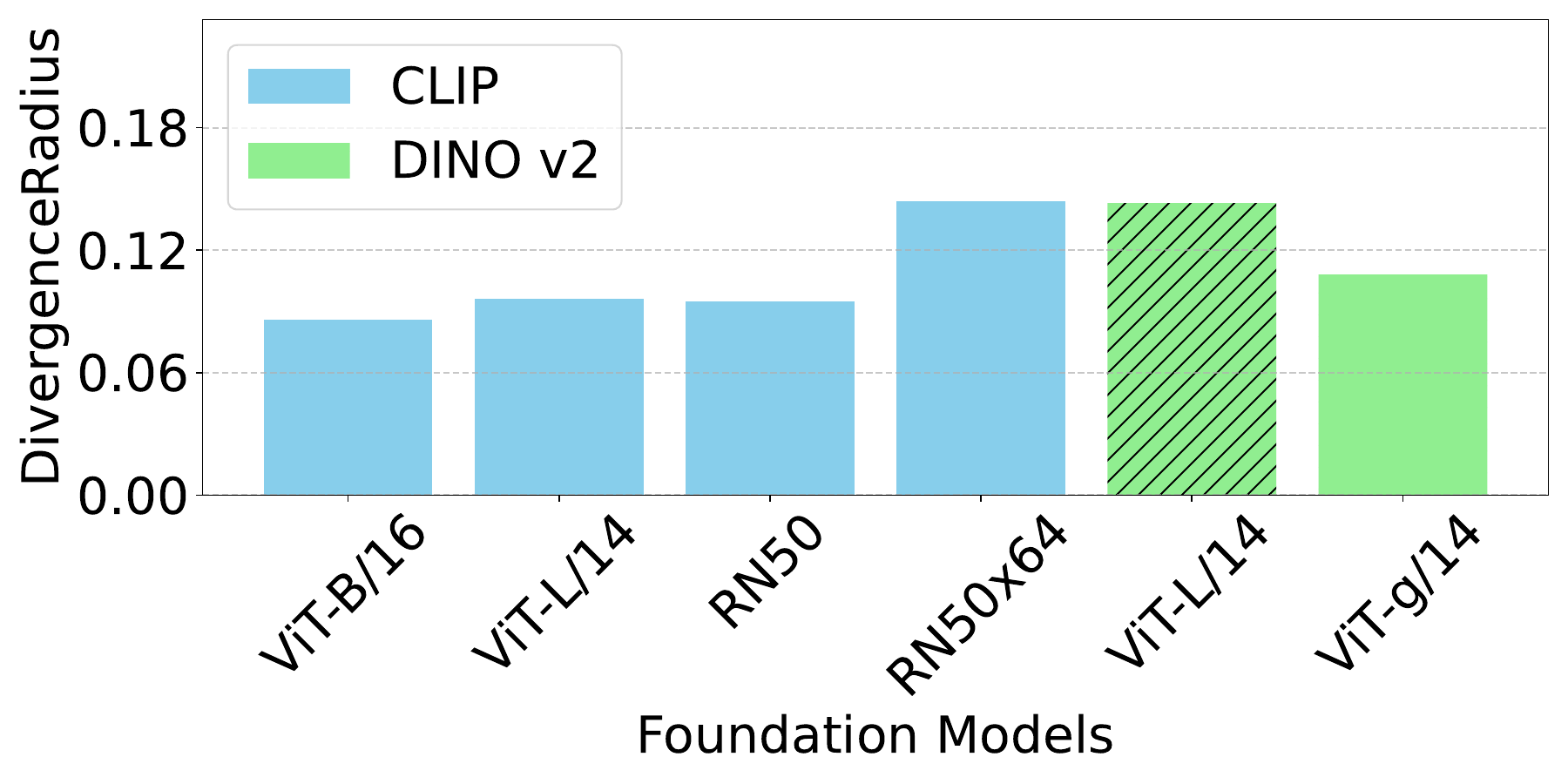}\label{fig:4}}
    \hfill
    \subfloat[Elastic blurring]{\includegraphics[width=0.2\textwidth]{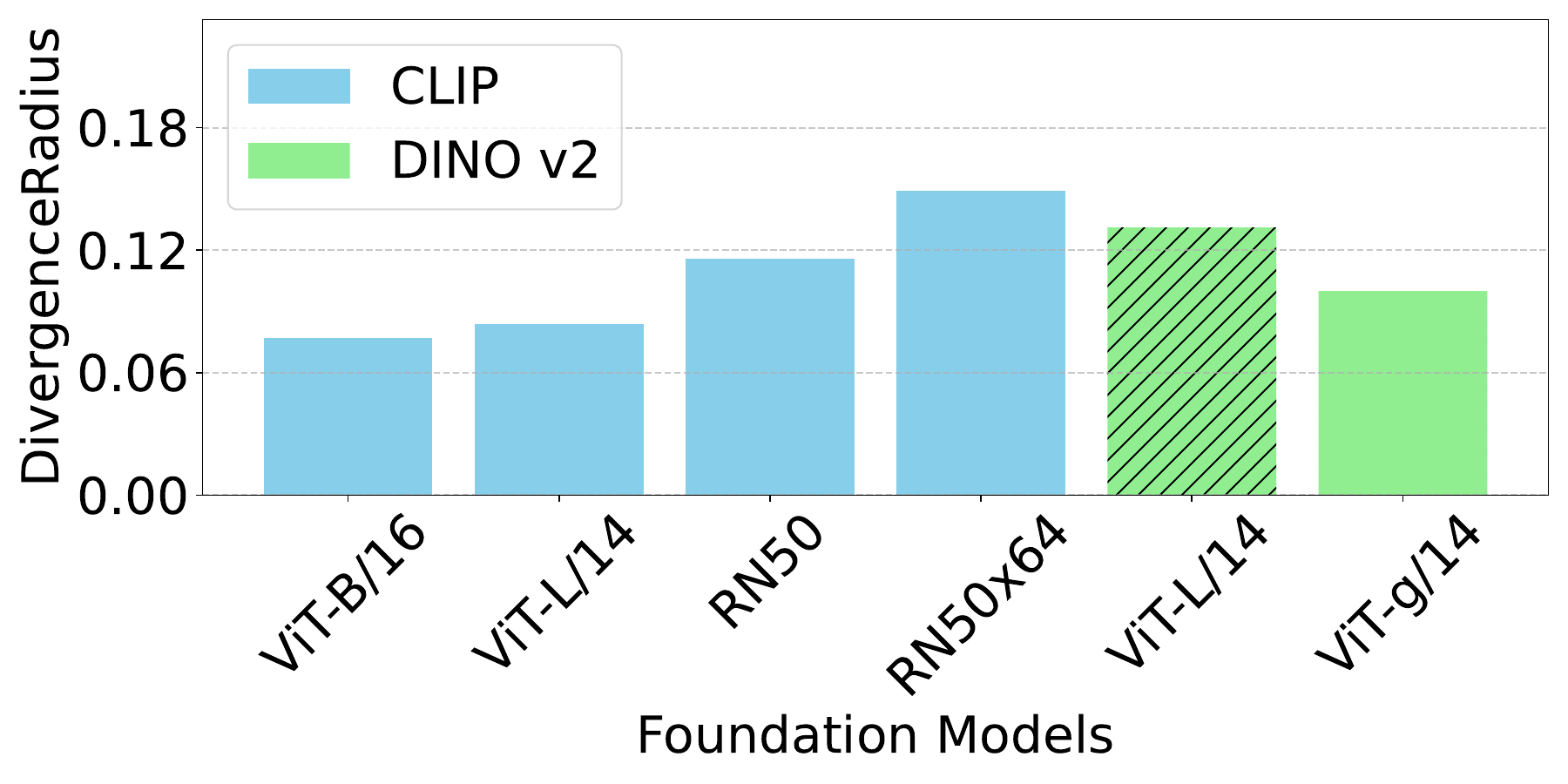}\label{fig:5}}
    \\
    \subfloat[Fog blurring]{\includegraphics[width=0.2\textwidth]{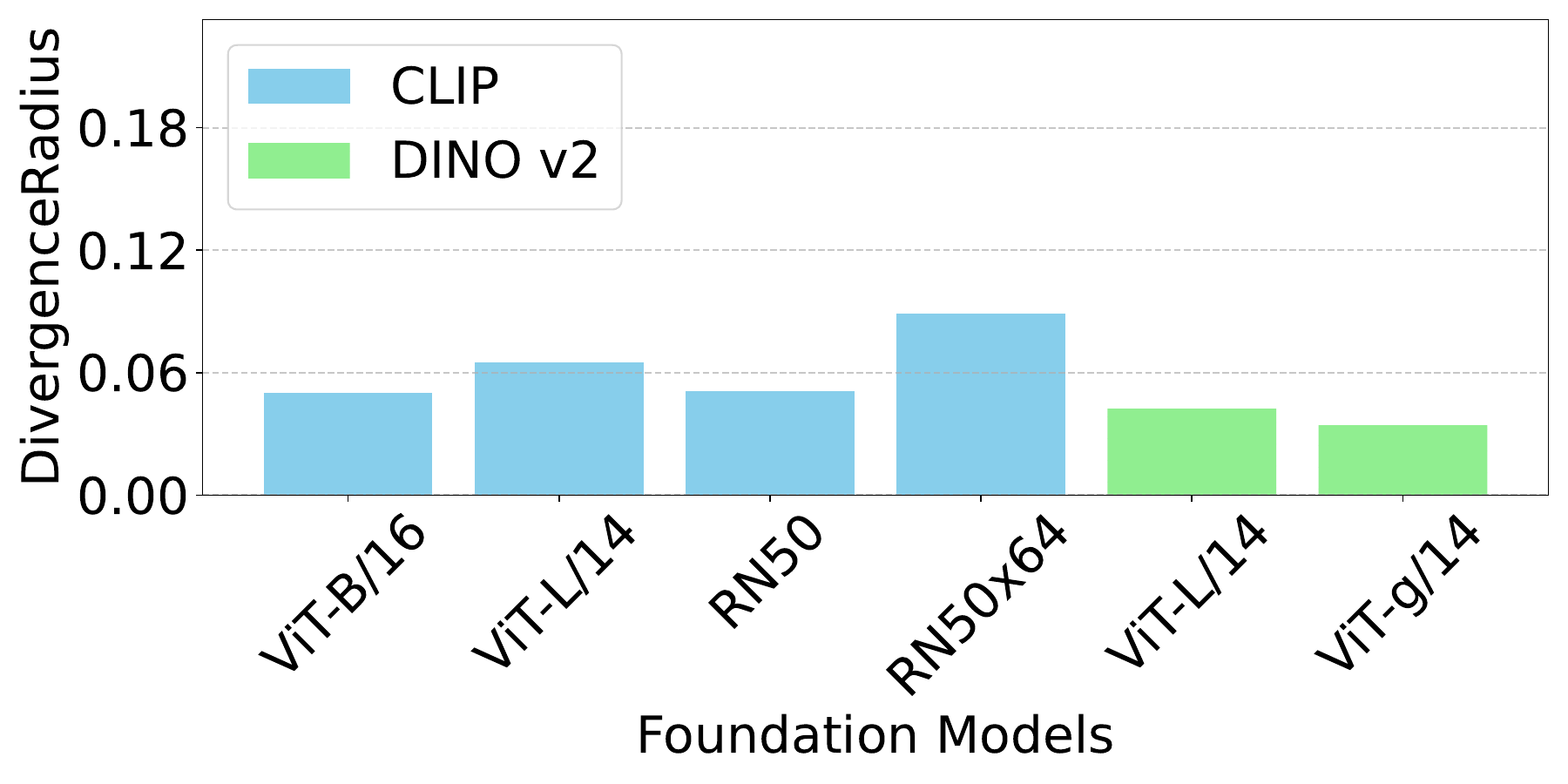}\label{fig:6}}
    \hfill
    \subfloat[Frost blurring]{\includegraphics[width=0.2\textwidth]{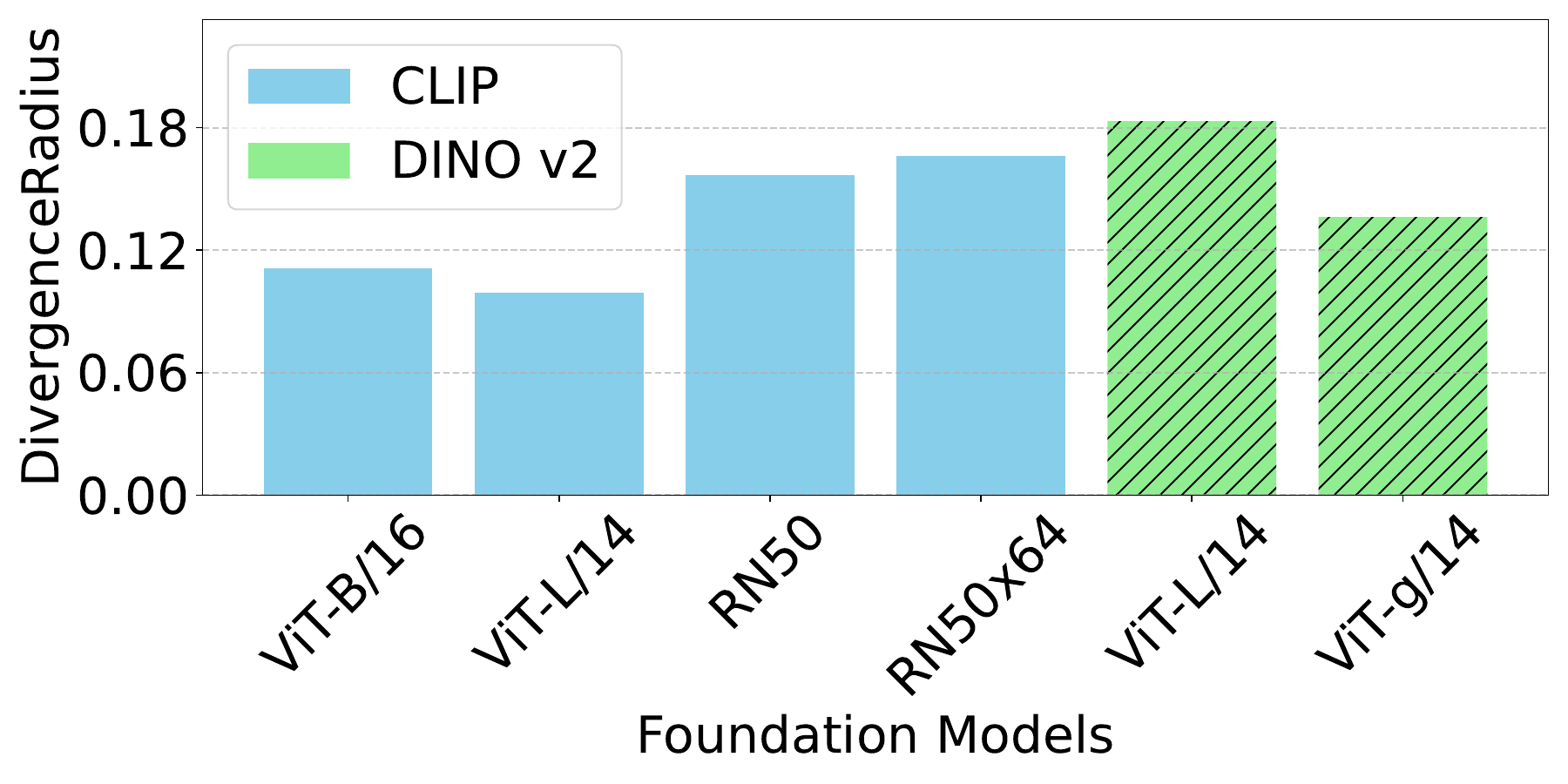}\label{fig:7}}
    \hfill
    \subfloat[Gaussian blurring]{\includegraphics[width=0.2\textwidth]{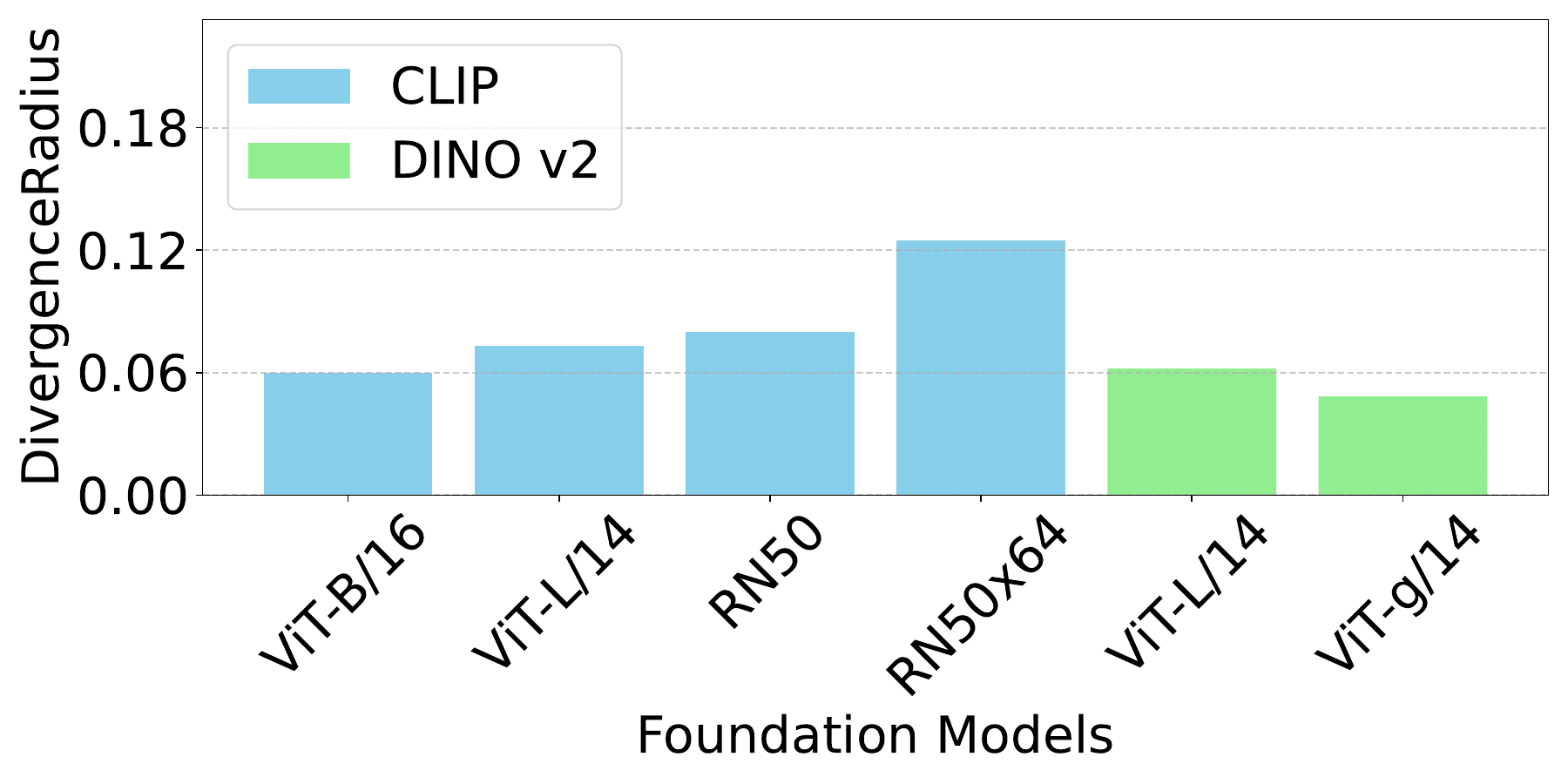}\label{fig:8}}
    \hfill
    \subfloat[Glass blurring]{\includegraphics[width=0.2\textwidth]{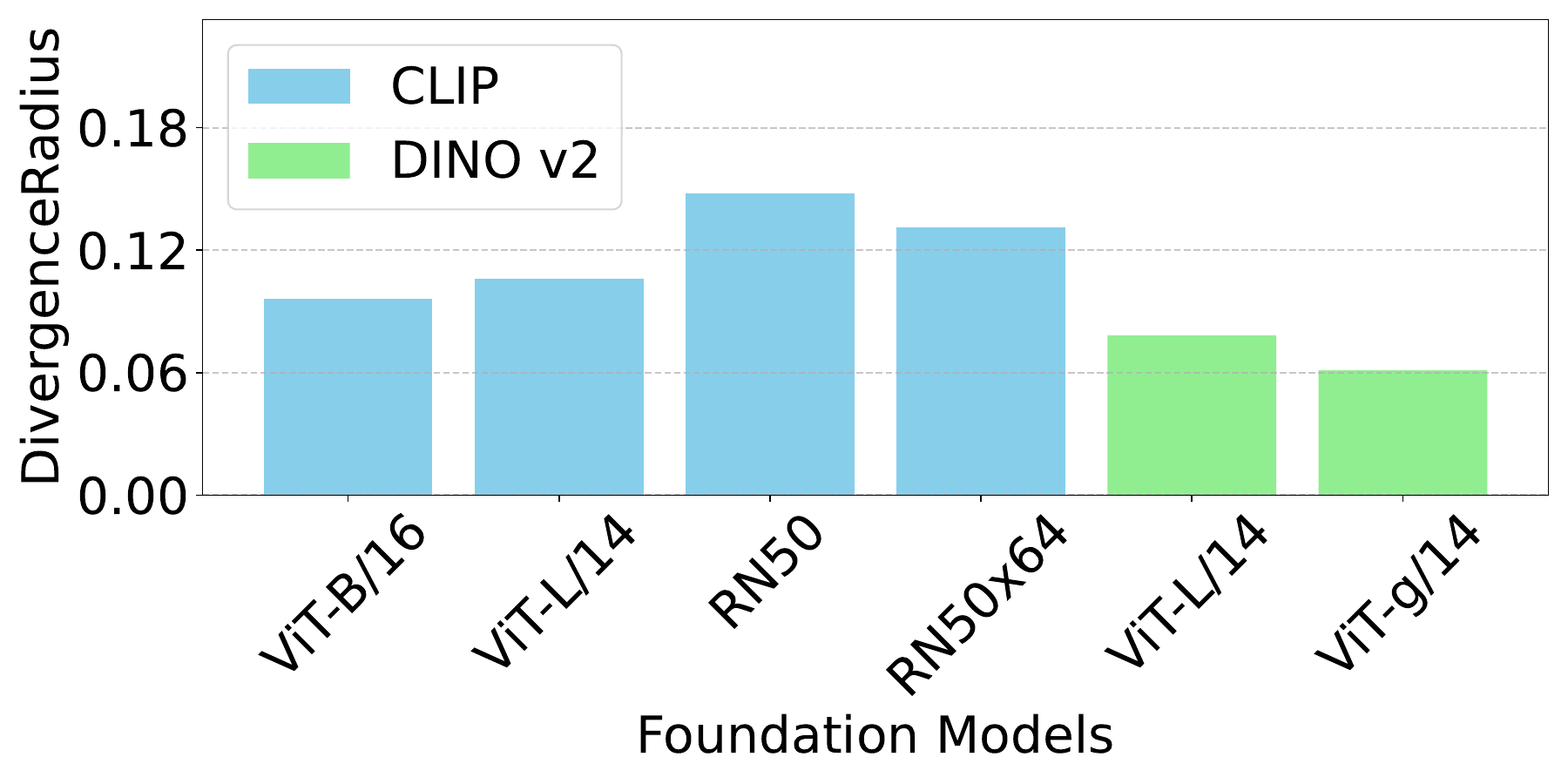}}
    \caption{Average \alg{} of NYU-Depth V2 testing images for different foundation models and  perturbation functions.}
    \label{fig:divergence_radius_vision_nyu}
\end{figure*}

\begin{figure*}[!h]
    \centering
    \subfloat[JPEG Compression]{\includegraphics[width=0.2\textwidth]{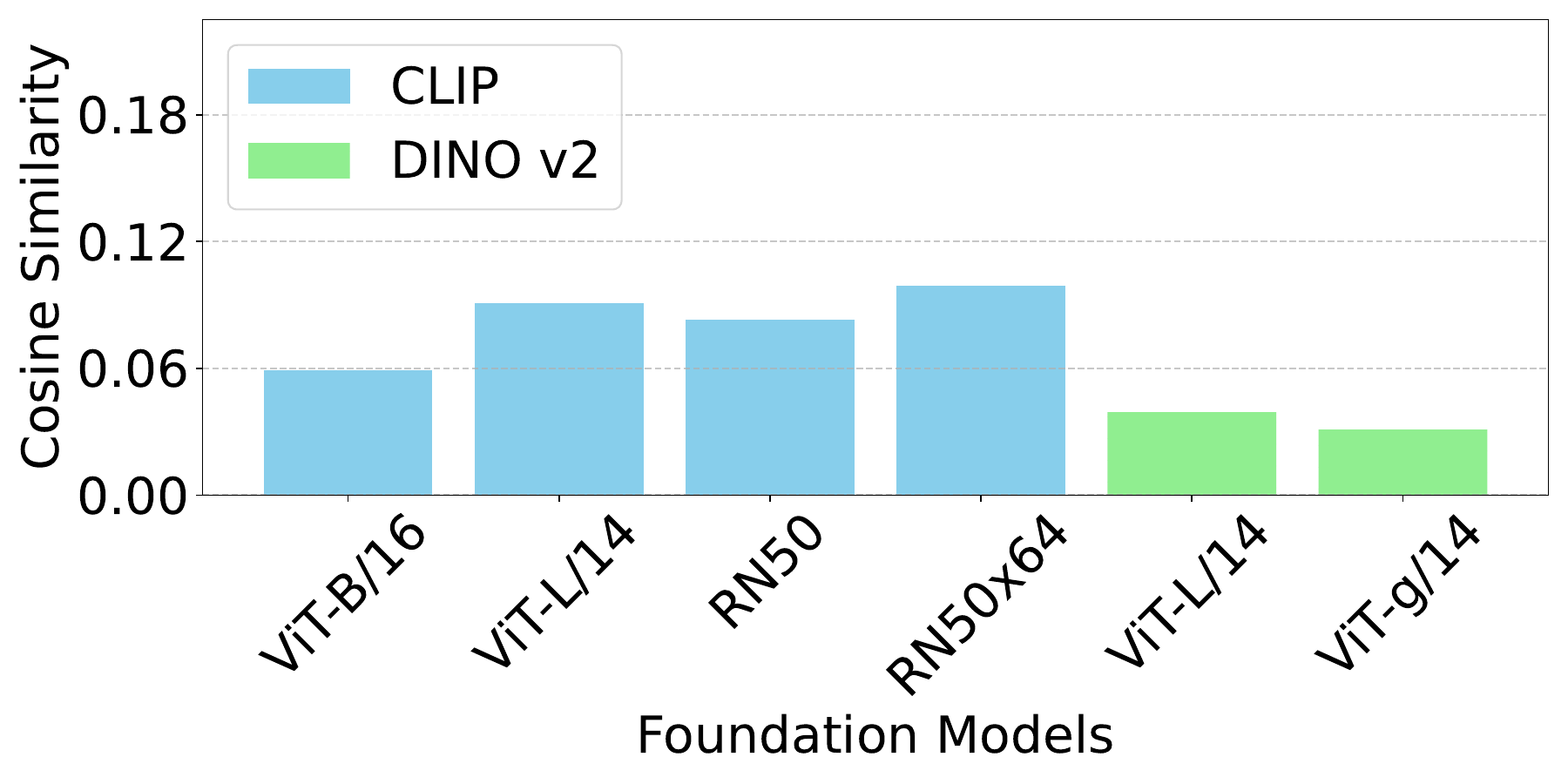}\label{fig:1}}
    \hfill
    \subfloat[Brightness adjustment]{\includegraphics[width=0.2\textwidth]{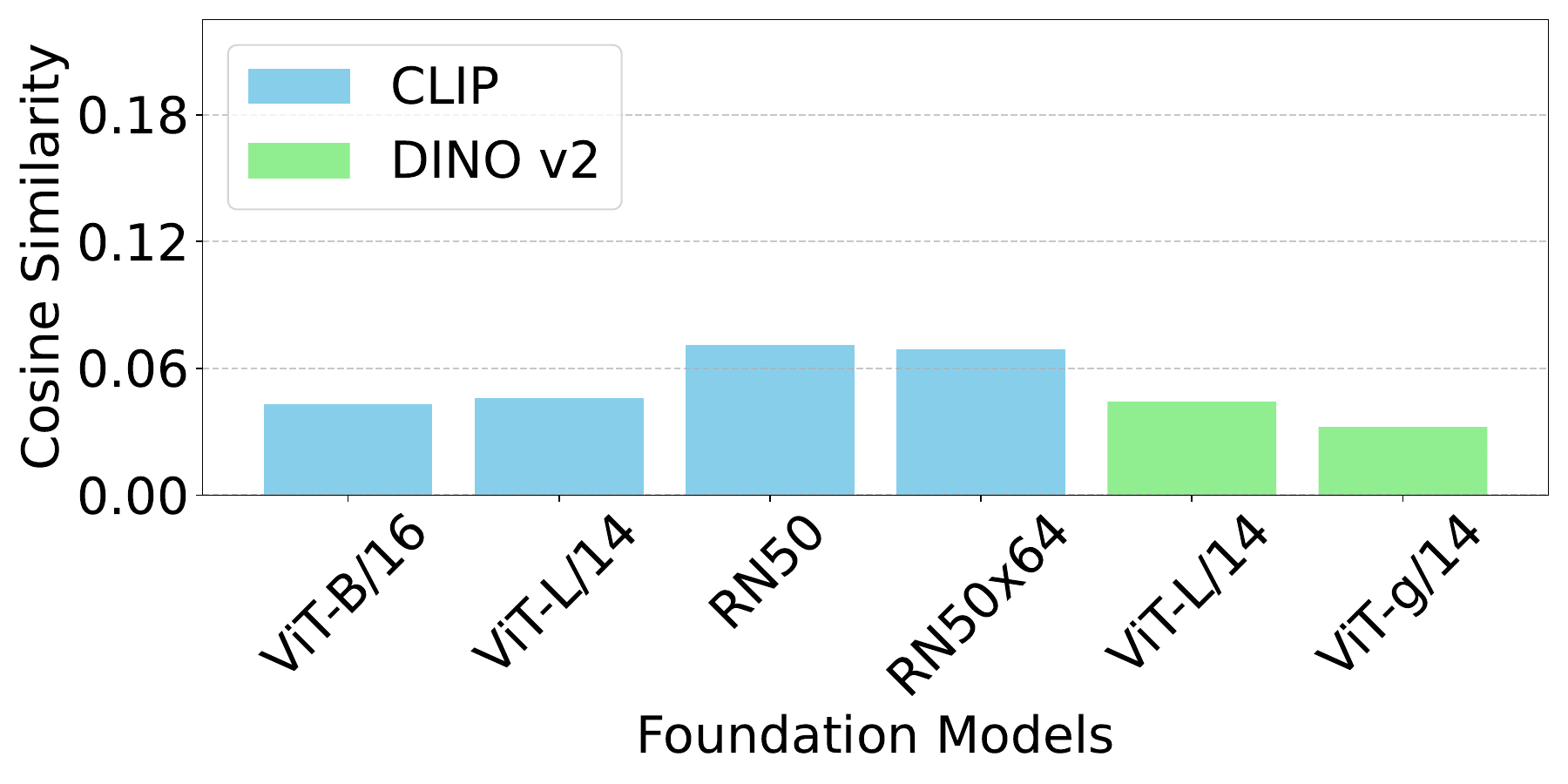}\label{fig:2}}
    \hfill
    \subfloat[Contrast adjustment]{\includegraphics[width=0.2\textwidth]{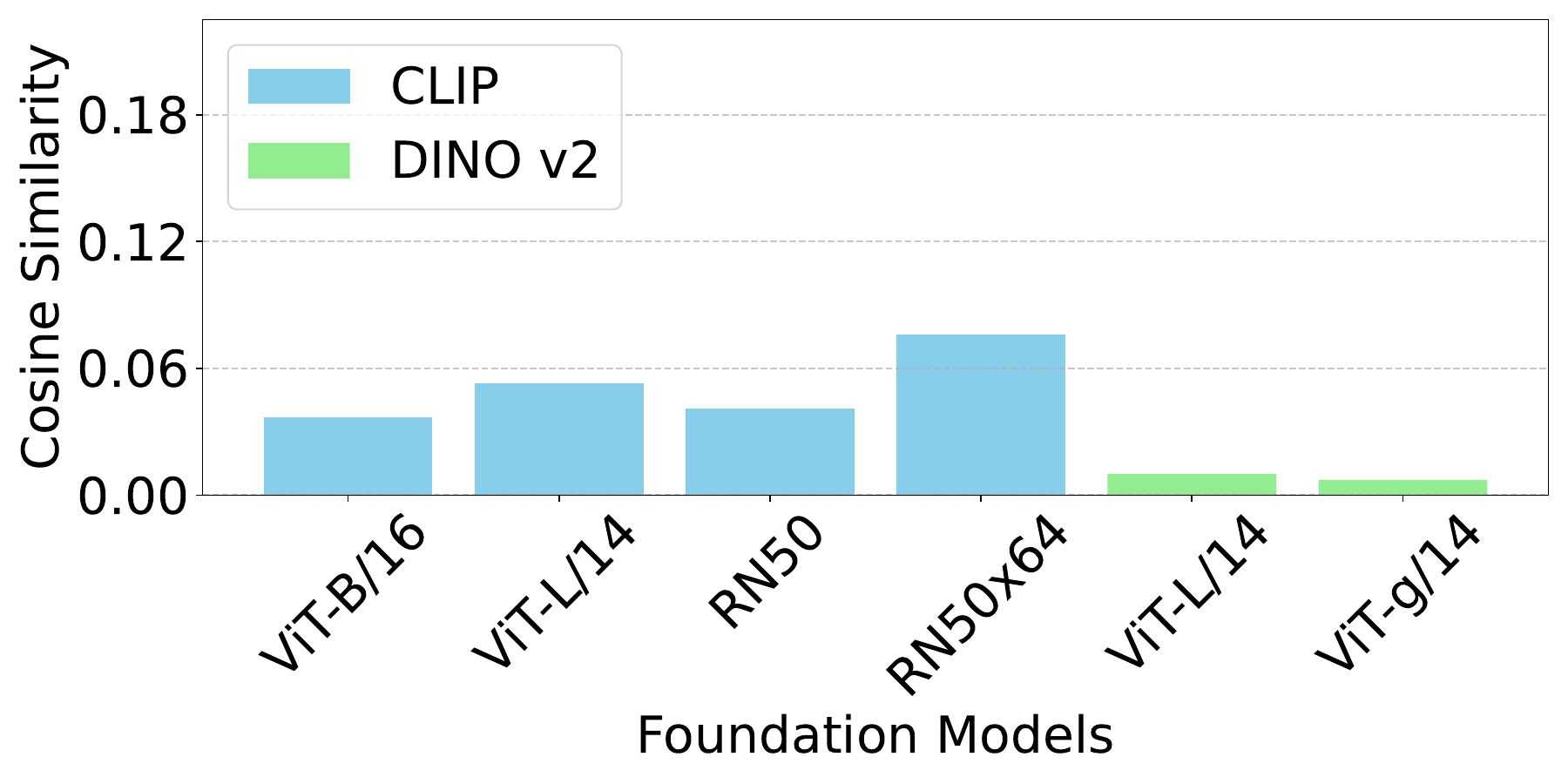}\label{fig:3}}
    \hfill
    \subfloat[Defocus blurring]{\includegraphics[width=0.2\textwidth]{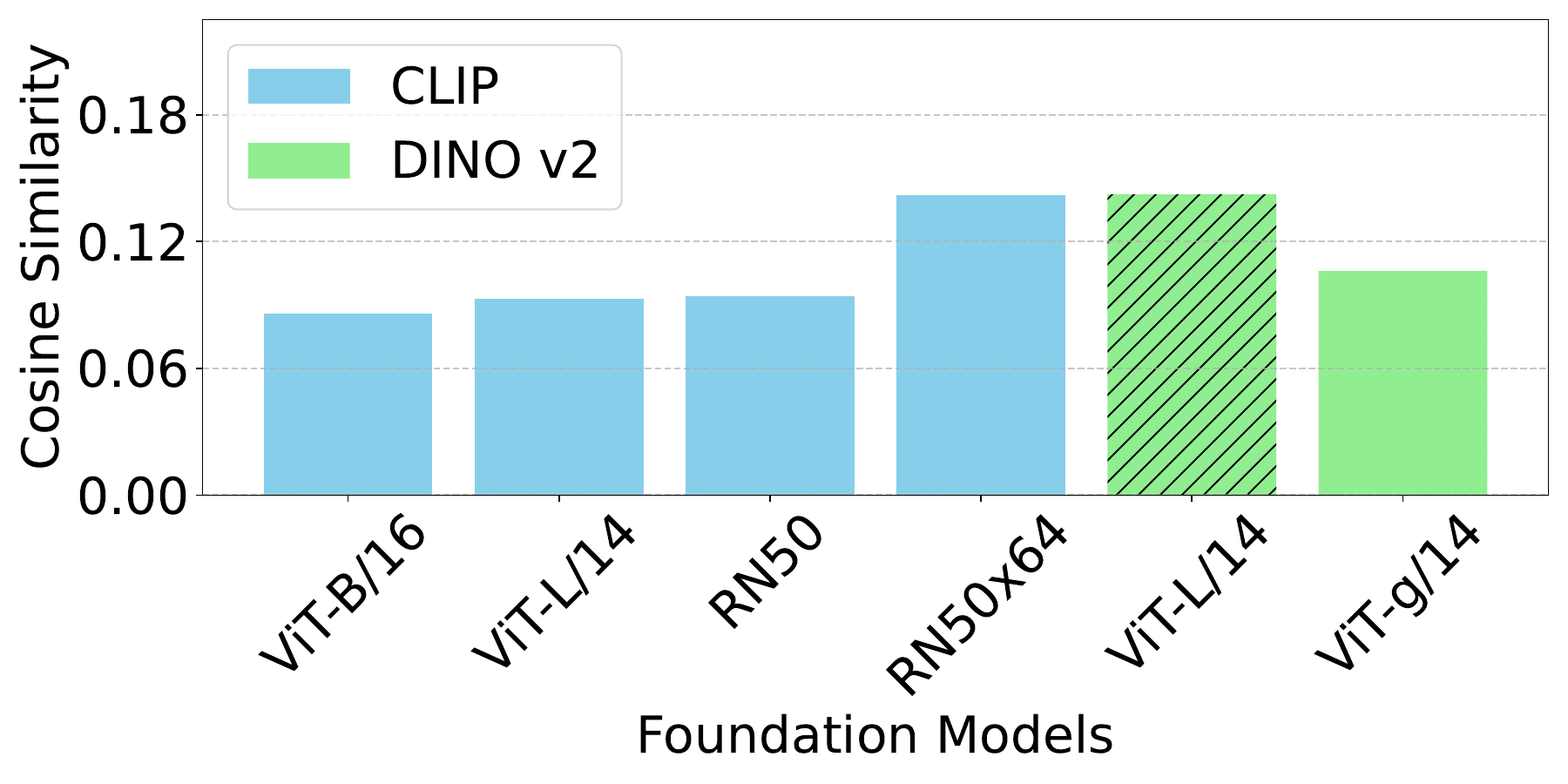}\label{fig:4}}
    \hfill
    \subfloat[Elastic blurring]{\includegraphics[width=0.2\textwidth]{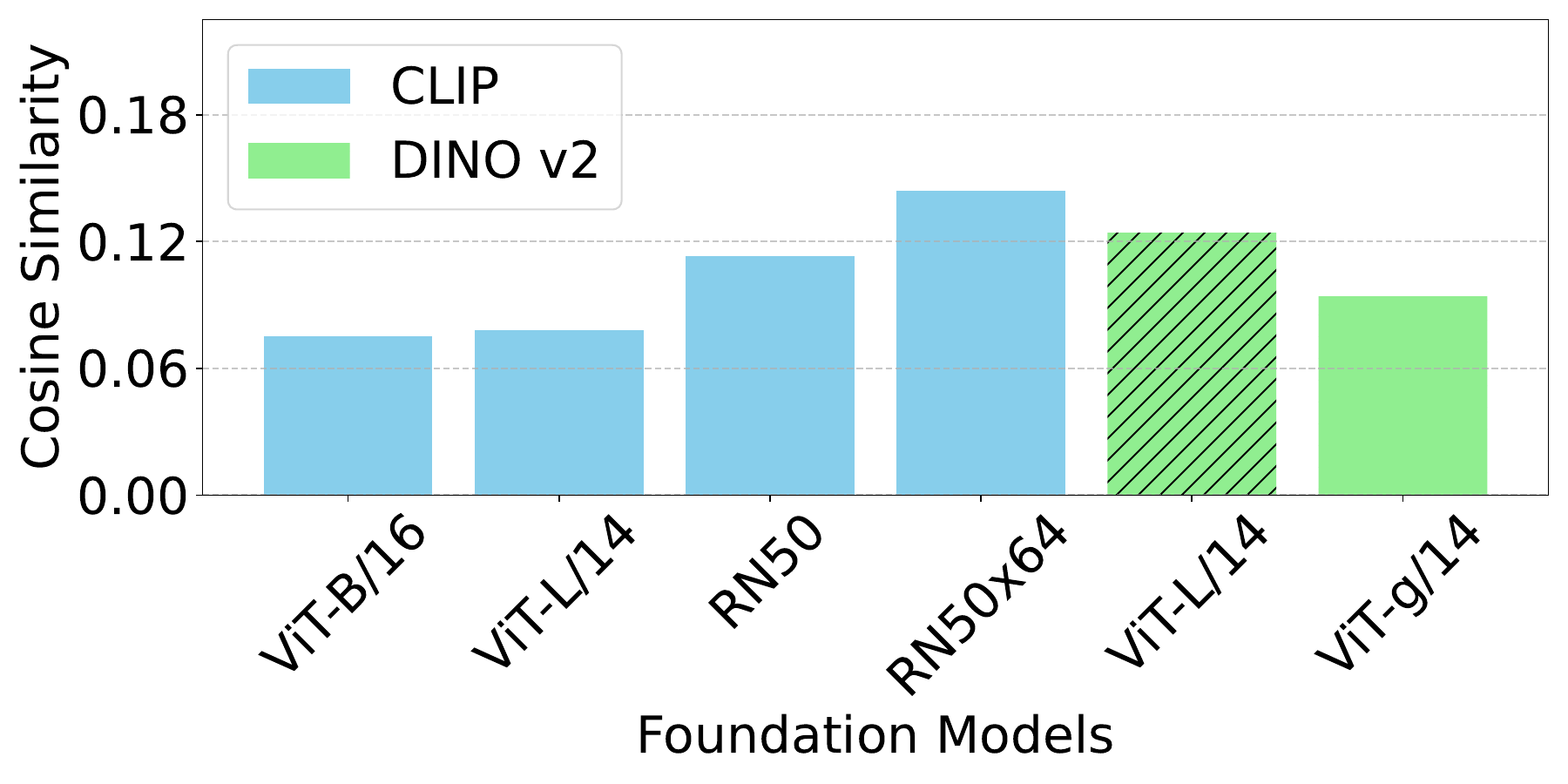}\label{fig:5}}
    \\
    \subfloat[Fog blurring]{\includegraphics[width=0.2\textwidth]{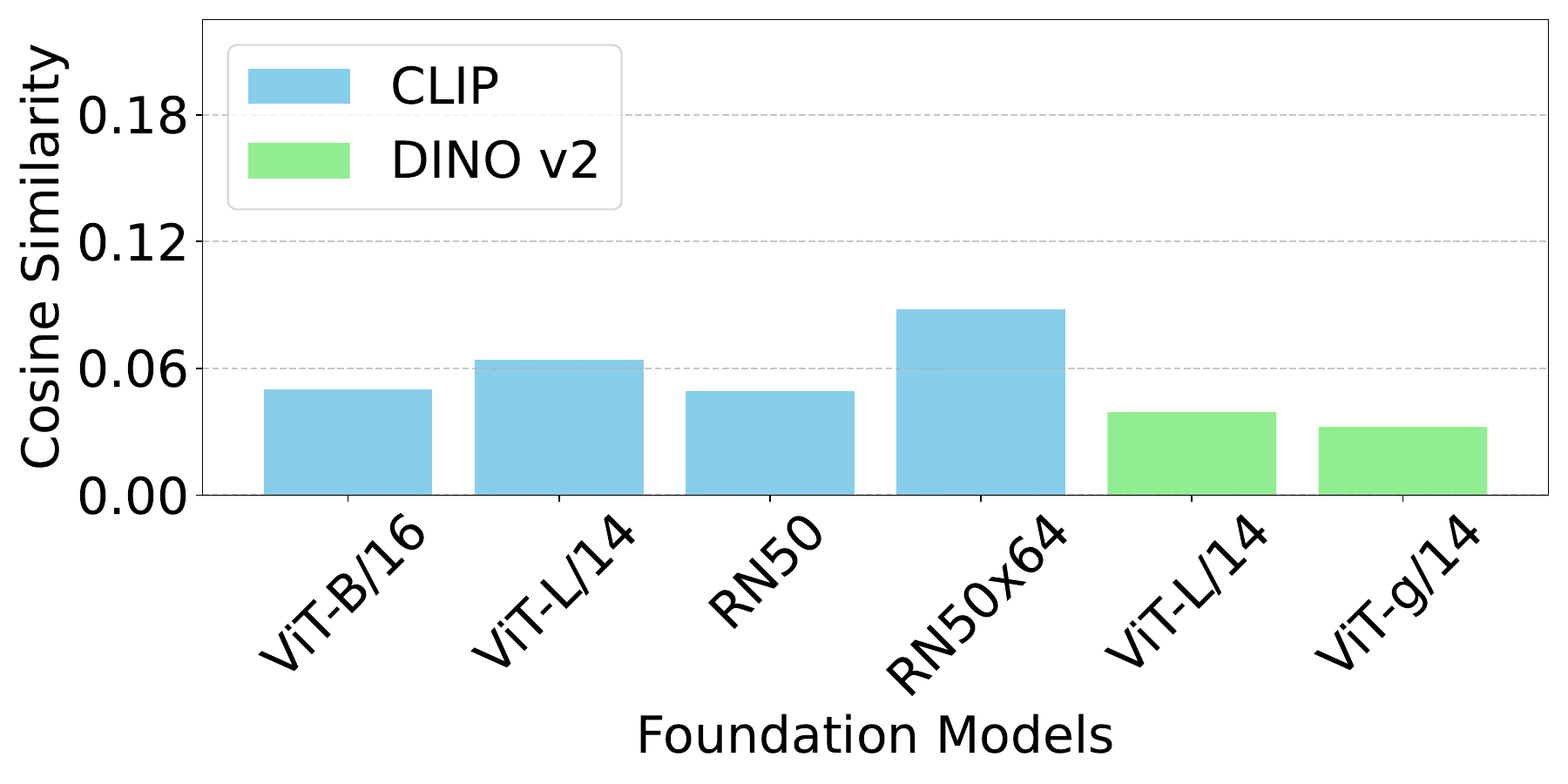}\label{fig:6}}
    \hfill
    \subfloat[Frost blurring]{\includegraphics[width=0.2\textwidth]{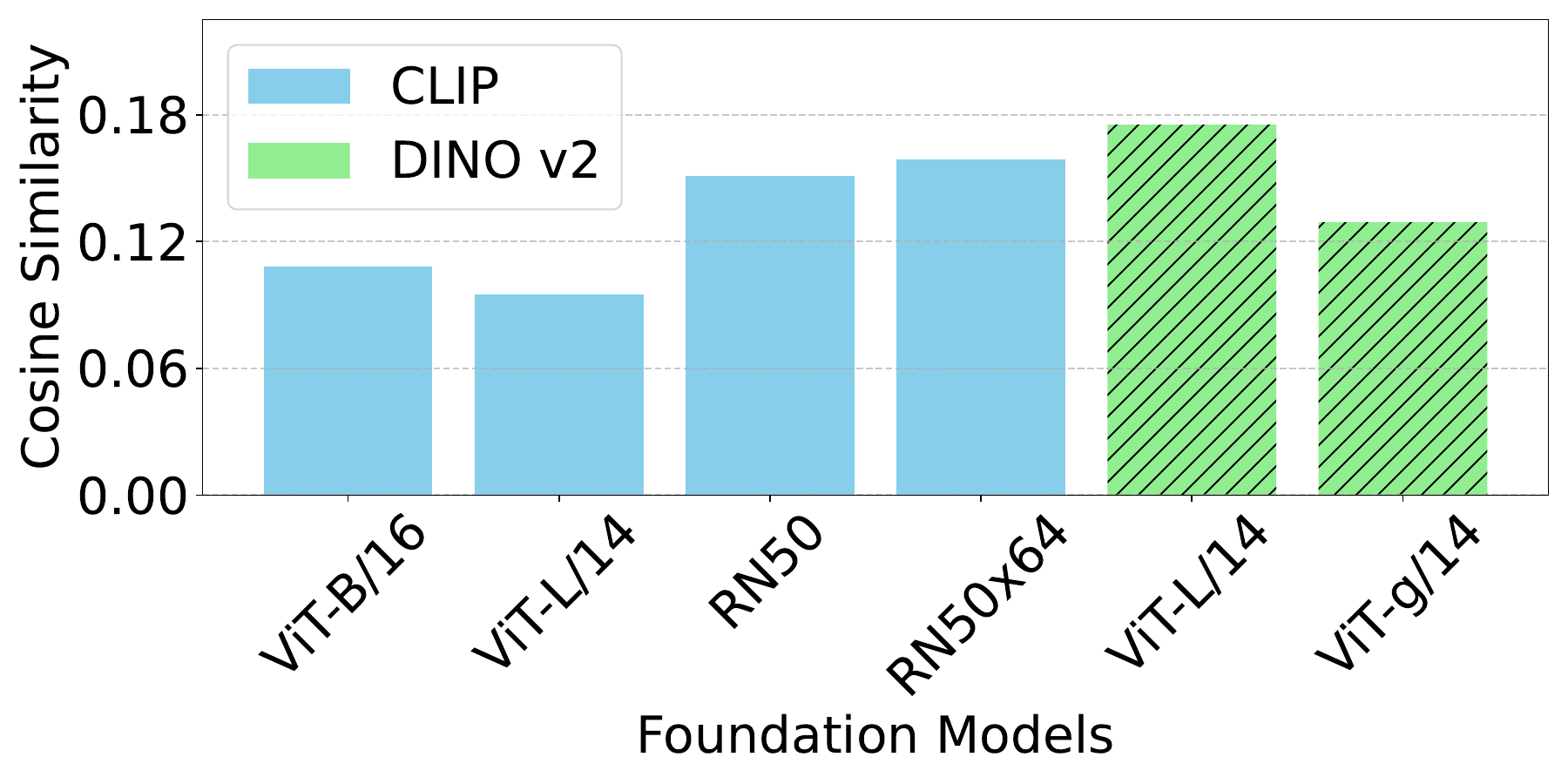}\label{fig:7}}
    \hfill
    \subfloat[Gaussian blurring]{\includegraphics[width=0.2\textwidth]{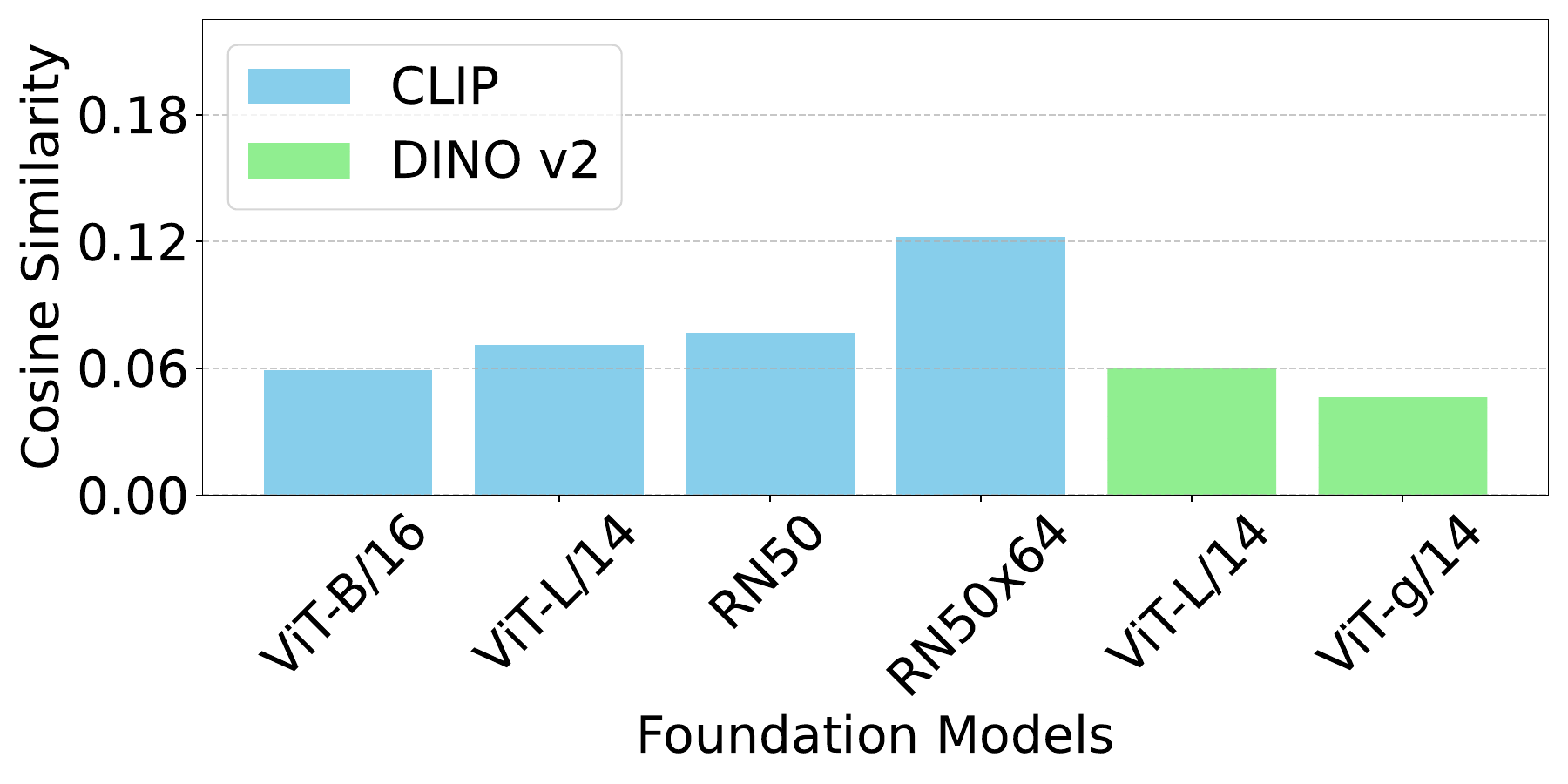}\label{fig:8}}
    \hfill
    \subfloat[Glass blurring]{\includegraphics[width=0.2\textwidth]{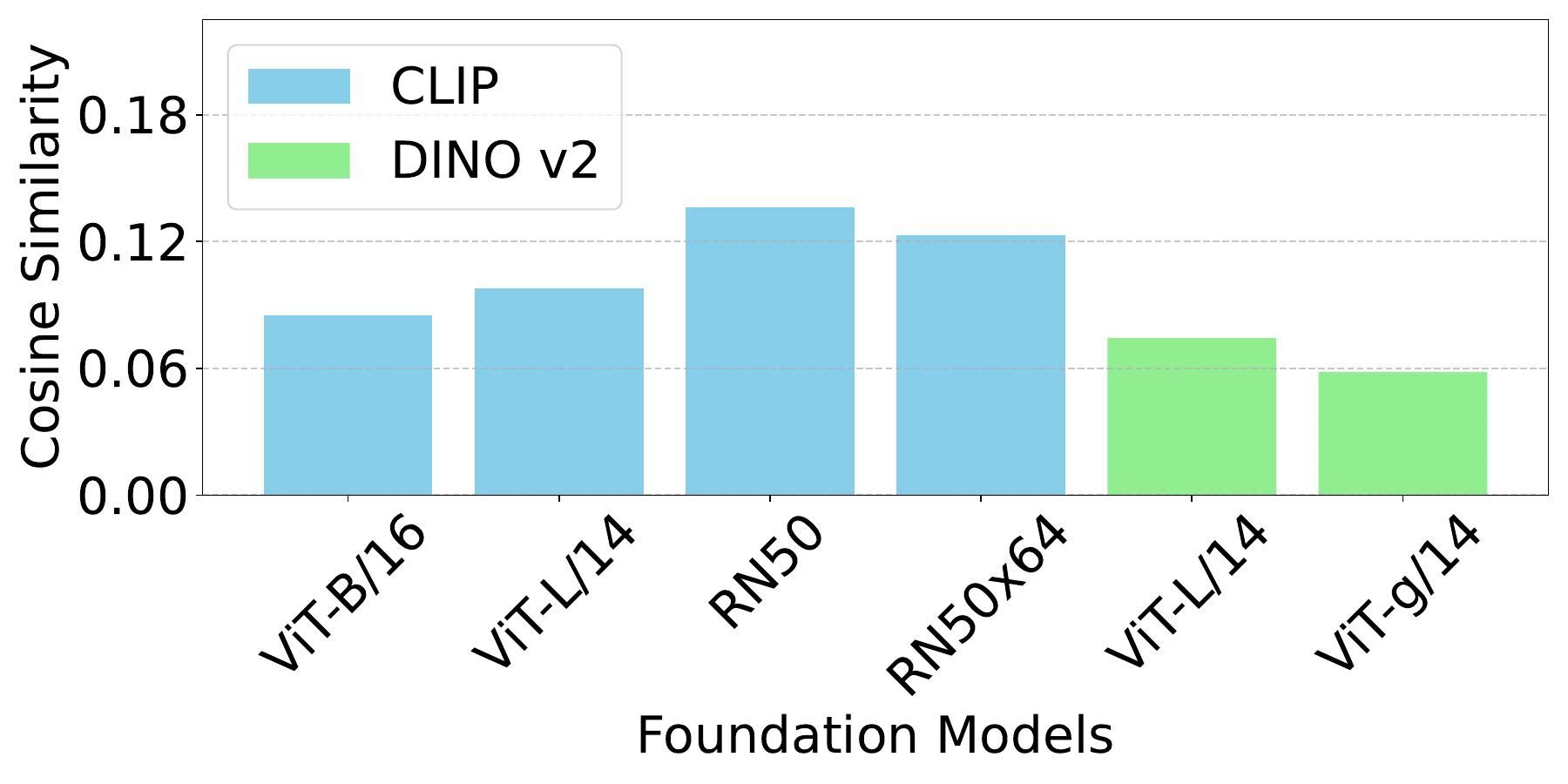}}
    \caption{Average cosine similarity of images in NYU-Depth V2 for different foundation models and  perturbation functions.}
    \label{fig:cosine_vision_nyu}
\end{figure*}

\begin{figure*}[t!]
    \centering
    \subfloat[ImageNet: Zero-shot classification]{\includegraphics[width=0.24\textwidth]{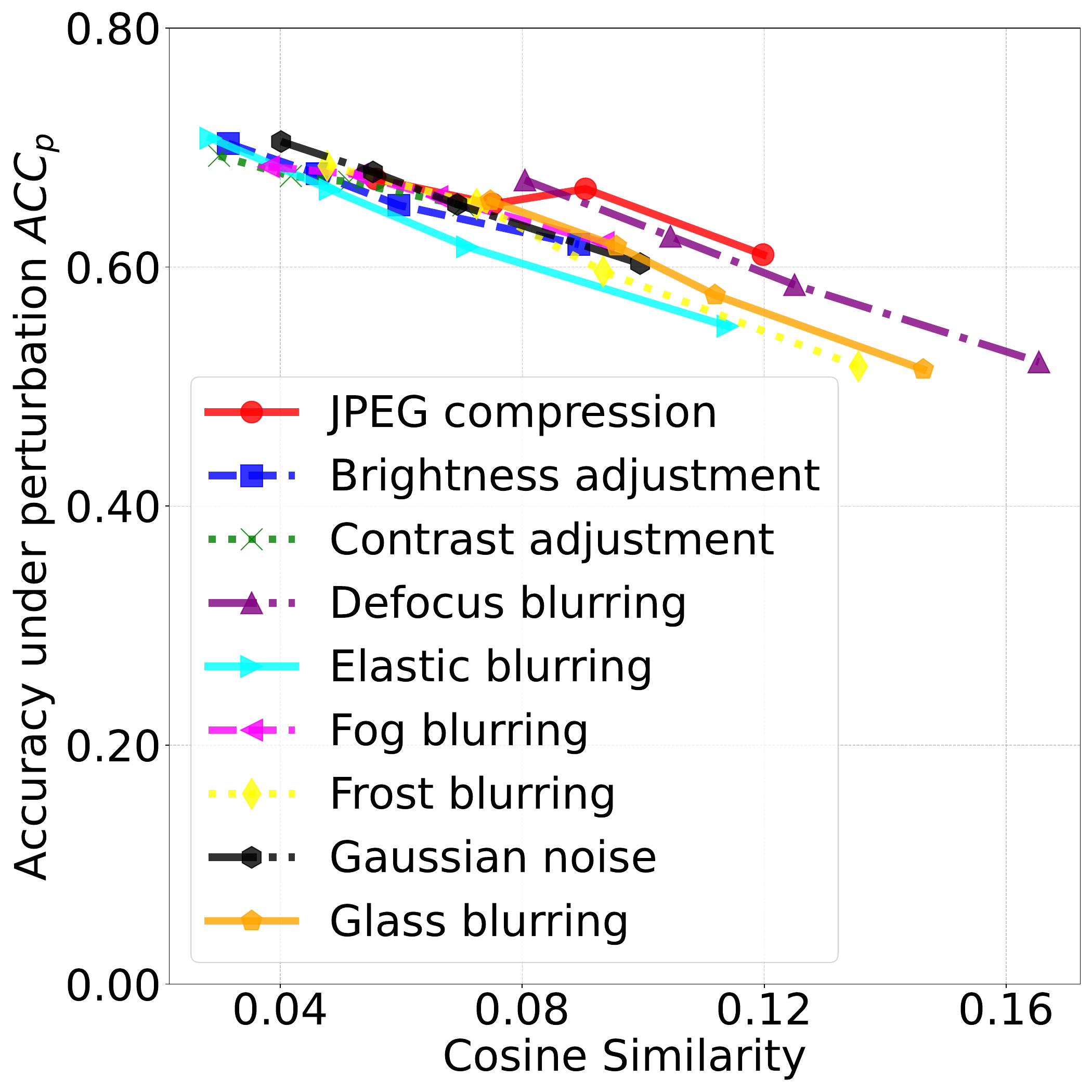}\label{fig:vision_downstream_zero_shot_imagenet_cosine}}
    \subfloat[ImageNet: Linear-probe classification]{\includegraphics[width=0.24\textwidth]{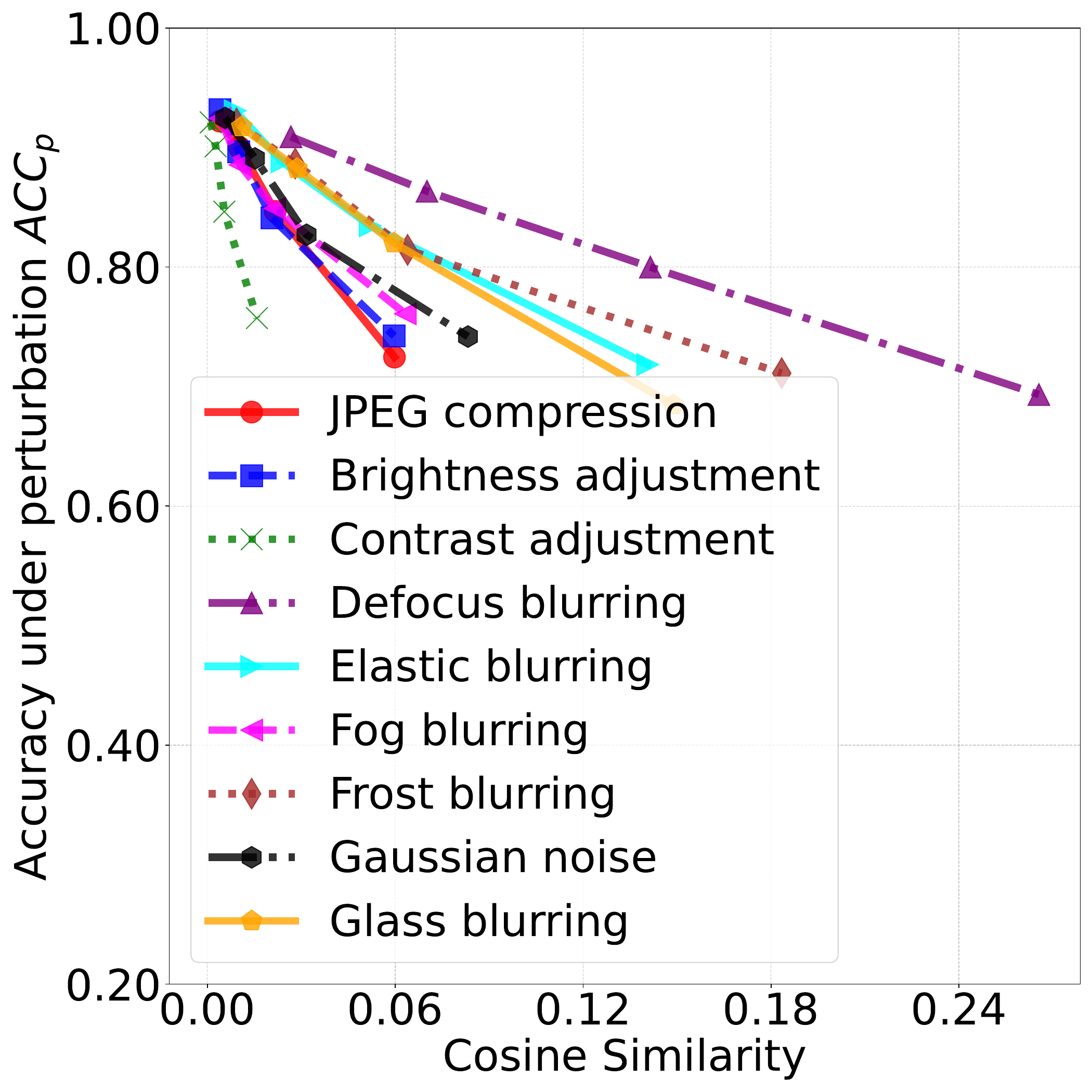}\label{fig:vision_downstream_linear_probe_imagenet_cosine}}
    \subfloat[Food101: Zero-shot classification]{\includegraphics[width=0.24\textwidth]{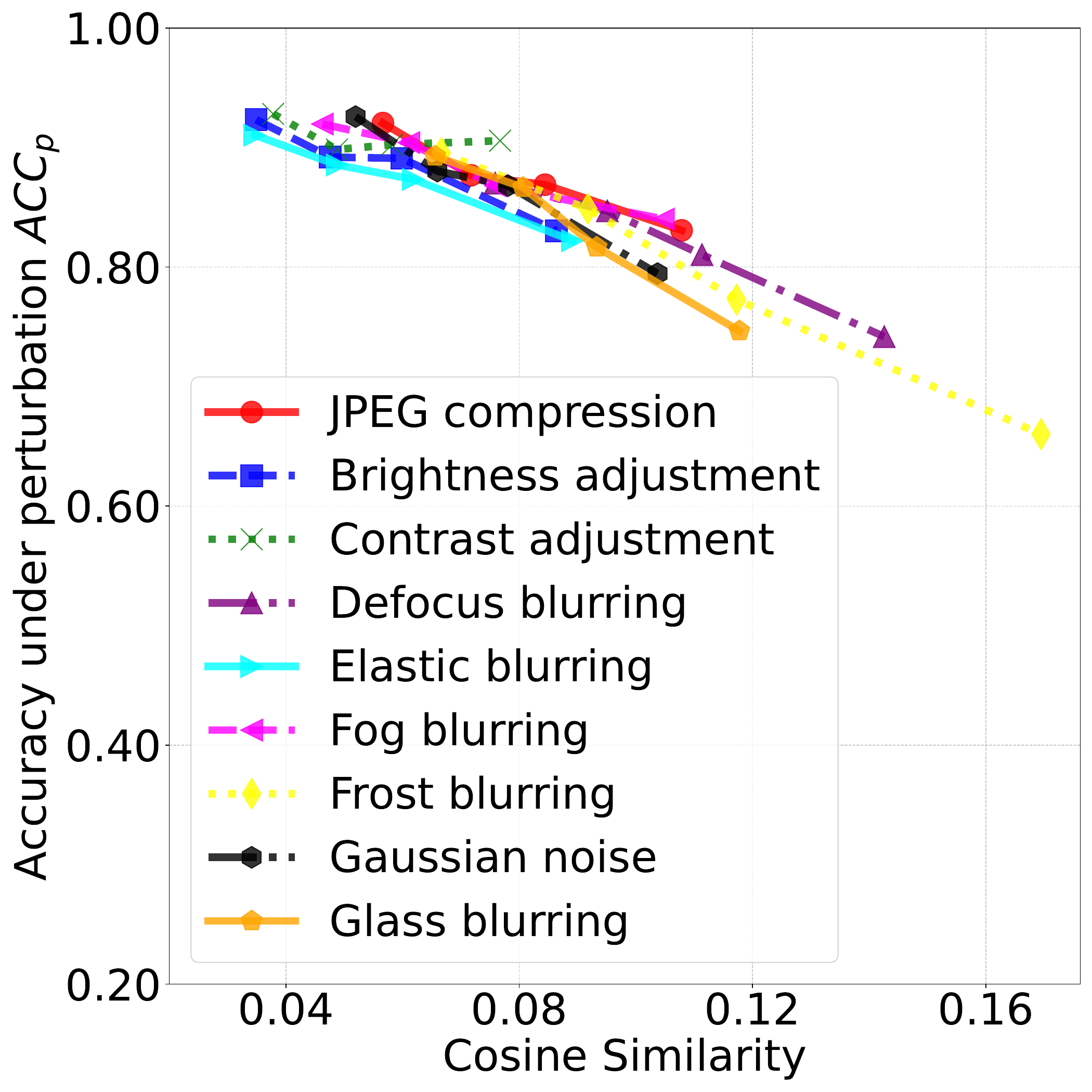}\label{fig:vision_downstream_zero_shot_imagenet_cosine}}
    \subfloat[Food101: Linear-probe classification]{\includegraphics[width=0.24\textwidth]{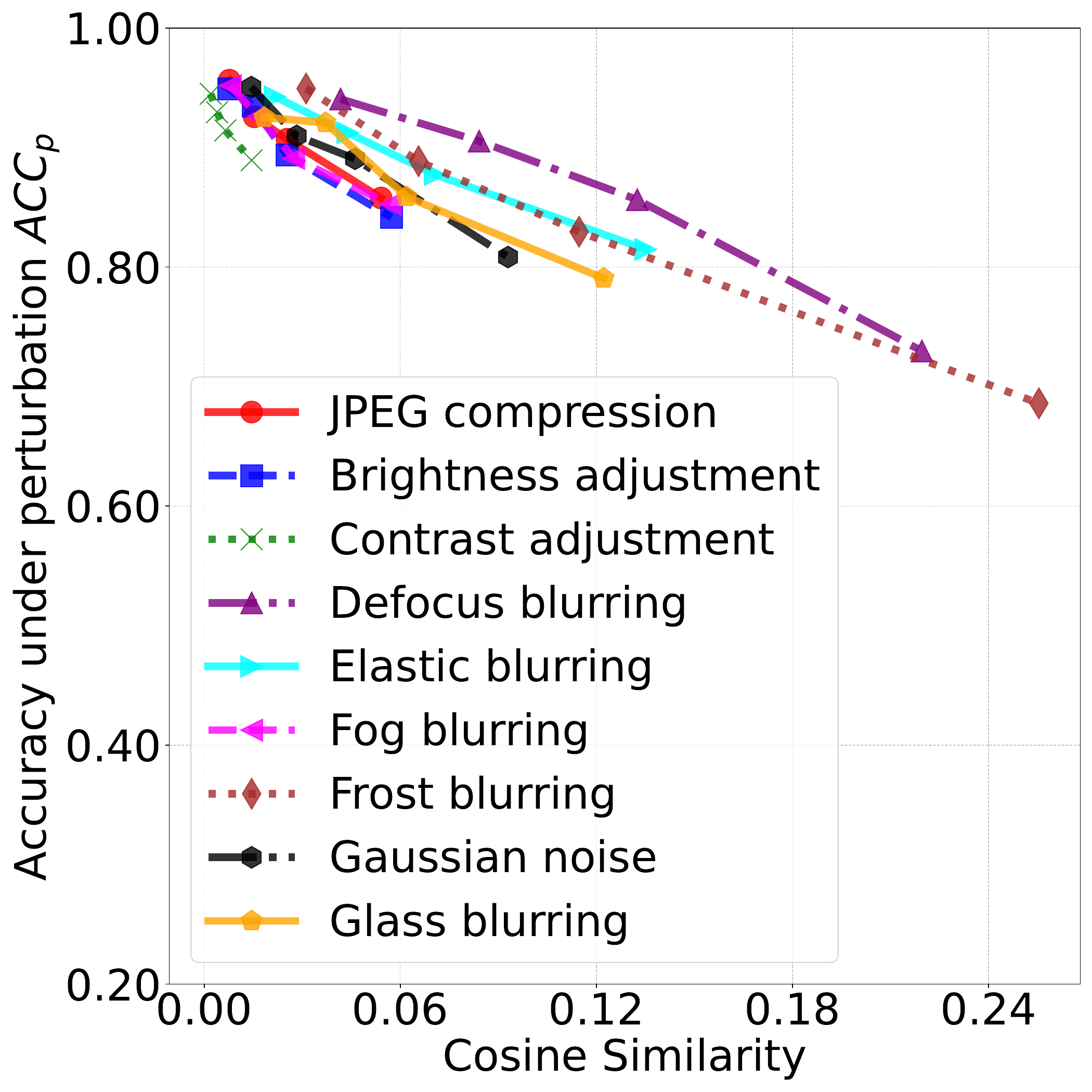}\label{fig:vision_downstream_linear_probe_imagenet_cosine}}
    \caption{Accuracy under perturbation $ACC_p$ vs. cosine similarity of ImageNet and Food101 testing images for zero-shot classification and linear-probe classification when different perturbation functions are used.   Zero-shot classification is based on the CLIP ViT-L/14 foundation model and linear-probe classification is based on the DINO v2 ViT-g/14 foundation model.} \label{fig:vision_downstream_imagenet_food101_cosine_similarity}
\end{figure*}

\begin{table}[!h]
\fontsize{8}{11}\selectfont 
\caption{Impact of $\lambda$ on average~\alg{} and $ACC$ of zero-shot classification for two datasets before/after enhancement, where the foundation model is CLIP
ViT-L/14.}
\label{tab:enhancement_impact_lambda}
\centering
\begin{tabular}{|c|c|c|}
\hline
$\lambda$ & \alg{} & $ACC$ (\%) \\\hline
0.0 & 0.01 & 0.16 \\ \hline
0.5 & 0.05 & 66.61 \\ \hline
1.0 & 0.06 & 67.99 \\ \hline
5.0 & 0.07 & 68.31 \\ \hline
\end{tabular}
\end{table}

%% file: main.bib
@String(CVPR= {IEEE Conf. Comput. Vis. Pattern Recog.})

@String(ICCV= {Int. Conf. Comput. Vis.})

@String(ECCV= {Eur. Conf. Comput. Vis.})

@String(ICLR = {Int. Conf. Learn. Represent.})

@String(AAAI = {AAAI})

@String(CVPR  = {CVPR})

@String(ICCV  = {ICCV})

@String(ECCV  = {ECCV})

@String(ICLR  = {ICLR})

@article{oquab2023dinov2,
  title={Dinov2: Learning robust visual features without supervision},
  author={Oquab, Maxime and Darcet, Timoth{\'e}e and Moutakanni, Th{\'e}o and Vo, Huy and Szafraniec, Marc and Khalidov, Vasil and Fernandez, Pierre and Haziza, Daniel and Massa, Francisco and El-Nouby, Alaaeldin and others},
  journal={arXiv},
  year={2023}
}

@inproceedings{radford2021learning,
  title={Learning transferable visual models from natural language supervision},
  author={Radford, Alec and Kim, Jong Wook and Hallacy, Chris and Ramesh, Aditya and Goh, Gabriel and Agarwal, Sandhini and Sastry, Girish and Askell, Amanda and Mishkin, Pamela and Clark, Jack and others},
  booktitle={ICML},
  year={2021},
}

@inproceedings{szegedy2013intriguing,
  title={Intriguing properties of neural networks},
  author={Szegedy, Christian and Zaremba, Wojciech and Sutskever, Ilya and Bruna, Joan and Erhan, Dumitru and Goodfellow, Ian and Fergus, Rob},
  booktitle={ICLR},
  year={2014}
}

@inproceedings{carlini2017towards,
  title={Towards evaluating the robustness of neural networks},
  author={Carlini, Nicholas and Wagner, David},
  booktitle={S\&P},
  year={2017},
}

@inproceedings{deng2009imagenet,
  title={Imagenet: A large-scale hierarchical image database},
  author={Deng, Jia and Dong, Wei and Socher, Richard and Li, Li-Jia and Li, Kai and Fei-Fei, Li},
  booktitle={CVPR},
  year={2009},
}

@inproceedings{bossard14,
  title = {Food-101 -- Mining Discriminative Components with Random Forests},
  author = {Bossard, Lukas and Guillaumin, Matthieu and Van Gool, Luc},
  booktitle = {ECCV},
  year = {2014}
}

@inproceedings{welzl2005smallest,
  title={Smallest enclosing disks (balls and ellipsoids)},
  author={Welzl, Emo},
  booktitle={New Results and New Trends in Computer Science},
  year={2005},
  organization={Springer}
}

@inproceedings{chen2020simple,
  title={A simple framework for contrastive learning of visual representations},
  author={Chen, Ting and Kornblith, Simon and Norouzi, Mohammad and Hinton, Geoffrey},
  booktitle={ICML},
  year={2020}
}

@inproceedings{hendrycks2019benchmarking,
  title={Benchmarking neural network robustness to common corruptions and perturbations},
  author={Hendrycks, Dan and Dietterich, Thomas},
  booktitle={ICLR},
  year={2019}
}

@inproceedings{li2022blip,
  title={Blip: Bootstrapping language-image pre-training for unified vision-language understanding and generation},
  author={Li, Junnan and Li, Dongxu and Xiong, Caiming and Hoi, Steven},
  booktitle={ICML},
  year={2022},
}

@inproceedings{he2020momentum,
  title={Momentum contrast for unsupervised visual representation learning},
  author={He, Kaiming and Fan, Haoqi and Wu, Yuxin and Xie, Saining and Girshick, Ross},
  booktitle={CVPR},
  year={2020}
}

@inproceedings{jia2022badencoder,
  title={Badencoder: Backdoor attacks to pre-trained encoders in self-supervised learning},
  author={Jia, Jinyuan and Liu, Yupei and Gong, Neil Zhenqiang},
  booktitle={S\&P},
  year={2022},
}

@article{zhu2023promptbench,
  title={PromptBench: Towards Evaluating the Robustness of Large Language Models on Adversarial Prompts},
  author={Zhu, Kaijie and Wang, Jindong and Zhou, Jiaheng and Wang, Zichen and Chen, Hao and Wang, Yidong and Yang, Linyi and Ye, Wei and Gong, Neil Zhenqiang and Zhang, Yue and others},
  journal={arXiv},
  year={2023}
}

@inproceedings{qu2023reaas,
  title={REaaS: Enabling Adversarially Robust Downstream Classifiers via Robust Encoder as a Service},
  author={Qu, Wenjie and Jia, Jinyuan and Gong, Neil Zhenqiang},
  booktitle={NDSS},
  year={2023}
}

@inproceedings{liu2022poisonedencoder,
  title={{PoisonedEncoder}: Poisoning the Unlabeled Pre-training Data in Contrastive Learning},
  author={Liu, Hongbin and Jia, Jinyuan and Gong, Neil Zhenqiang},
  booktitle={USENIX Security Symposium},
  year={2022}
}

@inproceedings{li2023embarrassingly,
  title={An Embarrassingly Simple Backdoor Attack on Self-supervised Learning},
  author={Li, Changjiang and Pang, Ren and Xi, Zhaohan and Du, Tianyu and Ji, Shouling and Yao, Yuan and Wang, Ting},
  booktitle={ICCV},
  year={2023}
}

@inproceedings{fan2021does,
  title={When does contrastive learning preserve adversarial robustness from pretraining to finetuning?},
  author={Fan, Lijie and Liu, Sijia and Chen, Pin-Yu and Zhang, Gaoyuan and Gan, Chuang},
  journal={NeurIPS},
  year={2021}
}

@article{jiang2020robust,
  title={Robust pre-training by adversarial contrastive learning},
  author={Jiang, Ziyu and Chen, Tianlong and Chen, Ting and Wang, Zhangyang},
  journal={NeurIPS},
  year={2020}
}

@inproceedings{saha2022backdoor,
  title={Backdoor attacks on self-supervised learning},
  author={Saha, Aniruddha and Tejankar, Ajinkya and Koohpayegani, Soroush Abbasi and Pirsiavash, Hamed},
  booktitle={CVPR},
  year={2022}
}

@inproceedings{Silberman:ECCV12,
  author    = {Nathan Silberman, Derek Hoiem, Pushmeet Kohli and Rob Fergus},
  title     = {Indoor Segmentation and Support Inference from RGBD Images},
  booktitle = {ECCV},
  year      = {2012}
}

@inproceedings{bai2021transformers,
  title={Are transformers more robust than cnns?},
  author={Bai, Yutong and Mei, Jieru and Yuille, Alan L and Xie, Cihang},
  booktitle={NeurIPS},
  year={2021}
}

@inproceedings{bhojanapalli2021understanding,
  title={Understanding robustness of transformers for image classification},
  author={Bhojanapalli, Srinadh and Chakrabarti, Ayan and Glasner, Daniel and Li, Daliang and Unterthiner, Thomas and Veit, Andreas},
  booktitle={ICCV},
  year={2021}
}

@inproceedings{wang2022can,
  title={Can cnns be more robust than transformers?},
  author={Wang, Zeyu and Bai, Yutong and Zhou, Yuyin and Xie, Cihang},
  booktitle={ICLR},
  year={2023}
}

@inproceedings{paul2022vision,
  title={Vision transformers are robust learners},
  author={Paul, Sayak and Chen, Pin-Yu},
  booktitle={AAAI},
  year={2022}
}

@inproceedings{hendrycks2021many,
  title={The many faces of robustness: A critical analysis of out-of-distribution generalization},
  author={Hendrycks, Dan and Basart, Steven and Mu, Norman and Kadavath, Saurav and Wang, Frank and Dorundo, Evan and Desai, Rahul and Zhu, Tyler and Parajuli, Samyak and Guo, Mike and others},
  booktitle={ICCV},
  year={2021}
}

@inproceedings{jiang2023evading,
  title={Evading watermark based detection of ai-generated content},
  author={Jiang, Zhengyuan and Zhang, Jinghuai and Gong, Neil Zhenqiang},
  booktitle={CCS},
  year={2023}
}
